\RequirePackage{latexml}
\iflatexml
    \documentclass{article}
    \author{%
        Théo Delemazure \\
        CNRS, LAMSADE, Université Paris Dauphine - PSL
        \and
        Chris Dong \\
        Technical University of Munich
        \and
        Dominik Peters \\
        CNRS, LAMSADE, Université Paris Dauphine - PSL
        \and
        Magdaléna Tydrichová \\
        CentraleSupélec, Paris Saclay University
    }

    \lxDefMath{\axis}{◃}[role=ADDOP]
    \newcommand{\cev}[1]{\overline{#1}}
    
    \newcommand{\wrapcaption}[1]{\caption{\parbox{5cm}{#1}}}
\else
    \documentclass[11pt]{scrartcl}
    \usepackage[a4paper, total={16cm, 24cm}]{geometry}
    \renewcommand\tableofcontents{\listoftoc*{toc}} %
    \KOMAoptions{toc=bibnumbered}
    \setcapindent{0pt}
    \setkomafont{captionlabel}{\sffamily\bfseries}

    \usepackage{authblk}

    \author[1]{Théo Delemazure}
    \author[2]{Chris Dong}
    \author[1]{\authorcr Dominik Peters}
    \author[3]{Magdaléna Tydrichová} %
    \affil[1]{CNRS, LAMSADE, Universit\'e Paris Dauphine - PSL}
    \affil[2]{Technical University of Munich}
    \affil[3]{CentraleSupélec, Paris Saclay University}

    \usepackage{multicol}
    \usepackage{tocloft}

    \newcommand{\axis}{\triangleleft}
    \newcommand{\cev}[1]{\reflectbox{\ensuremath{\vec{\reflectbox{\ensuremath{#1}}}}}}
    
    \newcommand{\wrapcaption}[1]{\caption{#1}}
\fi

\usepackage{natbib}

\usepackage[pagebackref]{hyperref}
\hypersetup{
	pdfencoding=auto, 
	psdextra,
	colorlinks=true,
	citecolor=green!50!black,
	linkcolor=red!60!black,
	urlcolor=purple!70!black
}

\usepackage{url}
\usepackage{microtype}
\usepackage[utf8]{inputenc}
\usepackage[small]{caption}
\usepackage{graphicx}
\usepackage{amsmath}
\usepackage{amsthm}
\usepackage{amssymb}
\usepackage{mathtools}
\usepackage{booktabs}
\usepackage{enumitem}

\usepackage{wrapfig}
\iflatexml\else
    \usepackage{etoolbox}
    \makeatletter
    \patchcmd\WF@putfigmaybe{\lower\intextsep}{}{}{\fail}%
    \AddToHook{env/wrapfigure/begin}{\setlength{\intextsep}{0pt}}
    \makeatother
\fi

\usepackage[nameinlink]{cleveref}
\iflatexml\else
    \usepackage{crossreftools}
\fi
\usepackage{subcaption}
\usepackage{stmaryrd}
\usepackage{thm-restate}

\usepackage{tikz}
\renewcommand{\epsilon}{\varepsilon}
\renewcommand*{\le}{\leqslant}
\renewcommand*{\leq}{\leqslant}
\renewcommand*{\ge}{\geqslant}
\renewcommand*{\geq}{\geqslant}
\usepackage{pgfplots}
\pgfplotsset{compat=1.15,
legend image code/.code={
\draw[mark repeat=2,mark phase=2]
plot coordinates {
(0cm,0cm)
(0.15cm,0cm)        %
(0.3cm,0cm)         %
};%
}}

\newtheoremstyle{sfthm}%
	{\topsep}%
	{\topsep}%
	{\itshape}%
	{}%
	{\sffamily\bfseries}%
	{}%
	{.5em}%
	{}%
\theoremstyle{sfthm}

\newtheorem{lemma}{Lemma}
\newtheorem{example}{Example}
\newtheorem{proposition}{Proposition}
\newtheorem{theorem}{Theorem}

\DeclareMathOperator*{\argmax}{arg\,max}
\DeclareMathOperator*{\argmin}{arg\,min}
\DeclareMathOperator{\con}{con}
\DeclareMathOperator{\cost}{\mathsf{cost}}

\iflatexml
    \newcommand{\drule}[2]{\textbf{#1} #2}
    \newcommand{\property}[2]{\textbf{#1} #2}
\else
    \usepackage{tcolorbox}
    \tcbuselibrary{breakable,theorems,skins}
    \newtcolorbox{rulebox}{
        blanker,
        interior engine=standard,
        left=3.5mm, right=7pt,
        top=6pt, bottom=6pt,
        before skip=8pt,after skip=8pt,
        colback=purple!6!white,
        borderline west={1mm}{0mm}{purple!30!white}}
    \newtcolorbox{propertybox}{
        blanker,
        interior engine=standard,
        left=3.5mm, right=7pt,
        top=6pt, bottom=6pt,
        before skip=8pt,after skip=8pt,
        colback=blue!6!white,
        borderline west={1mm}{0mm}{blue!30!white}}
    \newcommand{\drule}[2]{\begin{rulebox}\textsf{\textbf{#1}}\hspace{1pt} #2\end{rulebox}}
    \newcommand{\property}[2]{\begin{propertybox}\textsf{\textbf{#1}}\hspace{1pt} #2\end{propertybox}}
\fi 

\newcommand{\VD}{\textup{VD}}
\newcommand{\BC}{\textup{BC}}
\newcommand{\MF}{\textup{MF}}
\newcommand{\MS}{\textup{MS}}
\newcommand{\FT}{\textup{FT}}

\definecolor{nored}{RGB}{163, 109, 109} %
\newcommand{\xmark}{%
	\tikz[scale=0.23,draw=nored] {
		\draw[line width=0.7,line cap=round] (0,0) to [bend left=6] (1,1);
		\draw[line width=0.7,line cap=round] (0.2,0.95) to [bend right=3] (0.8,0.05);
}}

\definecolor{hole}{RGB}{204, 104, 104} %

\newcommand{\cmark}[1][green!50!black]{%
	\tikz[scale=0.23,draw=#1] {
		\draw[line width=0.7,line cap=round] (0.25,0) to [bend left=10] (1,1);
		\draw[line width=0.8,line cap=round] (0,0.35) to [bend right=1] (0.23,0);
}}

\usepackage{ifthen}
\newcommand{\rowheight}{0.5cm}
\newcommand{\colwidth}{0.6cm}
\newcommand{\axisheader}[2]{%
    \foreach \x/\headercolor [count=\xi] in {#2} {%
        \ifthenelse{\equal{\x}{\headercolor}}{
            \def\headercolor{black}
        }{}
        \node [anchor=base,text=\headercolor] at (\xi * \colwidth, -#1 * \rowheight) {$\x$};
    }
}

\newcommand{\costcolumn}[2]{
    \node[anchor=base] at (#1 * \colwidth + 0.5*\colwidth, -0 * \rowheight) {$#2$};
}

\definecolor{intervalbg}{HTML}{86B6F6}
\definecolor{intervalcheckmark}{HTML}{EEF5FF}

\iflatexml
\newcommand{\interval}[3]{%
    \draw[line width=11pt, intervalbg, line cap=round] (#2 * \colwidth - 0.3, -#1 * \rowheight) -- (#3 * \colwidth + 0.3, -#1 * \rowheight);
    \foreach \x in {#2,...,#3} {
        \node at (\x * \colwidth + 0.1, -#1 * \rowheight) {\cmark[intervalcheckmark]};
    }
}
\else
\newcommand{\interval}[3]{%
	\draw[line width=11pt, intervalbg, line cap=round] (#2 * \colwidth - 0.3, -#1 * \rowheight) -- (#3 * \colwidth + 0.3, -#1 * \rowheight);
	\foreach \x in {#2,...,#3} {
		\node at (\x * \colwidth, -#1 * \rowheight) {\cmark[intervalcheckmark]};
	}
}
\fi

\definecolor{minuscolor}{HTML}{084ca3}
\newcommand{\flipinterval}[3]{
	\draw[line width=11pt, intervalbg, line cap=round] (#2 * \colwidth - 0.3, -#1 * \rowheight) -- (#3 * \colwidth + 0.3, -#1 * \rowheight);
	\foreach \x in {#2,...,#3} {
		\node at (\x * \colwidth, -#1 * \rowheight) {\cmark[minuscolor]};
		\draw[minuscolor, line width=0.65] (\x * \colwidth + 0.8, -#1 * \rowheight - 2) -- +(0.1, 0);
	}
}

\newcommand{\redinterval}[3]{%
	\draw[line width=12pt, minuscolor, line cap=round] (#2 * \colwidth - 0.3, -#1 * \rowheight) -- (#3 * \colwidth + 0.3, -#1 * \rowheight);
    \draw[line width=11pt, intervalbg, line cap=round] (#2 * \colwidth - 0.3, -#1 * \rowheight) -- (#3 * \colwidth + 0.3, -#1 * \rowheight);
    \foreach \x in {#2,...,#3} {
        \node at (\x * \colwidth, -#1 * \rowheight) {\cmark[intervalcheckmark]};
    }
}
\newcommand{\interfering}[2]{%
    \fill [red!70!black!70!white] (#2 * \colwidth, -#1 * \rowheight) circle [radius = 0.08];
}

\newcommand{\swapslot}[2]{%
	\fill [black!30, yshift=-1pt] (#2 * \colwidth, -#1 * \rowheight) circle [radius = 0.05];
}

\newcommand{\ballotcost}[3]{%
    \node at (#2 * \colwidth + 0.5*\colwidth , -#1 * \rowheight) {$#3$};
}

\newcommand{\ft}[3]{
    \node [red!70!black!70!white,font=\footnotesize] at (#2 * \colwidth, -#1 * \rowheight) {#3};
}
\iflatexml 
	\newcommand{\multiplicity}[2]{%
	    \node at (-0.3, -#1 * \rowheight) {$#2 \times$};
	}
\else
	\newcommand{\multiplicity}[2]{%
		\node [anchor=east] at (0.3, -#1 * \rowheight) {#2 $\times$};
	}
\fi

\definecolor{pluscolor}{HTML}{b02727}
\newcommand{\missingapproval}[2]{
	\begin{scope}[xshift=#2 * \colwidth, yshift=-#1 * \rowheight]
		\node[scale=0.9, transform shape] at (0,0) {\cmark[pluscolor]};
		\draw[pluscolor, line width=0.5] (0.02,-0.06) -- +(0.12, 0);
		\draw[pluscolor, line width=0.5] (0.08,-0.12) -- +(0.0, 0.12);
	\end{scope}
}

\definecolor{swaparrowcolor}{HTML}{176B87}
\newcommand{\swaparrow}[3]{%
	\draw[<->,,yshift=-1pt,xshift=0pt,draw=swaparrowcolor] (#2*\colwidth + 4pt, -#1*\rowheight) to (#3*\colwidth -5pt, -#1*\rowheight);
}
\newcommand{\yes}{\cmark}
\newcommand{\no}{\xmark}

\newcommand{\axisweak}{\trianglelefteqslant}

\definecolor{LO}{HTML}{AA0000}
\definecolor{UPR}{HTML}{9ad6af}
\definecolor{SP}{HTML}{e0dcad}
\definecolor{DLF}{HTML}{0082C4}
\definecolor{LR}{HTML}{0066CC}
\definecolor{PS}{HTML}{FF8080}
\definecolor{R}{HTML}{26c4ec}
\definecolor{FN}{HTML}{0D378A}
\definecolor{RN}{HTML}{0D378A}
\definecolor{EM}{HTML}{FFD600}
\definecolor{LFI}{HTML}{cc2443}
\definecolor{NPA}{HTML}{bb0000}
\definecolor{LREM}{HTML}{FFD600}
\definecolor{REC}{HTML}{404040}
\definecolor{EELV}{HTML}{00c000}
\definecolor{PCF}{HTML}{DD0000}
\definecolor{textLO}{HTML}{FFFFFF}
\definecolor{textUPR}{HTML}{000000}
\definecolor{textSP}{HTML}{000000}
\definecolor{textDLF}{HTML}{FFFFFF}
\definecolor{textLR}{HTML}{FFFFFF}
\definecolor{textPS}{HTML}{FFFFFF}
\definecolor{textR}{HTML}{000000}
\definecolor{textFN}{HTML}{FFFFFF}
\definecolor{textRN}{HTML}{FFFFFF}
\definecolor{textEM}{HTML}{000000}
\definecolor{textLREM}{HTML}{000000}
\definecolor{textLFI}{HTML}{FFFFFF}
\definecolor{textNPA}{HTML}{FFFFFF}
\definecolor{textREC}{HTML}{FFFFFF}
\definecolor{textEELV}{HTML}{FFFFFF}
\definecolor{textPCF}{HTML}{FFFFFF}
\newcommand{\partysize}[1]{%
	\ifthenelse{\equal{#1}{LO} \OR \equal{#1}{NPA} \OR \equal{#1}{R}\OR \equal{#1}{UPR}\OR \equal{#1}{SP}}%
	{\scriptsize #1}%
	{\footnotesize #1}%
}
\newcommand{\party}[1]{\tikz[scale=1, transform shape,baseline=-3pt]{
	\node[shape=rectangle,fill=#1,inner sep=1pt,text=text#1,minimum width=15pt,minimum height=9pt] {\partysize{#1}};
}}

\title{Comparing Ways of Obtaining Candidate Orderings from Approval Ballots%
}

\date{\vspace{-1cm}}

\usepackage{multicol}

\begin{document}

\maketitle

\begin{abstract}
    \iflatexml\else
	\begin{center}
		\textbf{\textsf{Abstract}} \smallskip
	\end{center}\fi
	To understand and summarize approval preferences and other binary evaluation data, it is useful to order the items on an \emph{axis} which explains the data. In a political election using approval voting, this could be an ideological left-right axis such that each voter approves adjacent candidates, an analogue of single-peakedness. 
	In a perfect axis, every approval set would be an interval, 
	which is usually not possible, and so we need to choose an axis that gets closest to this ideal.
	The literature has developed algorithms for optimizing several objective functions (e.g., minimize the number of added approvals needed to get a perfect axis), but provides little help with choosing among different objectives.
	In this paper, we take a social choice approach and compare 5 different axis selection rules axiomatically, by studying the properties they satisfy.
	We establish some impossibility theorems, and characterize (within the class of scoring rules) the rule that chooses the axes that maximize the number of votes that form intervals, using the axioms of ballot monotonicity and resistance to cloning.
	Finally, we study the behavior of the rules on data from French election surveys, on the votes of justices of the US Supreme Court, and on synthetic data.
\end{abstract}

\iflatexml\else
\vspace{20pt}

\hrule

\vspace{5pt}
{
	\setlength\columnsep{35pt} %
	\setcounter{tocdepth}{1}
	\renewcommand\contentsname{\vspace{-20pt}}
	\begin{multicols}{2}
		{\small
			\tableofcontents}
	\end{multicols}
}

\vspace{11pt}
\hrule

\newpage
\fi

\section{Introduction}\label{sec:introduction}
This paper is about analyzing and understanding binary evaluation data.
Such data could come from many sources, such as user reviews featuring a thumbs up / thumbs down evaluation, or datasets of items with binary information about their features.
Another source of such data is \emph{approval voting}, where each evaluator is a \emph{voter} who \emph{approves} the candidates that have been assigned an evaluation of 1.
Since we will use techniques from computational social choice in our analysis, for simplicity we will generally use voting terminology to refer to our setting.
Our aim is to obtain an ordering of the candidates (an \emph{axis}) which is supposed to summarize the data.
Specifically, we interpret an axis to ``perfectly depict'' the data if every voter approves an \emph{interval} of the axis. This is an approval version of single-peaked preferences. For most datasets, such axes will not exist, so we study rules that, given an approval profile, find the axes that best approximate the interval structure and that thereby provide a good (ordinal) one-dimensional embedding of the profile.
Such rules have many applications for understanding and visualizing data, as well as direct use-cases where the axis itself plays a key role:
\begin{itemize}[leftmargin=9pt]
	\item \emph{Ordering political candidates and parties.} In politics, if voters are asked to approve candidates, an axis could correspond to an ideological ordering of the candidates from left-wing to right-wing. For example, in France, election polls are typically presented with candidates ordered by ideology, but the major pollsters use many different axes (see \Cref{sec:experiments:france}), which they apparently construct ad hoc. Our rules will find an axis in a principled way.
	\item \emph{Ordering members of parliament.} Once elected, we can interpret each bill as a ``voter'' who approves those members who supported it. An axis rule would then provide an ordering of members of parliament by ideology.
	\item \emph{Archaeological seriation.} A well-established approach in archaeology for ordering artefacts by their age is to let features that were temporarily ``in fashion'' (e.g., drawing styles) approve artefacts \citep{petrie1899sequences,baxter2003statistics}. In the true ordering by age, each feature is likely to induce an interval. Thus, a good axis rule will tend to produce an ordering of artefacts by their age with few errors.
	\item \emph{Scheduling}. A conference organizer could ask attendees about which talks they wish to see and then use our rules to arrange the talks so attendees can join for consecutive talks. A different way of applying our rules (without the need to ask for attendee preferences) is for key terms to ``approve'' the papers that mention the term, leading to a thematically coherent ordering of the talks.
\end{itemize}

\noindent
Algorithmically, the task of finding an axis optimizing a particular objective function is well studied.  To check whether a perfect axis exists (i.e., one where every voter approves an interval), one needs to check whether the 0/1 approval matrix has the \emph{consecutive ones property} (C1P), which can be done in linear time \citep{booth_lueker_C1P}. 
However, in all the applications discussed above, the 0/1 matrices are likely to only approximately satisfy C1P.
The problem of finding an axis that makes as many votes as possible into an interval is NP-complete and already appears in the book of \citet[Problem SR14]{garey1979computers} together with several similar problems about recognizing almost-C1P matrices like minimizing the number of approvals to add to satisfy C1P (Problem SR16). However, this complexity theoretic work does not tell us which of these objective functions ``work best''.

We provide a framework for answering the question of which is the ``best'' objective function via the axiomatic method used in social choice. We interpret different objective functions as \emph{rules} that take an approval profile as input and decide on an axis. We will compare these rules by identifying properties that they satisfy or fail. Given a context where some properties seem particularly desirable, this will help with selecting a good objective function.

The protagonists of our paper are the following five rules, with more precise definitions provided later:

\begin{itemize}[leftmargin=9pt, itemsep=0pt]
    \item \emph{Voter Deletion.} Minimize the number of votes that are not intervals of the axis.
    \item \emph{Minimum Flips.} Minimize the number of approvals that need to be added or removed from ballots to make all votes intervals of the axis.
    \item \emph{Ballot Completion.} Minimize the number of approvals that need to be added to ballots to make all votes intervals of the axis.
    \item \emph{Minimum Swaps.} Minimize the average number of swaps within the axis that are needed to turn votes into intervals of the axis.
    \item \emph{Forbidden Triples.} Minimize the total size of holes in a vote, weighted by how many approved candidates they separate. 
\end{itemize}

On a high level, we find that Voter Deletion and Ballot Completion satisfy a desirable monotonicity property (saying that the chosen axis should not change if some voters change their ballots to better align with that axis), while the last two rules use more information contained in the profile. We do not identify any positive features of Minimum Flips.

Besides introducing the rules and the axioms, we also prove an impossibility result saying that no scoring rule (which are rules that optimize a voter-additive objective function) can simultaneously satisfy two versions of the ``clones'' principle that a rule should behave reasonably in the presence of identical candidates: \emph{clone proximity} which says that such candidates must be placed next to each other on the axis and \emph{clone resistance} which says that deleting some of the identical candidates should not affect the relative placement of other candidates. We also establish a characterization result that the Voter Deletion rule is the unique scoring rule that satisfies clone resistance as well as ballot monotonicity.

We conclude the paper by applying our rules to different datasets, including French election surveys (ordering candidates from left to right), votes of the justices of the US Supreme Court (ordering justices from conservative to progressive), and synthetic datasets. The simulations show how our rules differ, which perform best, and how they compare to rules that are based on taking rankings rather than approvals as input.

\section{Related Work}  \label{sec:related-work}

The work of \citet{escoffier2021nearlysp}, extended in the thesis of \citet[Sec.~4.4]{tydrichova2023structural}, is closest to ours, as it compares different methods for finding axes that make a profile of \emph{rankings} of the candidates \emph{nearly single-peaked}. Single-peaked ranking preferences \citep{black1948rationale} are frequently studied in social choice because they can avoid impossibility theorems and computational hardness  \citep{elkind2017structured,elkind2022preference}.
\citet{escoffier2021nearlysp} focus on computational complexity, but also consider axiomatic properties satisfied by different objective functions. However, they do not provide axiomatic characterization or impossibility results, and our experiments suggest that the approval approach may lead to better axes than the ranking approach. Nearly single-peaked preferences are well-studied algorithmically, both in terms of their recognition \citep{bredereck2016there,erdelyi2017nearlysp,elkind2014detecting} and their impact on the winner determination problem of computationally hard voting rules \citep{misra2017complexity,chen2023efficient}.

For approval ballots, structured preferences are studied by \citet{elkind2015structure}, who say that a profile satisfies \emph{Candidate Interval} (CI) if there is a perfect axis for it [see also \citealp{faliszewski2011shield,terzopoulou2021restricted}].
\citet{dietrich2015} discuss a similar concept in judgement aggregation.
The study of the algorithmic problem of recognizing profiles that are \emph{nearly} C1P goes back to \citet{booth1975pq} and has received thorough attention since \citep[e.g.,][]{hajiaghayi2002note,tan2007consecutive,chauve2009gapped,dom2010approximation,narayanaswamy2015obtaining}. Our study uses axioms and experiments instead of computational complexity, and focusses on selecting a good axis rather than measuring nearly single-peakedness.

\section{Preliminaries}  \label{sec:preliminaries}

When $i \le j$ are integers, we write $[i,j] := \{i, i+1, \dots, j\}$.

Let $C$ be a set of $m$ \emph{candidates}, and $V$ a set of $n$ \emph{voters}. An \emph{approval ballot} is a non-empty subset of candidates $A \subseteq C$. 
An \emph{approval profile} $P$ is a collection of $n$ approval ballots $P = (A_i)_{i \in V}$. We denote by $\mathcal P$ the set of all approval profiles. For two profiles $P_1$ and $P_2$, we write $P_1 + P_2$ for the profile obtained by combining the ballots in the two profiles.

An \emph{axis} $\axis$ is a strict linear order of the candidates, so that $a \axis b$ means that candidate $a$ is strictly on the left of $b$ on the axis. We write $a \axisweak b$ if $a \axis b$ or $a = b$.
For brevity, we will sometimes omit the $\axis$ and write $abc$ for the axis $a \axis b \axis c$. Let $\mathcal A$ be the set of all axes over $C$.
The direction of an axis is irrelevant, so we will informally treat the axes $abc$ and $cba$ as being the same axis.

An approval ballot $A_i$ is an \emph{interval} of an axis $\axis$ if for all pairs of candidates $a, b \in A_i$ and every $c$ such that $a \axis c \axis b$, we have $c \in A_i$. If instead $c\notin A_i$, we say that $c$ is an \emph{interfering candidate}. A profile $P$ is \emph{linear} if there exists an axis $\axis$ such that all approval ballots in $P$ are intervals of $\axis$. We also say that this axis $\axis$ is \emph{consistent} with the profile $P$. We write $\con(P) \subseteq \mathcal A$ for the set of all axes consistent with $P$. 

For an approval ballot $A$ and an axis ${\axis} = c_1c_2\dots c_m$ with candidates relabeled by their axis position, we denote by $x_{A,\axis} = (x_{A,\axis}^1, \dots, x_{A,\axis}^m)$ the \emph{approval vector} where $x_{A,\axis}^i = 1$ if $c_i \in A$ and $0$ otherwise. 
For instance, for the axis ${\axis} = abcd$ and ballot $A = \{b,c\}$, we get the vector $(0,1,1,0)$, while $A' = \{a,d\}$ gives the vector $(1,0,0,1)$ (which has two interfering candidates).
The \emph{approval matrix} of a profile $P =(A_i)_i$ has $x_{A_i,\axis}$ as its $i$th row. Thus, its $(i,j)$-entry is equal to $1$ if $c_j\in A_i$ and equal to $0$ if $c_j \not\in A_i$.
Note that a profile is linear if and only if its approval matrix (derived from an arbitrary axis $\axis$) satisfies the \emph{consecutive one property} (or C1P, see the survey by \cite{consecutiveones}), i.e., its columns can be reordered such that in each row, the ``1''s form an interval.

An \emph{axis rule} $f$ is a function that takes as input an approval profile $P$ and returns a non-empty set of axes $f(P) \subseteq \mathcal A$, such that for each $\axis$ in $f(P)$ its \emph{reverse} axis $\cev{\axis}$ is also in $f(P)$, encoding the idea that the direction of the axis does not matter.

In this paper, we will focus on the family of \emph{scoring rules}, which we define in analogy to other social choice settings \citep{Myer95b,Piva13a}. Let $\cost : 2^{C} \times \mathcal A \rightarrow \mathbb R_{\ge 0}$ be a \emph{cost function}, indicating the cost $\cost(A_i,\axis)$ that a ballot $A_i\in P$ incurs when the axis $\axis$ is chosen. By summing up these costs, we get the cost $\cost(P, \axis) = \sum_{A_i \in P} \cost(A_i, \axis)$ of an axis $\axis$ for the profile $P$. 
An axis rule $f$ is a \emph{scoring rule} if there is a cost function $\cost_f$ such that
$f(P) = \argmin_{{\axis} \in \mathcal A} \cost_f(P,\axis)$ for all profiles $P$.

A focus on the class of scoring rules can be justified as an analogue to scoring rules in voting theory, in that every scoring rule satisfies the \emph{reinforcement} axiom \citep{young1975social} which says that if $f$ chooses the same axis $\axis$ in two disjoint profiles $P_1$ and $P_2$, so that $f(P_1) \cap f(P_2) \neq \emptyset$, then the axes it chooses in the combined profile $P_1 + P_2$ are exactly the common axes, i.e., $f(P_1 + P_2) = f(P_1) \cap f(P_2)$. However, providing an axiomatic characterization of this class using reinforcement appears to be difficult since the neutrality axiom turns out to be quite weak in our setting.
Another motivation for scoring rules is their natural interpretation as \emph{maximum likelihood estimators} when there is a ground truth axis, as observed by \citet{conitzer2009preference} in the voting setting.
To see the connection, let $\axis$ be the ground truth axis, and suppose voters obtain their approval ballots $A_i$ i.i.d.\ from a probability distribution $\mathbb P(A_i \mid \axis)$ (where intuitively ballots are more likely the closer they are to forming an interval of $\axis$). Then, the likelihood of a profile $P$ is $\mathbb P( P \mid \axis) = \prod_{i} \mathbb P(A_i \mid \axis)$. To find the axis inducing maximum likelihood, we solve 
$\text{MLE}(P) := \argmax_{\axis} \mathbb P( P \mid \axis) = \argmin_{\axis} - \sum_{i} \log(\mathbb P(A_i \mid \axis))$, 
which is a scoring rule with costs $\cost_f(A_i, \axis) = - \log(\mathbb P(A_i \mid \axis))$.

\section{Axis Rules}  \label{sec:rules}

In this section, we introduce five scoring rules. Many are inspired by objective functions proposed for near single-peakedness \citep{faliszewski2014voterdeletion,escoffier2021nearlysp} in the context of ranking profiles.

The first and simplest rule is called \emph{Voter Deletion (VD)}:

\drule{Voter Deletion}{
   This rule returns the axes that minimize the number of ballots to delete from the profile $P$ in order to become consistent with it. This rule is a scoring rule based on the cost function $\cost_{\VD}$ such that $\cost_{\VD}(A,\axis) = 0$ if $A$ is an interval of $\axis$, and $1$ otherwise.
}

\begin{wrapfigure}{r}{0.3\linewidth}
    \centering
	\begin{tikzpicture}
        \axisheader{0}{a,b,c,d,e}
        
        \interval{1}{2}{4}

        \redinterval{2}{1}{1}
        \redinterval{2}{5}{5}

        \redinterval{3}{1}{2}
        \redinterval{3}{4}{5}

        \redinterval{4}{1}{2}
        \redinterval{4}{5}{5}

        \redinterval{5}{1}{1}
        \redinterval{5}{3}{3}
        \redinterval{5}{5}{5}
        \costcolumn{6}{\cost_{\VD}}
        \ballotcost{1}{6}{0}
        \ballotcost{2}{6}{1}
        \ballotcost{3}{6}{1}
        \ballotcost{4}{6}{1}
        \ballotcost{5}{6}{1}
	\end{tikzpicture}

    \wrapcaption{Costs of some ballots under the VD rule.}
    \label{tab:ex_vd}
\end{wrapfigure}

The idea behind this rule is that perhaps some ``maverick'' voters are ``irrational'', and should hence be disregarded. The aim is to delete as few maverick voters as possible. \Cref{tab:ex_vd} shows the costs of some ballots under the VD rule, and we clearly observe that the rule gives the same cost to all non-interval ballots.

An intuitive shortcoming of VD is that it does not measure the \emph{degree of incompatibility} of a given vote with an axis.
For example, VD does not distinguish ballots that miss just one candidate to be an interval, and an approval ballot in which only the two extreme candidates of the axis are approved. For this reason, more gradual rules might do better. 

The first rule in this direction is \emph{Minimum Flips} (MF) which changes ballots by removing and adding candidates.

\drule{Minimum Flips}{
	\setlength{\abovedisplayskip}{7pt}
	This rule returns the axes that minimize the total number of candidates that need to be removed from and added to approval ballots in order to make the profile linear. It is the scoring rule based on:
	\begin{align*}
		\cost_{\MF}(A,\axis) = \min\limits_{x,y \in A \, : \, x \axisweak y} &\:\,|\{ z \in A : z \axis x \mathrm{ \ or  \ } y \axis z \}| \\[-4pt]
		+&\:\, |\{ z \notin A : x \axis z \axis y\}|.
	\end{align*}
}
\begin{wrapfigure}{r}{0.3\linewidth}
    \centering
	\begin{tikzpicture}
	     \axisheader{0}{a,b,c,d,e}
	     
	     \interval{1}{2}{4}
	
	     \interval{2}{1}{1}
	     \flipinterval{2}{5}{5}
	
	     \interval{3}{1}{2}
	     \missingapproval{3}{3}
	     \missingapproval{3}{4}
	     \interval{3}{4}{5}
	
	     \interval{4}{1}{2}
	     \flipinterval{4}{5}{5}
	
	     \flipinterval{5}{1}{1}
	     \interval{5}{3}{3}
	     \missingapproval{5}{4}
	     \interval{5}{5}{5}
	
	     \costcolumn{6}{\cost_{\MF}}
	     \ballotcost{1}{6}{0}
	     \ballotcost{2}{6}{1}
	     \ballotcost{3}{6}{1}
	     \ballotcost{4}{6}{1}
	     \ballotcost{5}{6}{2}
	\end{tikzpicture}

    \wrapcaption{Costs of some ballots under the MF rule. Candidates that need to be added are represented by red ticks, candidates that need to be removed by blue ticks.}
    \label{tab:ex_mf}
\end{wrapfigure}

The definition of $\cost_{\MF}$ optimizes the choice of the left- and right-most candidates $x$ and $y$ in the ballot after removing and adding candidates, and then counts the number of candidates that were thus removed (first term of the sum) and added (second term).
We can equivalently view MF as finding for each vote $A_i$ the interval ballot closest to $A_i$ in Hamming distance, with that distance being the cost of $\axis$. In another equivalent view, the rule finds the linear profile of minimum total Hamming distance to the input profile, and returns its axes. \Cref{tab:ex_mf} shows the costs of some ballots under the MF rule. Observe that we can obtain an interval by only removing candidates (as in the second ballot), by only adding candidates (as in the third ballot), or by both removing and adding candidates (as in the last ballot).

In many applications, adding approvals seems better motivated than removing them. For example, a voter $i$ might not approve a candidate $c$ because $i$ does not know who $c$ is; fixing this error corresponds to adding a candidate.
On the other hand, choosing to approve some candidate by accident seems less likely.
The \emph{Ballot Completion} (BC) rule implements this thought.

\drule{Ballot Completion}{
	\setlength{\abovedisplayskip}{7pt}
    This rule returns the axes that minimize the number of candidates to add to approval ballots to make the profile consistent with it. It is the scoring rule based on:
    \[
    \cost_{\BC}(A,\axis) = |\{b \not\in A : a \axis b \axis c \text{ for some } a,c \in A\}|.
    \]
}

\begin{wrapfigure}{r}{0.3\linewidth}
    \centering
	\begin{tikzpicture}
        \axisheader{0}{a,b,c,d,e}
        
        \interval{1}{2}{4}

        \interval{2}{1}{1}
        \missingapproval{2}{2}
        \missingapproval{2}{3}
        \missingapproval{2}{4}
        \interval{2}{5}{5}

        \interval{3}{1}{2}
        \missingapproval{3}{3}
        \missingapproval{3}{4}
        \interval{3}{4}{5}

        \interval{4}{1}{2}
        \missingapproval{4}{3}
        \missingapproval{4}{4}
        \interval{4}{5}{5}

        \interval{5}{1}{1}
        \missingapproval{5}{2}
        \interval{5}{3}{3}
        \missingapproval{5}{4}
        \interval{5}{5}{5}
        \costcolumn{6}{\cost_{\BC}}
        \ballotcost{1}{6}{0}
        \ballotcost{2}{6}{3}
        \ballotcost{3}{6}{1}
        \ballotcost{4}{6}{2}
        \ballotcost{5}{6}{2}
	\end{tikzpicture}

	\wrapcaption{Costs of some ballots under the BC rule. Candidates that need to be added are represented by red ticks.}
    \label{tab:ex_bc}
\end{wrapfigure}

Thus, given a ballot $A$ and an axis $\axis$, this rule counts all interfering candidates with respect to $A$ and $\axis$. 
To see the difference between MF and BC, observe that $\cost_{\BC}(\{a,d\},abcd) = 2$ as we need to add $b$ and $c$ to obtain an interval, while $\cost_{\MF}(\{a,d\},abcd) = 1$ as we can just remove $a$. \Cref{tab:ex_bc} shows the costs of some ballots under the BC rule.

In the approval context, BC is the only rule we know of that has already been used in the literature to find an underlying political axis of voters, on the data of experiments conducted during the 2012 and 2017 French presidential elections \citep{france2012,france2017}. The axes found by BC were close to the orderings discussed in the media. 

The \emph{Minimum Swaps} (MS) rule modifies the \emph{axis} rather than the ballots. Given an approval ballot $A$, the MS rule asks how many candidate swaps we need to perform in an axis $\axis$ until $A$ becomes an interval of it: the cost $\cost_{\MS}(A, \axis)$ is the minimum Kendall-tau distance between $\axis$ and an axis $\axis'$ (the number of swaps of adjacent candidates needed to go from $\axis$ to $\axis'$) such that $A$ is an interval of $\axis'$. For instance, $\cost_{\MS}(\{ a,d\}, abcd) = 2 $ because we need to have $a$ next to $d$ on any axis consistent with $\{a, d\}$, and we need at least two swaps to obtain this.

\drule{Minimum Swaps}{
	\setlength{\abovedisplayskip}{7pt}
    This scoring rule uses the cost function
    \[ 
    \textstyle
    \cost_{\MS}(A,\axis) = \sum_{x \notin A} \min(|\{y \in A : y \axis x\}|, |\{y \in A : x \axis y \}|) .
    \]
}
\begin{wrapfigure}{r}{0.3\linewidth}
    \centering
	\begin{tikzpicture}
        \axisheader{0}{a,b,c,d,e}
        
        \interval{1}{2}{4}

        \interval{2}{1}{1}
        \interval{2}{5}{5}
        \swaparrow{2}{1}{2}
        \swapslot{2}{2}
        \swaparrow{2}{2}{3}
        \swapslot{2}{3}
        \swaparrow{2}{3}{4}
        \swapslot{2}{4}

        \interval{3}{1}{2}
        \interval{3}{4}{5}
        \swapslot{3}{3}
        \swaparrow{3}{3}{4}
        \swaparrow{3}{4}{5}

        \interval{4}{1}{2}
        \interval{4}{5}{5}
        \swaparrow{4}{4}{5}
        \swapslot{4}{3}
        \swaparrow{4}{3}{4}
        \swapslot{4}{4}

        \interval{5}{1}{1}
        \interval{5}{3}{3}
        \interval{5}{5}{5}
        \swaparrow{5}{1}{2}
        \swapslot{5}{2}
        \swapslot{5}{4}
        \swaparrow{5}{4}{5}

	    \costcolumn{6}{\cost_{\MS}}
	    \ballotcost{1}{6}{0}
	    \ballotcost{2}{6}{3}
	    \ballotcost{3}{6}{2}
	    \ballotcost{4}{6}{2}
	    \ballotcost{5}{6}{2}
	\end{tikzpicture}

	\wrapcaption{Costs of some ballots under the MS rule. Arrows indicate the swaps needed to make the vote an interval.}
    \label{tab:ex_MS}
\end{wrapfigure}
\noindent
To see why this formula implements our swapping description of $\cost_{\MS}(A, \axis)$, note that to modify the axis $\axis$ such that $A$ becomes an interval of it, we need to ``push outside'' all $x \notin A$ such that there exist $y,z \in A$ with $y \axis x \axis z$. We can either push $x$ to the left side or to the right side, and thus we must swap $x$ with at least all candidates $y \in A$ to its right or to its left. We prove this more formally in \Cref{sec:appMS}. \Cref{tab:ex_MS} shows the costs of some ballots under the MS rule. Note that the order in which the swaps are performed matters. For instance, we need to swap the same pairs for the third and fourth ballot ($\{c,d\}$ and $\{d,e\}$), but we start by swapping $c$ and $d$ in the third ballot and $d$ and $e$ in the fourth ballot.

Our last rule is \emph{Forbidden Triples} (FT), inspired by a proposal for rankings by \citet{escoffier2021nearlysp}. It is defined by counting the number of violations of the interval condition, as formally defined in \Cref{sec:preliminaries}.

\drule{Forbidden Triples}{
	\setlength{\abovedisplayskip}{7pt}
This scoring rule uses the cost function
\begin{align*}
    \cost_{\FT}(A, \axis) &= |\{(x,y,z) : x,z \in A, y \notin A, x \axis y \axis z \}|. 
\end{align*}
}

\begin{wrapfigure}{r}{0.3\linewidth}
    \centering
	\begin{tikzpicture}
        \axisheader{0}{a,b,c,d,e}
        
        \interval{1}{2}{4}

        \interval{2}{1}{1}
        \ft{2}{2}{1}
        \ft{2}{3}{1}
        \ft{2}{4}{1}
        \interval{2}{5}{5}

        \interval{3}{1}{2}
        \ft{3}{3}{4}
        \interval{3}{4}{5}

        \interval{4}{1}{2}
        \ft{4}{3}{2}
        \ft{4}{4}{2}
        \interval{4}{5}{5}

        \interval{5}{1}{1}
        \ft{5}{2}{2}
        \interval{5}{3}{3}
        \ft{5}{4}{2}
        \interval{5}{5}{5}

        \costcolumn{6}{\cost_{\FT}}
        \ballotcost{1}{6}{0}
        \ballotcost{2}{6}{3}
        \ballotcost{3}{6}{4}
        \ballotcost{4}{6}{4}
        \ballotcost{5}{6}{4}
	\end{tikzpicture}

	\wrapcaption{Costs of some ballots under the FT rule. The number of forbidden triples involving each interfering candidate is shown in red.}
    \label{tab:ex_ft}
\end{wrapfigure}

Note that there is one forbidden triple for each combination of an interfering candidate and a pair of candidates lying on its left and its right, respectively. Thus, $\cost_{\FT}(A,\axis) = \textstyle \sum_{x \notin A} |\{y \in A : y \axis x\}| \times |\{y \in A : x \axis y \}|$. 
For instance, we have $\cost_{\FT}(\{a,b,d,e\},abcde) = 2 \times 2= 4$ while $\cost_{\FT}(\{a,b,c,e\},abcde) = 3 \times 1 = 3$.
Intuitively, this rule looks at the holes in a vote, with larger holes  separating many approved candidates counting more. \Cref{tab:ex_ft} shows the costs of some ballots under the FT rule.

The cost functions of our five scoring rules can be related via a chain of inequalities, suggesting that they form a natural collection of rules to study.
\begin{restatable}{proposition}{hierarchy}
     For all ballots $A$ and axes $\axis$, we have $\cost_{\VD}(A,\axis) \le \cost_{\MF}(A,\axis) \le \cost_{\BC}(A,\axis) \le \cost_{\MS}(A,\axis) \le \cost_{\FT}(A,\axis)$.
\end{restatable}
\begin{proof}
    To see $ \cost_{\VD}(A,\axis) \le \cost_{\MF}(A,\axis)$, note that if $A$ is not an interval of $\axis$ then $\cost_{\VD}(A,\axis) = 1$ and at least one candidate must be flipped to make $A$ an interval of $\axis$, so $\cost_{\MF}(A,\axis) \ge 1$. If $A$ is an interval then $\cost_{\MF}(A,\axis) = \cost_{\VD}(A,\axis) = 0$.
    
    We have $\cost_{\MF}(A,\axis) \le \cost_{\BC}(A,\axis)$ because in MF we can add and remove approvals, but in BC we can only add approvals. We can also see this from the formal definitions: $\cost_{\MF}(A,\axis)$ is a minimum, and if $x$ and $y$ are the left- and right-most approved candidates in $A$, we obtain the value of $\cost_{\BC}(A,\axis)$, which must thus be at least as high as the minimum over $x$ and $y$.
    
    Finally, observe that for any interfering candidate $x$ on $A$,
    $ \min(|\{y \in A : y \axis x\}|, |\{y \in A : x \axis y \}|) \ge 1$. Moreover, as these are all natural numbers,  $\min(|\{y \in A : \mbox{\ensuremath{y \axis x}}\}|, |\{y \in A : x \axis y \}|) \le |\{y \in A : y \axis x\}| \times |\{y \in A : x \axis y \}|   $. Thus, $\cost_{\BC}(A,\axis) \le \cost_{\MS}(A,\axis)  \le \cost_{\FT}(A,\axis)$ by the definitions of these rules.
\end{proof}

We say that two axis rules $f_1$ and $f_2$ are \emph{equivalent} if for all profiles $P$ we have $f_1(P) = f_2(P)$. Note that if $n \le 2$ or $m \le 2$, every profile is linear. Moreover, if there are $m = 3$ candidates, all the rules defined in this section are equivalent (as there is only one non-interval approval vector, so the only possible costs are $0$ and $1$). If there are $m=4$ candidates, VD and MF are equivalent and BC and MS are equivalent.  
This is because the respective cost functions coincide for $m\leq 4$, which does not remain true for $m \ge 5$. 
Indeed, for $m \ge 5$, the rules are pairwise non-equivalent.
\Cref{ex1} shows a profile with $m=4$ for which VD, BC, and FT all select different axes. We give another profile in \Cref{sec:app_rules_not_eq} with $m=7$ for which no two rules select the same axes.

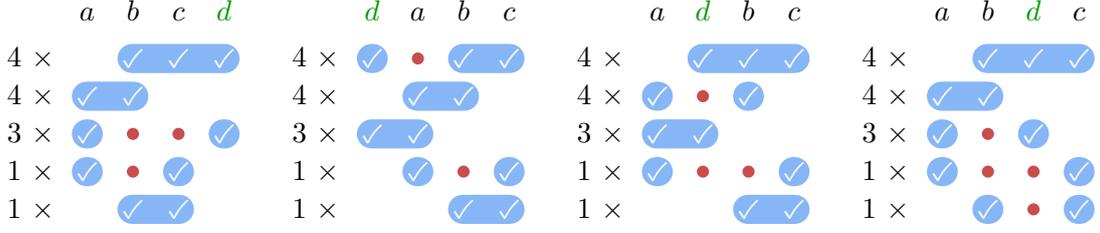
\begin{figure}[!t]
    \centering
\begin{tikzpicture}
    \begin{scope}[xshift=0cm]
        \axisheader{0}{a,b,c,d/green!60!black}
        
        \multiplicity{1}{4}
        \interval{1}{2}{4}

        \multiplicity{2}{4}
        \interval{2}{1}{2}

        \multiplicity{3}{3}
        \interval{3}{1}{1}
        \interfering{3}{2}
        \interfering{3}{3}
        \interval{3}{4}{4}

        \multiplicity{4}{1}
        \interval{4}{1}{1}
        \interfering{4}{2}
        \interval{4}{3}{3}
        
        \multiplicity{5}{1}
        \interval{5}{2}{3}
    \end{scope}
    
    \begin{scope}[xshift=3.75cm]
        \axisheader{0}{d/green!60!black,a,b,c}
        
        \multiplicity{1}{4}
        \interval{1}{1}{1}
        \interfering{1}{2}
        \interval{1}{3}{4}

        \multiplicity{2}{4}
        \interval{2}{2}{3}

        \multiplicity{3}{3}
        \interval{3}{1}{2}

        \multiplicity{4}{1}
        \interval{4}{2}{2}
        \interfering{4}{3}
        \interval{4}{4}{4}

        \multiplicity{5}{1}
        \interval{5}{3}{4}
    \end{scope}

    \begin{scope}[xshift=7.5cm]
        \axisheader{0}{a,d/green!60!black,b,c}
        
        \multiplicity{1}{4}
        \interval{1}{2}{4}

        \multiplicity{2}{4}
        \interval{2}{1}{1}
        \interfering{2}{2}
        \interval{2}{3}{3}

        \multiplicity{3}{3}
        \interval{3}{1}{2}

        \multiplicity{4}{1}
        \interval{4}{1}{1}
        \interfering{4}{2}
        \interfering{4}{3}
        \interval{4}{4}{4}

        \multiplicity{5}{1}
        \interval{5}{3}{4}
    \end{scope}

    \begin{scope}[xshift=11.25cm]
        \axisheader{0}{a,b,d/green!60!black,c}
        
        \multiplicity{1}{4}
        \interval{1}{2}{4}

        \multiplicity{2}{4}
        \interval{2}{1}{2}

        \multiplicity{3}{3}
        \interval{3}{1}{1}
        \interfering{3}{2}
        \interval{3}{3}{3}

        \multiplicity{4}{1}
        \interval{4}{1}{1}
        \interfering{4}{2}
        \interfering{4}{3}
        \interval{4}{4}{4}

        \multiplicity{5}{1}
        \interval{5}{2}{2}
        \interfering{5}{3}
        \interval{5}{4}{4}
    \end{scope}
    
\end{tikzpicture}

    \caption{Profile of \Cref{ex1} on 4 different axes. Red circles indicate interfering candidates.}
    \label{tab:ex2}
\end{figure}

\begin{example} \label{ex1}
    Consider the profile $P = (4 \times \{b,c,d\},4 \times \{a,b\}, 3 \times \{a,d\}, 1 \times \{a,c\}, 1 \times \{b,c\})$. On this profile, all rules agree that $a \axis b \axis c$, but they disagree on the position of $d$. {Indeed, $\axis_1 = abc\underline{d}$ is optimal for VD and MF, $\axis_2 = \underline{d}abc$ for BC and MS, and $\axis_3= a\underline{d}bc$ and $\axis_4 = ab\underline{d}c$ for FT. \Cref{tab:ex2} shows the profile aligned according to the four possible axes. One can easily see that among these axes (1) the axis $\axis_1$ on the left minimizes the VD cost with only 4 non-interval ballots, (2) the axis $\axis_2$ in the middle minimizes the BC cost with 5 red circles and (3) the axes $\axis_3$ and $\axis_4$ on the right minimizes the FT cost with 6 forbidden triplets.}
\end{example}

As we already mentioned, problems about recognizing matrices that are almost C1P have long been known to be NP-complete. Hardness of VD and BC is explicitly known (see \citet{booth1975pq}), and the reductions only use approval sets of size 2. The results for other rules directly follow from the observation that they are equivalent to either VD or BC when $\max_i |A_i| = 2$ (See \Cref{sec:appComplexity} for a detailed proof.) 

\begin{restatable}{theorem}{complexity}
    The VD, MF, BC, MS, and FT rules are NP-complete to compute, even for profiles in which every ballot approves at most 2 candidates.
\end{restatable}

A lot of other axis rules could be defined. However, in this paper, we focus on the five rules introduced above, and leave the study of potential other rules to further research. In particular, we think that greedy variants of the rules we introduced are of interest to circumvent computational hardness.

\section{Axiomatic Analysis}  \label{sec:axiom}

In this section, we conduct an axiomatic analysis of the rules we introduced. \Cref{tab:summary_properties} summarizes the results of this section. 

We start with some basic axioms that all our rules satisfy. The first two are classic symmetry axioms: a rule $f$ is \emph{anonymous} if whenever two profiles $P$ and $P'$ are such that every ballot appears exactly as often in $P$ as in $P'$, then $f(P) = f(P')$. It is \emph{neutral} if for every profile $P$, renaming the candidates in $P$ leads to the same renaming in $f(P)$. 
The third basic property fundamentally captures the aim of an axis rule: if there are perfect axes, then the rule should return those.

\property{Consistency with linearity}{
A rule $f$ is \emph{consistent with linearity} if $f(P) = \con(P)$ for all linear profiles $P$.
}

If $f$ is a scoring rule and it satisfies these three axioms, we can deduce that its underlying cost function has a certain structure. In particular, the cost function attains its minimum value for consistent axes, it is invariant under reversing the axis, and it is symmetric.

\begin{restatable}{lemma}{costfunction}
	\label{lem:costfunction}
	Let $f$ be a scoring rule. Then, $f$ is neutral and consistent with linearity if and only if it is induced by a cost function $\cost_f$ such that
	\begin{enumerate}
		\item[(1)] for all $A$ and all $\axis$, we have $\cost_f(A, \axis) \ge 0$, and $\cost_f(A,\axis) = 0$ if and only if $A$ is an interval of $\axis$,
		\item[(2)] for all $A$ and all $\axis$, we have $\cost_f(A, \axis) = \cost_f(A,\cev\axis)$, and
		\item[(3)] there exists a function $g:\{0,1\}^m \rightarrow \mathbb R_{\ge 0}$ such that for all $A$ and all $\axis$, we have $\cost_f(A,\axis) = g(x_{A,\axis}) = g(x_{A,\cev\axis})$ (in other words,  $\cost_f$ depends only on the induced approval vector $x_{A,\axis}$).
	\end{enumerate}
\end{restatable}

\noindent
We provide the formal proof in \Cref{sec:appNeutr}.

\subsection{Stability and Monotonicity}

\begin{table}[!t]
    \centering
    \begin{tabular}{l c c c c c} 
	\toprule
    & VD & MF & BC & MS & FT \\ 
    \midrule
    Stability & \yes & \no & \no & \no & \no \\
    Ballot monotonicity & \yes & \no  & \yes & \no & \no \\
    Clearance & \no & \no & \yes & \yes & \yes \\
    Veto winner centrism & \no & \no & \no & \yes & \yes \\
    Clone-proximity & \no & \no & \no & \no & \yes \\ 
    Clone-resistance & \yes & \no  & \no & \no & \no \\
    \bottomrule
    \end{tabular}
    \caption{Properties satisfied by the axis rules.
    }
    \label{tab:summary_properties}
\end{table}
Some rules are more sensitive to changes in information than others. 
Intuitively, Voter Deletion rarely reacts to changes in the profile, as it only checks whether the ballots are intervals of the axis or not. Thus, a single voter will have little effect on the axes selected. Indeed, for VD, adding a new ballot to the profile cannot completely change the set of optimal solutions. For other rules, this is not the case. We can formalize this behavior in the following axiom.

\property{Stability}{
A rule $f$ satisfies \emph{stability} if for every profile $P$ and approval ballot $A$, we have $f(P) \cap f(P + \{A\}) \ne \emptyset$.
}

This axiom was also considered by \citet[Sec.~4.4.2]{tydrichova2023structural} in the context of rankings.
A similar axiom is used by \citet{ceron2021approval} to characterize Approval Voting as a single-winner voting rule.
Whether stability is a desirable property depends on the context: while it implies that the rule is robust, it also means that the rule might disregard too much information.

\begin{restatable}{proposition}{stability} \label{thm:stability}
    Stability is satisfied by VD, but not by MF, BC, MS, and FT. 
\end{restatable}

\begin{proof}
    Let us prove that VD satisfies stability. Let $P$ be a profile and $A$ an approval ballot. If  $\VD(P + \{A\}) \subseteq \VD(P)$, then clearly $\VD(P) \cap \VD(P + \{A\}) \neq \emptyset$. Otherwise, let ${\axis'} \in \VD(P + \{A\})\setminus \VD(P)$ and ${\axis} \in \VD(P)$. Then, $\cost_{\VD}(P, \axis) \le \cost_{\VD}(P, \axis') -1 $. Moreover, by definition of VD, $0 \le \cost_{\VD}(A, \axis) \le 1$ for all axes $\axis$. Put together, we have: 
    \begin{align*}
    \cost_{\VD}(P + \{A\},\axis) &= \cost_{\VD}(P,\axis) + \cost_{\VD}(A, \axis) \\
    &\le \left(\cost_{\VD}(P, \axis') - 1\right) + 1 \\
    &\le \cost_{\VD}(P, \axis') + \cost_{\VD}(A, \axis') \\ 
    & = \cost_{\VD}(P + \{A\}, \axis')
    \end{align*}
    Therefore, $\cost_{\VD}(P + \{A\},\axis) \le \cost_{\VD}(P + \{A\}, \axis') $ and thus, because ${\axis'} \in \VD(P + \{A\})$, we must also have ${\axis} \in \VD(P + \{A\})$, and thus $\VD(P) \cap \VD(P + \{A\}) \neq \emptyset$, as required.

    For $f \in \{ \MF,\BC,\MS,\FT\}$, let us consider the profile $P = (\{ a,b,e\}, \{ a,b,c,e\}, \{ b,c,d,e,f\})$.   
    By consistency with linearity, $f(P) = \{\mathit{aebcfd}, aebcdf, abecdf, \mathit{abecfd} \}$  (up to the reversed axes). Now, consider the ballot  $A = \{a,b,d,f\}$. For every ${\axis} \in f(P)$, we have $\cost_{\MF}(P + \{A\}, \axis ) = \cost_{\BC}(P + \{A\}, \axis ) = 2$, $\cost_{\MS}(P + \{A\}, \axis ) \in \{3, 4\}$ and $\cost_{\FT}(P + \{A\}, \axis ) = 8$. 
    
    However, let us consider the axis ${\axis'} = \mathit{ceabfd} \notin f(P)$. The only ballot in $P + \{A\}$ that is not an interval of $\axis'$ is $\{ b,c,d,e,f\}$, and thus we can calculate that $\cost_{\MF}(P + \{A\}, \axis' ) = \cost_{\BC}(P + \{A\}, \axis' ) = 1$, $\cost_{\MS}(P + \{A\}, \axis' ) = 2$ and $\cost_{\FT}(P + \{A\}, \axis' ) = 6$. Therefore, none of the axes in $f(P)$ is optimal for the profile $P + \{A\}$, and hence $f(P) \cap f(P + \{A\}) = \emptyset$ for $f \in \{ \MF,\BC,\MS,\FT\}$. Thus, these rules do not satisfy stability.

    \begin{figure}[!t]
        \centering
        \begin{tikzpicture}
            \begin{scope}[xshift=0cm]
                \axisheader{0}{a,e,b,c,f,d}
                
                \interval{1}{1}{3}

                \interval{2}{1}{4}

                \interval{3}{2}{6}
                
                \node at (-0.3, -4*\rowheight) {$A:$};
                \interval{4}{1}{1}
                \interfering{4}{2}
                \interval{4}{3}{3}
                \interfering{4}{4}
                \interval{4}{5}{6}
                
            \end{scope}
            
            \begin{scope}[xshift=6cm]
                \axisheader{0}{c,e,a,b,f,d}

                \interval{1}{2}{4}

                \interval{2}{1}{4}

                \interval{3}{1}{2}
                \interfering{3}{3}
                \interval{3}{4}{6}
                
                \node at (-0.3, -4*\rowheight) {$A:$};
                \interval{4}{3}{6}
                
            \end{scope}
        \end{tikzpicture}
        \caption{Profile $P$ and ballot $A$ (the last row) with the axes $\mathit{aebcfd}$ and $\mathit{ceabfd}$ in the proof of \Cref{thm:stability}. Red circles indicate interfering candidates.}
        \label{tab:ex_stability}
    \end{figure}
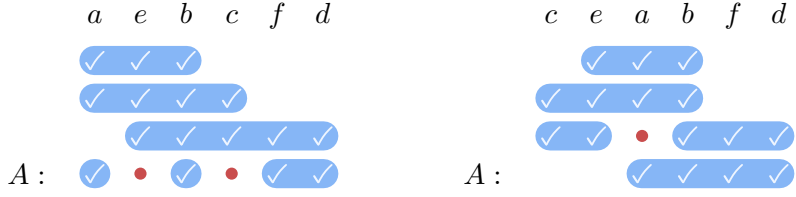
\end{proof}

Monotonicity axioms say that if the input changes so as to more strongly support the current output, then the output should stay the same. 
For our setting, we define monotonicity to say that if some voters \emph{complete} their ballots by approving all interfering candidates
with respect to the current axis $\axis$, then $\axis$ should continue being selected.%
\footnote{One could define monotonicity in other ways, but we leave the study of those variants to future work.}

\property{Ballot monotonicity}{
A rule $f$ satisfies \emph{ballot monotonicity} if for every profile $P$, ballot $A \in P$ and axis ${\axis} \in f(P)$ such that $A$ is not an interval of $\axis$, we still have ${\axis} \in f(P')$ for the profile $P'$ obtained from $P$ by replacing $A$ by the interval $A' = \{x \in C :\exists y,z \in A \text{ s.t. } y \axisweak x \axisweak z\}$.
}

VD and BC satisfy this axiom, but the other rules do not.

\begin{restatable}{proposition}{monot}
    Ballot monotonicity is satisfied by VD and BC, but not by MF, MS, and FT. 
\end{restatable}
\begin{proof}
    We first prove the result for VD and BC. Informally, by changing the ballot $A$ to $A'$ we decrease the cost of $\axis$ by $1$ for VD, and the cost of all other axes decreases by at most $1$, so $\axis$ is still among the selected axes. For BC, suppose we need to add $k$ candidates to the ballot $A$ as part of making the profile linear. Then the change to $A'$ reduces the cost of $\axis$ by $k$, and the cost of any other axis decreases by at most $k$ (as we added only $k$ candidates to $A$), so $\axis$ is still  selected. 

    We now prove it more formally. Let $P$ be a profile and $\axis \in f(P)$ an optimal axis. Let $A \in P$ be a ballot that is not an interval of $\axis$, $A' = \{x \in C :\exists y,z \in A \text{ s.t. } y \axisweak x \axisweak z\}$ the completion of $A$, and $P'$ the profile obtained from $P$ by replacing $A$ by $A'$.

    For VD, we have that $\cost_{\VD}(A, \axis) = 1$ and $\cost_{\VD}(A', \axis) = 0$. For every axis $\axis'$, we have $\cost_{\VD}(A, \axis') \le 1$ and $\cost_{\VD}(A', \axis') \ge 0$. This gives the following.
    \begin{alignat*}{7}
    	\cost_{\VD}(P', \axis) &= \cost_{\VD}(P, \axis) &&- \cost_{\VD}(A, \axis) &&+ \cost_{\VD}(A', \axis) &&= \cost_{\VD}(P, \axis) -1 \text{, and} \\
    	\cost_{\VD}(P', \axis') &= \cost_{\VD}(P, \axis') &&- \cost_{\VD}(A, \axis') &&+ \cost_{\VD}(A', \axis') &&\ge \cost_{\VD}(P, \axis') - 1 \text{ for all $\axis'$.}
    \end{alignat*}
    Since $\axis \in f(P)$, we have $\cost_{\VD}(P, \axis) \le \cost_{\VD}(P, \axis')$ and thus $\cost_{\VD}(P', \axis) \le \cost_{\VD}(P', \axis')$ for all axes $\axis'$. Therefore, $\axis \in f(P')$, and VD satisfies ballot monotonicity.

    We use a similar reasoning for BC. Let $k = \cost_{\BC}(A,\axis)$ and $\cost_{\BC}(A', \axis) = 0$. This means that we add $k > 0$ candidates to ballot $A$, thereby obtaining a ballot $A'$ that forms an interval of $\axis$. As before, since we added only $k$ candidates, the cost of any other axis $\axis'$ decreases by at most $k$, i.e., $\cost_{\BC}(A, \axis') - \cost_{\BC}(A', \axis') \le k$. Thus,
    \begin{alignat*}{7}
    	\cost_{\BC}(P', \axis) &= \cost_{\BC}(P, \axis) &&- k &&+ 0 &&= \cost_{\BC}(P, \axis) - k \text{, and} \\
    	\cost_{\BC}(P', \axis') &= \cost_{\BC}(P, \axis') &&- \cost_{\BC}(A, \axis') &&+ \cost_{\BC}(A', \axis') &&\ge \cost_{\BC}(P, \axis') - k \text{ for all $\axis'$.}
    \end{alignat*}
    Since $\axis \in f(P)$, we have $\cost_{\BC}(P, \axis) \le \cost_{\BC}(P, \axis')$ and thus $\cost_{\BC}(P', \axis) \le \cost_{\BC}(P', \axis')$ for all $\axis'$. Therefore, $\axis \in f(P')$, and BC satisfies ballot monotonicity.
    
    To see that the other rules do not satisfy ballot monotonicity, consider some rule $f \in \{\MF, \MS, \FT\}$, a set of $6$ candidates $C = \{a,b,c,d,e,f\}$, and the profile $P$ containing each of the $\binom{6}{4}$ possible ballots of $4$ candidates once. As $f$ satisfies neutrality, all axes are chosen, and there is $x \in \mathbb R$ such that $\cost_{f}(\axis, P) = x$ for all $\axis$. 
    Consider now the axis $\axis_1 = \mathit{abcdef}$, and the ballot $A = \{a,b,c,f\} \in P$. Let $P'$ be the profile obtained from $P$ in which $A$ is replaced by $A' = \{a,b,c,d,e,f\}$.
    Since ${\axis_1} \in f(P)$, it suffices to show that ${\axis_1} \notin f(P')$.
    Denote $\axis_2 = \mathit{abdefc}$. 
    For every axis $\axis$, we have $\cost_f(\axis, A') = 0$, and thus by definition of $P'$, we have
    \[
    	\cost_{f}(\axis, P') = \cost_{f}(\axis, P) + \cost_f(\axis, A') - \cost_f(\axis, A) = x - \cost_f(\axis, A).
    \]
    Additionally, note that for MF, MS, and FT, the cost of $A$ on $\axis_1$ is respectively $1$, $2$, and $6$, while the cost on $\axis_2$ is respectively $2$, $4$, and $8$. Therefore, we have $\cost_f(\axis_2, A) > \cost_f(\axis_1, A)$, which implies  $\cost_{f}(\axis_2, P') < \cost_{f}(\axis_1, P')$, and so ${\axis_1} \notin f(P')$, as required. 
\end{proof}

\subsection{Centrists and Outliers}

On a high level, good axes should place less popular candidates towards the extremes, where they are less likely to destroy intervals. Conversely, popular candidates are safer to place in the center. We will define two axioms that identify profiles where this expectation is strongest, and that require candidates to be accordingly placed in center or extreme positions.

Our first axiom considers the placement of very unpopular candidates. The axiom is easiest to satisfy by placing them at the extremes, but it does not require doing so in all cases.

\property{Clearance}{
A rule $f$ satisfies \emph{clearance} if for every profile $P$ in which some candidate $x$ is never approved, all ${\axis} \in f(P)$ are such that there is no $A \in P$ with $y,z \in A$ and $y \axis x \axis z$.}

Thus, under clearance, never-approved candidates cannot be interfering.%

\begin{restatable}{proposition}{clearance} \label{thm:clearance}
Clearance is satisfied by BC, MS, and FT, but not by VD and MF. 
\end{restatable}
\begin{proof}
        Let $f \in \{\BC, \MS, \FT\}$. We show that $f$ satisfies clearance. Let $P$ be a profile with a never-approved candidate $x$ and let $\axis$ be an axis such that there is a ballot $A$ in $P$ with $x$ interfering $A$ on $\axis$.       
        We will show that $\axis \not \in f(P)$. Consider the axis $\axis'$ identical to $\axis$ but in which $x$ was moved to the left extreme. As $x$ is interfering $A$ on $\axis$, we have
        \begin{align*}
        	\cost_{\BC}(A,\axis') &= \cost_{\BC}(A,\axis) - 1, \\
        	\cost_{\MS}(A,\axis') &= \cost_{\MS}(A,\axis) - \min(|\{y: y \axis x \}|, |\{y: x \axis y\}|)\text{, and} \\
        	\cost_{\FT}(A,\axis') &= \cost_{\FT}(A,\axis) - |\{y: y \axis x \}| \cdot |\{y: x \axis y\}|.
        \end{align*}
        In each case, we have $\cost_{f}(A,\axis') < \cost_f(A,\axis)$. This is true for all ballots $A$ for which $x$ is an interfering candidate on $\axis$. For all other ballots $A \in P$ for which $x$ is not interfering, note that already in $\axis$, candidate $x$ is placed outside of all candidates approved by $A$ (either to the left of the left-most approved candidate, or to the right of the right-most approved candidate). In $\axis'$, we have moved $x$ to an even more extreme position, but it follows from the rules' definitions that this does not change the cost, i.e., $\cost_f(A,\axis') = \cost_f(A,\axis)$. Since there exists at least one ballot for which $x$ is interfering on $\axis$, we have that $\cost_f(P, \axis') < \cost_f(P, \axis)$ and $\axis \not \in f(P)$, as required.
    
        Now let $f \in \{\VD, \MF \}$. We show that $f$ fails clearance. Consider the profile $P = (\{a,b\},\{a,c\},\{a,d\})$ on the set of candidates $C = \{a,b,c,d,e\}$. This profile is not linear because at most two of the candidates $b,c,d$ can be placed next to $a$. Thus, for each ${\axis} \in f(P)$, we have $\cost_f(P, \axis) \geq 1$. Consider the axis ${\axis} = baced$. We have $\cost_f(P, \axis) = 1$, so ${\axis} \in f(P)$. But $e$ is never approved and it interferes with the ballot $\{ a,d\}$ on $\axis$. Hence $f$ does not satisfy clearance. 
\end{proof}

While VD and MF always choose \emph{some} axis that satisfies the clearance condition, they can additionally choose axes which violate this condition, and hence they fail the axiom.

For another way of formalizing the intuition that unpopular candidates should be placed at the extremes, we consider \emph{veto profiles} in which every ballot has size $m-1$, i.e., each voter approves all but one of the candidates.
For a veto profile, the only voters who will approve an interval are those who veto a candidate at one extreme of the axis. Since veto profiles do not have any interesting structure, the best candidates to put at the left and right end of the axis are the two candidates with the lowest approval score (i.e., the \emph{most} vetoed candidates). All of our rules indeed choose only such outcomes.

We can extend this intuition to say that candidates that are vetoed more frequently should be placed at positions closer to the extremes. This would imply that the \emph{least} vetoed candidate should be placed in the center, so that as few ballots as possible have holes in the center.

\property{Veto winner centrism}{
A rule $f$ satisfies \emph{veto winner centrism} if for every veto profile $P$, the median candidate (or one of the two median candidates if the number of candidates is even) 
of every axis ${\axis} \in f(P)$ has the highest approval score.
}

Among the rules studied in this paper, only MS and FT satisfy veto winner centrism.

\begin{restatable}{proposition}{veto}
    Veto winner centrism is satisfied by MS and FT, but not by VD, MF, and BC.
\end{restatable}
\begin{proof}
   Let $P$ be a veto profile. Let us denote the candidates by $c_1, c_2, \hdots, c_m$, and let $A_{{-i}}$ be the ballot approving all candidates but $c_i$. For each $i \in \{1, \hdots, m \}$, we denote by $n_i$ the number of occurrences of $A_{{-i}}$ in $P$. Since $P$ is a veto profile, $n = n_1 + n_2 + \hdots + n_m$.
   
   Let us now prove that FT and MS satisfy veto winner centrism. For simplicity, we assume that $m$ is odd, and so there is only one median candidate on the axis, at position $\textsf{med} := \frac{m+1}{2}$ (however, note that the reasoning below also works for $m$ even, with only a  slight straightforward modification).
   
   Regarding FT, given an axis $\axis$, let us denote by $k_i$ the position of $c_i$ on $\axis$. Then each copy of $A_{{-i}}$ in $P$ creates $t_{k_i} = (k_i -1) \cdot(m - k_i)$ forbidden triples and $\cost_{\FT}(P, \axis) = \sum_{i = 1}^m n_i \cdot t_{k_i}$. 
   Note that $t_{k_i}$ is maximal when $k_i = \textsf{med}$ (i.e., when $c_i$ is the median candidate of $\axis$).
   Without loss of generality, let $c_1$ be a most approved candidate and let $c_2$ be a candidate with strictly fewer approvals, i.e., $n_2 > n_1$. 
   Consider any axis $\axis$ for which $c_2$ is the central candidate, and let $\axis'$ be an axis obtained from $\axis$ by swapping the positions of $c_1$ and $c_2$. We claim that $\cost_{\FT}(P, \axis') < \cost_{\FT}(P, \axis)$. For each $i \neq 1,2$, the number of forbidden triples induced by ballots of type $A_{{-i}}$ is the same for both axes, as this number only depends on the position of $c_i$ on the axis.
   Hence, $\cost_{\FT}(P, \axis')$ and $\cost_{\FT}(P, \axis)$ only differ in triples caused by ballots of type $A_{{-1}}$ and $A_{{-2}}$. Thus,
   \begin{align*}
        \cost_{\FT}(P, \axis') - \cost_{\FT}(P, \axis) 
       &= (n_1t_{\textsf{med}} + n_2t_{k_1} )-  (n_1t_{k_1} + n_2t_{\textsf{med}}) \\
       &= (n_1 - n_2)(t_{\textsf{med}}-t_{k_1}) 
        < 0, 
   \end{align*} 
   as $n_1 < n_2$ and $t_{\textsf{med}}>t_{k_1}$. This implies that no axis whose center candidate has fewer than the maximum number of approvals can be optimal, and thus FT satisfies veto winner centrism.

   Regarding MS, we proceed similarly to the case of FT, by noting that each copy of $A_{{-i}}$ in $P$ generates $t_{k_i} = \min\{ k_i -1, m- k_i\}$ swaps. Indeed, $c_i$ is the unique non-approved candidate in $A_{{-i}}$, so it needs to be swapped with all the candidates on its left, or its right. This value is maximal if $c_i$ is the median candidate of the axis, i.e., if $k_i = \textsf{med} = \frac{m+1}{2}$. It is now easy to see that the previous argument also works for MS.
   
   Regarding VD, MF, and BC, note that these three rules are all equivalent on veto profiles, since their cost functions equal $1$ for every ballot $A_{{-i}}$ that is not an interval of a given axis. 
   Thus, each axis $\axis$ with left- and rightmost candidates $c_l$ and $c_r$ has a cost of $n - n_l - n_r$ according to these rules. It follows that an axis is optimal if and only if its two outermost candidates correspond to the two least approved candidates. In particular, the optimality of a solution is independent of the position of the most approved candidate (provided it is not placed at the extremes). Hence, if $m \ge 5$, for all veto profiles $P$ there exists an optimal axis such that the most approved candidate is not the median candidate. Therefore, VD, MF, and BC fail veto winner centrism. 
\end{proof}

In fact, MS and FT always place candidates so that the approval scores are single-peaked.

Clearance and veto winner centrism suggest that MS and FT use the information in a profile accurately by correctly placing popular and unpopular candidates. Their tendency to put low-approval candidates towards the ends is also confirmed by our experiments in \Cref{sec:experiments}.
While this generally seems sound, in the political context it can lead to wrong answers: there can be ideologically centrist candidates who don't get many votes due to not being well-known. We leave for future work whether there are rules that can correctly place candidates in these contexts.

\subsection{Clones and Resistance to Cloning}

We now focus on the behaviour of rules in the presence of essentially identical candidates. We say that $a, b \in C$ are \emph{clones} if for each voter $i \in V$, $a \in A_i$ if and only if $b \in A_i$. While perfect clones are rare, two candidates may have very similar sets of supporters, and studying clones gives insights for how rules handle similar candidates. 

Intuitively, one would expect clones to be next to each other on any optimal axis. This is captured by the following axiom:

\property{Clone-proximity}{
A rule $f$ satisfies \emph{clone-proximity} if for every profile $P$ in which $a, a' \in C$ are clones, for every axis ${\axis} \in f(P)$, every candidate $x$ with $a \axis x \axis a'$ or $a' \axis x \axis a$, and every $A \in P$, we have $x \in A$ whenever $a, a' \in A$.
}
Note that in the definition, $x$ is not necessarily a clone of $a$ and $a'$, because $x$ can be approved even if $a$ and $a'$ are not approved.

Surprisingly, only FT satisfies clone-proximity. 
All of our rules choose at least one axis where the clones are next to each other, but the rules other than FT may choose extra axes with a violation, as we show in the following result.

\begin{restatable}{proposition}{cloneproximity} \label{thm:cloneprox}
	Clone-proximity is satisfied by FT, but not by VD, MF, BC, and MS.
\end{restatable}
\begin{proof}
    We first prove that FT satisfies clone-proximity. Let $P = (A_1, \dots, A_n)$ be a profile where $a$ and $a'$ are clones. Let $\axis$ be an axis. We denote by $T_{\axis}$ the set of all forbidden triples $(i, l,c,r) \in V \times C^3$ such that $l \axis c \axis r$ and $l, r \in A_i$ but $c \notin A_i$. Then, $\cost_{\FT}(P, \axis) = |T_{\axis}|$.

    First note that we cannot have a forbidden triple $(i,l,c,r) \in T_{\axis}$ with one of the clones as $c$ (in the center) and the other on one side ($l$ or $r$), as $a$ and $a'$ are always approved together. Thus, the only triples in $T_{\axis}$ involving both $a$ and $a'$ are those for which both sides $l$ and $r$ are one of the clones, i.e., triples of form $(\cdot,a,\cdot,a')$ and $(\cdot,a',\cdot,a)$. For an axis $\axis$, let us denote by $S^{\{a,a'\}}_{\axis}$ the number of such triples in $T_{\axis}$. Moreover, let $S^a_{\axis}$ be the number of triples involving $a$ and not $a'$ and let $S^{a'}_{\axis}$ be the number of triples involving $a'$ and not $a$. Finally, let $S^0_{\axis}$ be the number of triples involving neither $a$ nor $a'$. For every axis $\axis$, we have $\cost_{\FT}(P, \axis) = S^{\{a,a'\}}_{\axis} + S^a_{\axis} + S^{a'}_{\axis} + S^0_{\axis}$.
    
    Now take any $\axis$ where the clones are not next to each other, i.e., there exists $x \in C$ such that $a \axis x \axis a'$ or $a' \axis x \axis a$ and a ballot $A_i \in P$ such that $a,a' \in A_i$ and $x \notin A_i$.
    We will show that $\axis\not\in\FT(P)$.
    
    Note that by choice of $\axis$, we have $S^{\{a,a'\}}_{\axis} \ge 1$. Assume without loss of generality that $S^a_{\axis} \le S^{a'}_{\axis}$. Let us consider the axis $\axis'$ obtained by moving $a'$ next to $a$ on $\axis$, i.e., there is no $x \in C$ such that $a \axis x \axis a'$ or $a' \axis x \axis a$. Thus, we have $S^{\{a,a'\}}_{\axis'} = 0$, and $S^{a'}_{\axis'} = S^{a}_{\axis'} = S^a_{\axis}$, as all triples that do not involve $a'$ will not be affected by the move, and the triples involving $a'$ will be the same as those involving $a$ now that they are next to each other. For the same reason, $S^0_{\axis'} = S^0_{\axis}$. Thus, we have the following:
    \begin{alignat*}{7}
        \cost_{\FT}(P, \axis) &= S^{\{a,a'\}}_{\axis} &&+ S^a_{\axis} &&+ S^{a'}_{\axis} &&+ S^0_{\axis} \\
        &\ge \quad 1 &&+ S^a_{\axis} &&+ S^{a}_{\axis} &&+ S^0_{\axis} \\
        &> \quad 0 &&+ S^a_{\axis'} &&+ S^{a'}_{\axis'} &&+ S^0_{\axis'} \\
        &= \mathrlap{\cost_{\FT}(P, \axis').}
    \end{alignat*}
    Since $\axis'$ has a lower FT cost than $\axis$, we have $\axis\not\in\FT(P)$, as required. Hence FT satisfies clone-proximity.

    We now show that other rules do not satisfy clone-proximity.
    For VD, BC, and MF, consider the following profile:

    \begin{minipage}{0.5\textwidth}
        \begin{align*}
            \\[7pt]
            2 & \times\{a_1, a_2\}, \\
            2 & \times \{a_2, a_3\}, \\
           1 & \times\{x,x',a_1,a_3\}.\\
        \end{align*}
        \end{minipage}
        \begin{minipage}{0.5\textwidth}
            \begin{tikzpicture}
                \axisheader{0}{x,a_1,a_2,a_3,x'}
                \multiplicity{1}{2}
                \interval{1}{2}{3}
    
                \multiplicity{2}{2}
                \interval{2}{3}{4}
    
                \multiplicity{3}{1}
                \interval{3}{1}{2}
                \interfering{3}{3}
                \interval{3}{4}{5}
            \end{tikzpicture}
        \end{minipage}

    Because of the cycle among $a_1, a_2, a_3$, this profile is not linear and so all axes have cost at least $1$. Now observe that the axis $x \axis a_1 \axis a_2 \axis a_3 \axis x'$ has cost $1$ for VD, BC, and MF. On this axis, $a_2$ is between the clones $x$ and $x'$ but is never approved with them. Thus, these three rules fail clone-proximity. (However, note that all these rules also choose the compliant axis $x \axis x' \axis a_1 \axis a_2 \axis a_3'$.)

    For MS, consider the following profile:
    \begin{align*}
         1: \{a,a',b,b'\}, \quad 1: \{b,b',x, x'\}, \quad 1: \{x, x', a,a'\}
    \end{align*}
    By neutrality of Minimum Swaps, the cost of all axes in which clones are next to each other is the same as the cost of $a \axis a' \axis x \axis x' \axis b \axis b'$, which is $4$. However, another axis has cost $4$ for Minimum Swaps: $x \axis a \axis a' \axis x' \axis b \axis b'$. On this axis, $a$ is between the clones $x$ and $x'$ but $a$ is not approved in the ballot $\{b,b',x,x'\}$, containing $x$ and $x'$. Thus, Minimum Flips also fails clone-proximity.
\end{proof}

Inspired by axioms from voting theory \citep{tideman1987independence}, we could require that removing or adding a clone to the profile does not change the result. More precisely, if we remove a clone from a profile, the restriction of any optimal axis should remain optimal, and adding a clone to a profile should not modify the relative order of the other candidates on any optimal axis.  
To formally define this, we need some notation.
For a profile $P$ defined on a set $C$ of candidates, we denote by $P_{C'}$ the restriction of $P$ to a subset of candidates $C' \subseteq C$. We also denote $P_{-c}$ the restriction of the profile to $C \setminus \{c\}$ where $c \in C$ is a given candidate. Similarly, we define $\axis_{C'}$ and $\axis_{-c}$. We can now state the axiom:

\property{Resistance to cloning}{
A rule $f$ is \emph{resistant to cloning} if for every profile $P$ in which $a, a' \in C$ are clones, (1) for all axes ${\axis} \in f(P)$, we have $\axis_{-a} \in f(P_{-a})$ and (2) for all axes $\axis^* \in f(P_{-a})$, there is an axis ${\axis} \in f(P)$ with $\axis_{-a} = \axis^*$.
}

Among the rules studied in this paper, only VD is resistant to cloning.

\begin{restatable}{proposition}{resistancetocloning} \label{thm:clone_resistance}
    Resistance to cloning is satisfied by VD, but not by MF, BC, MS, and FT.
\end{restatable}
\begin{proof}
    We start by proving that $f = \VD$ satisfies resistance to cloning. We first check (1). Let ${\axis} \in f(P)$; we need to show that $\axis_{-a} \in f(P_{-a})$. Because all interval ballots of $P$ on $\axis$ will remain interval ballots of $P_{-a}$ on $\axis_{-a}$, we have $\cost_{\VD}(\axis_{-a},P_{-a}) \le \cost_{\VD}(\axis,P)$. Now, assume for a contradiction that $\axis_{-a} \notin f(P_{-a})$ and instead some axis ${\axis'} \in f(P_{-a})$ is optimal, with $\cost_{\VD}(\axis',P_{-a}) < \cost_{\VD}(\axis_{-a},P_{-a})$. Consider the axis $\axis'_{+a}$ obtained from $\axis'$ by placing $a$ next to $a'$. Because $a$ and $a'$ are clones, an approval ballot of $P$ is an interval of $\axis'_{+a}$ if and only if its restriction in $P_{-a}$ is an interval of $\axis'$. Thus, $\cost_{\VD}(\axis'_{+a},P) = \cost_{\VD}(\axis',P_{-a}) $. Combining all of this, we have:
    \begin{align*}
        \cost_{\VD}(\axis'_{+a},P) = \cost_{\VD}(\axis', P_{-a}) < \cost_{\VD}(\axis_{-a},P_{-a}) \le \cost_{\VD}(\axis,P)
    \end{align*}
    which contradicts the optimality of $\axis$ for $P$. Hence $\axis_{-a} \in f(P_{-a})$.

    Next, we check (2) using the same reasoning. Let ${\axis} \in f(P_{-a})$ and let $\axis_{+a}$ be the axis obtained from $\axis$ by placing $a$ next to $a'$. We will show that $\axis_{+a} \in f(P)$. Again, we have $\cost_{\VD}(\axis_{+a},P) = \cost_{\VD}(\axis,P_{-a})$. Now assume for a contradiction that there is ${\axis'} \in f(P)$ with a lower cost than $\axis_{+a}$, i.e., $\cost(\axis',P) < \cost(\axis_{+a},P)$. As explained above, we have $\cost_{\VD}(\axis'_{-a},P_{-a}) \le \cost_{\VD}(\axis',P)$. Combining these three inequalities gives $\cost_{\VD}(\axis'_{-a},P_{-a}) < \cost_{\VD}(\axis,P_{-a})$, which contradicts the optimality of $\axis$. Hence $\axis_{+a} \in f(P)$, as required.

To prove that BC does not satisfy resistance to cloning, let us consider the profile $P = (3 \times \{ b,a,a'\}, 4 \times \{ c, a, a'\}, 2 \times \{ b,c\})$. It is easy to check that the unique optimal axis (up to reversal, and up to permutation of $a$ and $a'$) is $b \axis c \axis a \axis a'$ with $\cost_{\BC}(P, \axis)$ = 3. Indeed, if $a$ and $a'$ are not next to each other, at least two types of ballot will not be interval of the axis, which will yield a cost of at least 5, and the axes on which $b$ and $c$ are the extremities a cost of at least 4. 
However, if we remove the candidate $a'$, the cost of $ \axis_{-a'} = bca$ is 3. It is hence no longer optimal, as the axis $\axis^* = b ac$ achieves a lower cost of 2. 

We use a very similar idea to prove that MF does not satisfy resistance to cloning. We consider the profile $P = (1 \times \{ b,d\}, 2 \times \{ b, a, a'\}, 2 \times \{ c, a, a'\}, 1 \times \{ a,e\}, 3 \times \{ b,c,d,e\})$. We can check that the axis ${\axis} = dbceaa'$ is optimal for MF with $\cost_{\MF}(P, \axis) = 4$. If $a$ and $a'$ are not next to the other on the axis, at least two of the ballot types $\{b,a,a'\}$, $\{c,a,a'\}$ and $\{b,c,d,e\}$ are not intervals, which yields a cost of at least 4.

Any axis of form (up to reversal) $\{ d,b\} \axis \{ a, a'\} \axis \{ c,e\}$ has a cost greater or equal than 6 because of ballots $\{b,c,d,e\}$. Any axis with one candidate on the left of $\{ a, a'\}$ and three candidate on the right of $\{ a, a'\}$ (up to reversal) has a cost of at least 4: the ballots $\{ b,c,d,e\}$ generates at least 3 flips, and at least one of the ballots $\{ b,d\}, \{ c,e \}$ is not an interval either. 

However, $\axis_{-a'} = dbcea$ is not optimal for $P_{-a'}$: $\cost_{\MF}(P_{-a'}, \axis_{-a'}) = 4$ (ballots $\{ b,a\}$ and $\{c,a\}$ are not intervals). The axis $\axis^* = dbace$ has a lower cost of 3 (as each copy of ballot $\{ d,b,c,e\}$ generates one flip and all other ballots are intervals of the axis). 

To prove that FT and MS do not satisfy resistance to cloning, let us consider the profile $P = (3 \times \{ a, b\}, 3 \times \{ b, c\}, 1 \times \{ a,c,d\})$, and let $f \in \{ \FT, \MS\}$. We have $f(P) = \{ \axis^1, \axis^2 \}$ with $\axis^1 = abcd$ and $\axis^2 = dabc$  (up to the reversed axes). Indeed, $\cost_{\FT}(P, \axis^i) = 2$  and $\cost_{\MS}(P, \axis^i) = 1$ for $i \in \{ 1,2\}$. These are the only axes on which both $\{ a,b\}$ and $\{ b,c\}$ are intervals -- in other words, the cost of any other axis will be at least 3. Let us now consider a profile $P' = (3 \times \{ a', a, b\}, 3 \times \{ b, c\}, 1 \times \{ a', a,c,d\})$. We note that it is the profile $P$ to which we have added a candidate $a'$, clone of $a$. Under resistance to cloning, there should be an axis ${\axis} \in f(P')$ such that $\axis_{-a} = \axis^1$. Among all possible axes generalizing $\axis^1$, the best one for $\MF$ and $\FT$ (up to the permutation of $a$ and $a'$) is ${\axis} = aa'bcd$, with a cost of 4 for FT and 2 for MS. However, we can find an axis $\axis^* = daa'bc$ with cost of 3 for FT and 1 for MS. Hence, there is no ${\axis} \in f(P')$ such that $\axis_{-a} = \axis^1$. Thus, $\FT$ and $\MS$ do not satisfy resistance to cloning. 
\begin{figure}[!t]
    \centering
\begin{tikzpicture}
    \begin{scope}[xshift=0cm]
        \axisheader{0}{a,\color{gray}a',b,c,d}
        
        \multiplicity{1}{3}
        \interval{1}{1}{3}

        \multiplicity{2}{3}
        \interval{2}{3}{4}

        \multiplicity{3}{1}
        \interval{3}{1}{2}
        \interfering{3}{3}
        \interval{3}{4}{5}
    \end{scope}
    
    \begin{scope}[xshift=6cm]
        \axisheader{0}{d,a,\color{gray}a',b,c}
        
        \multiplicity{1}{3}
        \interval{1}{2}{4}

        \multiplicity{2}{3}
        \interval{2}{4}{5}

        \multiplicity{3}{1}
        \interval{3}{1}{3}
        \interfering{3}{4}
        \interval{3}{5}{5}
        
    \end{scope}
\end{tikzpicture}
    \caption{Profile $P$ with the axes $aa'bcd$ and $daa'bc$ in the proof of \Cref{thm:clone_resistance}. Red circles indicate interfering candidates, and the clone $a'$ of $a$ is grayed out.}
    \label{tab:ex_clones}
\end{figure}
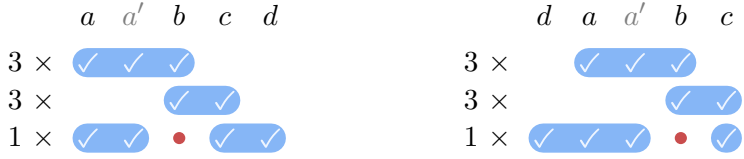
\end{proof}

% \end{document}
These two clone axioms are quite strong: each excludes all but one of our rules.
Indeed, we now show that if a scoring rule satisfies neutrality and consistency with linearity, then clone-proximity and resistance to cloning are actually incompatible.%
\footnote{There are rules that are not scoring rules which satisfy both clone resistance and proximity. For example, consider the rule that takes a profile, identifies all maximal clone sets, and replaces each by a single representative candidate. Then apply a rule to the collapsed profile and de-replace the representatives.}
\begin{restatable}{theorem}{impossibility}
    No neutral scoring rule satisfies resistance to cloning, clone proximity, and consistency with linearity.
\end{restatable}
\begin{proof}
Let $f$ be a scoring rule satisfying all four axioms, and $\cost_f$ its cost function.
As proven in \Cref{lem:costfunction}, by neutrality there is a function $g_f:\{0,1\}^m \rightarrow \mathbb R_{\ge 0}$ such that %
$\cost_f(A, \axis) = g_f(x_{A,\axis}) = g_f(x_{A,\cev\axis})$,
where $x_{A,\axis}$ is the approval vector of $A$ and $x_{A,\cev\axis}$ is the reversed vector.

Let $y$ be the minimal cost over all approval vectors that are not intervals. In particular,  $y \le  g_f((1,0,1,0)) = g_f((0,1,0,1))$ and $y \le g_f((1,0,0,1))$. By \Cref{lem:costfunction} (using consistency with linearity), $y > 0$.
Moreover, let $y' = g_f((1,0,1,1)) = g_f((1,1,0,1))$.
Let $q \in \mathbb N$ with $q > y'/y$ and consider the profile $P = (q \times \{b,c\}, q \times \{c,d\}, 1 \times \{a,b,d\})$. 
For ${\axis} \in \{ \axis^1, \axis^2\}$ with $\axis ^1 = a b c d$ and $\axis ^2 = b c d a$, we have $\cost_f(\{a,b,d\}, \axis) = y'$.
All other axes break one of the pairs $\{ b,c\}$, $\{ c,d\}$, thus ensuring a cost of at least $q\cdot y > y'$. Therefore, ${\axis_1}, {\axis_2} \in f(P)$.  

Consider now the profile $P'$ in which we add a clone $b'$ of $b$: $ P' = (q \times \{b,b',c\}, q \times \{c,d\}, 1 \times \{a,b,b',d\})$. By clone-proximity, 
$b$ and $b'$ are next to each other on every ${\axis} \in f(P')$. 
By resistance to cloning, there exists $\axis^3$ (resp. $\axis^4$) in $f(P')$ extending $\axis^1$ (resp. $\axis^2$). Combining this with neutrality, $f(P')$ contains $\axis ^3= a b b' c d$ and $\axis ^4 = b b' c d a$,
which thus must have the same cost.
Since the ballots $\{b,b',c\}$ and $\{c,d\}$ are intervals of both of these axes and the rule is consistent with linearity, they contribute a cost of $0$ and thus the cost difference of the two axes only depends on the remaining ballot $\{a,b,b',d\}$. This implies $\cost_f(\{a,b,b',d\},\axis^3) = \cost_f(\{a,b,b',d\},\axis^4)$, i.e., $g_f((1,1,1,0,1)) = g_f((1,1,0,1,1))$. 

Now, consider the profile $P''$ which is a copy of $P$ but with a clone $a'$ of $a$:  $ P'' = (q \times \{b,c\}, q \times \{c,d\}, 1 \times \{a,a',b,d\})$. 
Using the same arguments as in the case of $P'$ yields two optimal axes $\axis^5 = a a' b c d$ and $\axis^6 =  b c d a a'$.
However, let us now compare $\axis^5$ to $\axis^7 = a  b c d a'$. The ballots $\{b,c\}$ and $\{c,d\}$ are intervals of both axes, and the cost of $\{a,a',b,d\}$ is the same on both, as we already showed that $g_f((1,1,1,0,1)) = g_f((1,1,0,1,1))$. Thus, $\axis^7$ is also an optimal axis, which is in contradiction with clone-proximity, since $a$ and $a'$ are not next to each other.
\end{proof}

We can show that resistance to cloning and ballot monotonicity in fact characterize VD among scoring rules. This not only distinguishes VD from the other introduced rules, but shows its normative appeal among the entire class of scoring rules. 
The full proof is in \Cref{sec:appCharacterization}, where we also show that the axioms are logically independent, assuming neutrality.
\begin{restatable}{theorem}{vdcharac}
	\label{thm:vdcharac}
    Let $m \ge 6$, and let $f$ be a neutral scoring rule. Then $f$ satisfies consistency with linearity, ballot monotonicity, and resistance to cloning if and only if it is VD.
\end{restatable}
\begin{proof}[Proof sketch.]
      Let $f$ be a scoring rule satisfying neutrality, consistency with linearity, resistance to cloning and ballot monotonicity. As shown in \Cref{sec:appNeutr}, $f$ is induced by a symmetric cost function $\cost$ with $\cost(A,\axis) = 0$ if and only if $A$ forms an interval in $\axis$. Further, $\cost$ only depends on the approval vector $x_{A,\axis}$, i.e., there exists a function $g:\{0,1\}^m \rightarrow \mathbb R_{\ge 0}$ such that $\cost(A,\axis) = g(x_{A,\axis})$ for all ballots $A$ and axis $\axis$.

    The steps of the proof are as follows:
    \begin{enumerate}
        \item Using ballot monotonicity, we show that there is a function $h$ such that for all $A$ and $\axis$ such that $A$ is not an interval of $\axis$, $\cost(A,\axis) = h(m, k_{\text{app}} ,k_{\text{int}})$, where $m$ is the number of candidates, $k_{\text{app}} =|A|$ is the number of approved candidates and $k_{\text{int}}$ is the number of interfering candidates.
        \item Using resistance to cloning, we show that for $A$ not interval of $\axis$, $\cost(A,\axis)$ only depends on the sum $k_{\text{app}}+k_{\text{int}}$, i.e., there is $h$ such that $\cost(A,\axis) = h(m,  k_{\text{app}}+k_{\text{int}})$.
        \item We show that for $A$ not interval of $\axis$, $\cost(A,\axis)$ can only take two values: $\cost(A,\axis) = h_m^*$ if $k_{\text{app}}+k_{\text{int}}=m$ and $\cost(A,\axis) = h_m$ otherwise.
        \item Finally, we show that $h^*_m = h_m$ and thus that the rule is VD. \qedhere
    \end{enumerate}
\end{proof}

If we drop ballot monotonicity, we can use similar ideas to show that resistance to cloning characterizes (under mild conditions) the class of  \emph{topological rules}. These are scoring rules with a function $h$ such that $\cost_f(A,\axis) = h(k)$, where $k$ is the number of contiguous holes that the axis $\axis$ creates in $A$ (\Cref{sec:appTopological}), such as the \emph{genus rule} with $h(k) = k$ that counts the total number of contiguous holes. For instance, on ${\axis} = abcde$, according to the genus rule, the cost induced by $\{a,e\}$ is 1, but the one induced by $\{a,c,e\}$ is 2.

\subsection{Heredity and Partition Consistency}

Resistance to cloning can be strengthened to \emph{heredity} \citep{tydrichova2023structural}, a kind of independence of irrelevant alternatives axiom. It states that if we remove \emph{any} candidate (not just a clone), the rule should return the original axes with that candidate omitted.

\property{Heredity}{
A rule $f$ satisfies \emph{heredity} if for every profile $P$ and every subset of candidates $C' \subseteq C$, we have that for each axis ${\axis} \in f(P)$, there exists $\axis^* \in f(P_{C'})$ such that $\axis_{C'} = \axis^*$.
}

However, an easy impossibility theorem shows that no reasonable axis rule can satisfy this axiom.

\begin{proposition}
    No axis rule satisfies heredity and consistency with linearity.
\end{proposition}
\begin{proof}
	Let $f$ be an axis rule satisfying heredity and let $P = ( \{a,b\},\{a,c \},\{a,d\})$. Let ${\axis} \in f(P)$. In $\axis$, there must be at least two candidates on the same side of $a$ (as there are two sides and three candidates $b$, $c$, and $d$), without loss of generality $b$ and $c$. By heredity, if we remove $d$, in $f(P_{-d})$ there must be an axis where $a$ is in an extreme position. However by consistency with linearity, $f(P_{-d}) = \{bac, cab\}$, a contradiction.	
\end{proof}

Because of this impossibility, we cannot construct an axis by succesively adding candidates. However, if a profile can be decomposed into subprofiles, we can expect the optimal axes to be the various concatenations of the optimal axes for the subprofiles. Specifically, given a profile $P$, consider the co-approval equivalence relation $\sim$ on $C$ with $x \sim y$ whenever there is some ballot $A\in P$ with $x,y \in A$. We call the equivalence classes $C_1, \dots, C_k$ of $\sim$ the \emph{co-approval partition} of $P$. Consider for instance the following profile:
\vspace{-10pt}

\begin{minipage}{0.5\textwidth}
    \begin{align*}
        \\[7pt]
        5 & \times \{a,b,c\}, \\
        4 & \times \{c,d\}, \\
        3 & \times \{x,y\}, \\
        2 & \times \{w,x,y\}, \\
        1 & \times \{a,b,d\}, \\
        1 & \times \{a,c\}, \\
        1 & \times \{w,y,z\}. \\
    \end{align*}
    \end{minipage}
    \begin{minipage}{0.5\textwidth}
    \begin{tikzpicture}
        \renewcommand{\rowheight}{0.595cm}
        \axisheader{0}{a,b,c,d,w,x,y,z}
        
        \multiplicity{1}{5}
        \interval{1}{1}{3}
    
        \multiplicity{2}{4}
        \interval{2}{3}{4}
    
        \multiplicity{3}{3}
        \interval{3}{6}{7}
    
        \multiplicity{4}{2}
        \interval{4}{5}{7}
    
        \multiplicity{5}{1}
        \interval{5}{1}{2}
        \interval{5}{4}{4}
    
        \multiplicity{6}{1}
        \interval{6}{1}{1}
        \interval{6}{3}{3}
    
        \multiplicity{7}{1}
        \interval{7}{5}{5}
        \interval{7}{7}{8}
    \end{tikzpicture}
    \end{minipage}

The co-approval partition of this profile consists of $C_1 = \{a,b,c,d\}$ and $C_2 = \{w,x,y,z\}$, since every ballot forms a subset of one of these two classes. We  can partition $P$ into two subprofiles $P_1$ and $P_2$ defined on $C_1$ and $C_2$. We define a property called \emph{partition consistency}, which in this example says that if the optimal axis in $P_1$ is $abcd$ and in $P_2$ is $wxyz$, then the optimal axes for the complete profile $P$ should be $abcdwxyz$, $abcdzyxw$, $dcbawxyz$ and $dcbazyxw$ (up to reversal).

\property{Partition consistency}{
A rule $f$ satisfies \emph{partition consistency} if for every profile $P$ with co-approval partition $C_1, \dots, C_k$, we have that $\axis \in f(P)$ if and only if for each $j \in [1,k]$, the class $C_j$ is an interval of $\axis$ and the axis $\axis_{C_j}$ restricted to $C_j$ is a member of $f(P_{C_j})$, where $P_{C_j}$ is the profile obtained from $P$ by restricting to $C_j$.
}

From the computational point of view, this property helps to reduce computation time as it divides the task into several sub-profiles with a smaller number of candidates (formally, computing an optimal axis can be done in time that is fixed-parameter tractable with respect to the size of the largest set $C_j$ in the partition). Note also that partition consistency implies clearance, as a never-approved candidate $x$ forms a singleton equivalence class, so partition consistency implies that for each $\axis \in f(P)$, there cannot be co-approved candidates $a, a'$ with $a \axis x \axis a'$. Thus, since VD and MF fail clearance, they also fail partition consistency.

\begin{restatable}{proposition}{partitionconsistency} \label{thm:partition_consistency}
    Partition consistency is satisfied by BC, MS and FT, but not by VD and MF.
\end{restatable}
\begin{proof}
    Let $f \in \{\BC, \MS, \FT\}$. We show that $f$ satisfies partition consistency. Let $P$ be a profile with co-approval partition $C_1, \dots, C_k$. Let $\axis$ be an axis on which we have $a \axis b \axis a'$ for some candidates $a, a' \in C_j$ and $b \not\in C_j$. We will show that $\axis\notin f(P)$.
    
    Define $\axis_{C_j}$  to be the restriction of $\axis$ to candidates from $C_j$ and let $\axis' = \axis_{C_1} \dots \axis_{C_k}$. Since every ballot contains only candidate from one of the $C_j$ and since the relative order of the candidates in each subaxis is preserved, we only removed interfering candidates for each ballot by moving from $\axis$ to $\axis'$. It follows that $\cost_f(A, \axis) \ge \cost_f(A, \axis')$ for all ballots.\footnote{Note that this  inequality is true for all rules defined in \Cref{sec:rules} (but the following strict version can fail for VD and MF). It is clear for VD. It is also clear for BC, MS and FT as we sum over interfering candidates and that this term in the sum only depends on the approved candidates. MF also satisfies this condition as the first term does not depend on interfering candidates, and the second term increases when we add interfering candidates.} Because $a,a' \in C_j$, there is some ballot $A^* \in P$ approving both of them. Also, $A^*$ only approves candidates from $C_j$ and thus does not approve $b \not\in C_j$, which is thus an interfering candidate for $A^*$. Moving $b$ away when moving from $\axis$ to $\axis'$ thus means, from the definitions of rule $f$, that $\cost_f(A^*, \axis) > \cost_f(A^*, \axis')$. This proves that $\axis \notin f(P)$. 
    
    Thus, in all optimal axes $\axis \in f(P)$, all the sets $C_j$ form intervals. Thus, we can write each such axis as $\axis = \axis_{\sigma_1} \dots \axis_{\sigma_k}$ for some permutation $(\sigma_1, \dots, \sigma_k)$ of $[1,k]$, where $\axis_j$ orders the candidates in $C_j$. For each $j$, let $P_{C_j}$ be the profile obtained from $P$ by restricting to $C_j$. From the definitions of $f$, for each ballot $A \in P_{C_j}$, we have $\cost_f(A, \axis) = \cost_f(A, \axis_j)$. Thus $\cost_f(P, \axis) = \sum_{j =1}^k \cost_f(P_{C_j}, \axis_j)$. From this we directly deduce that $\axis \in f(P)$ if and only if $\axis_j \in f(P_{C_j})$ for all $j \in [1,k]$. This proves that $f$ satisfies partition consistency.

    For $f \in \{\VD,\MF\}$ we can use the examples from the proof of \Cref{thm:clearance} with $C_1 = \{a,b,c,d\}$ and $C_2 = \{e\}$, $P_1 = P$ and $P_2 = \emptyset$ to show that $f$ does not satisfy partition consistency.
\end{proof}

\section{Experiments}  \label{sec:experiments}

In this section, we investigate the rules from \Cref{sec:rules} using an experimental analysis, based on synthetic and real datasets. While the rules are hard to compute, 
for $m$ up to about 12, we can find the best axes in reasonable time. We describe the algorithms we used in \Cref{sec:appExpe}, including a brute force approach (using pruning and heuristics) and ILP encodings.

Our main aims are (1) to compare our rules to each other, and (2) to compare our rules for approval profiles to two known rules for nearly single-peakedness for ranking profiles.

\subsection{Synthetic Data}
\label{sec:synthetic}

To better understand how different rules behave, 
we tested them on several synthetic data models (see \Cref{sec:appSynthetic} for detailed descriptions and results) which sample a linear profile on a ground truth axis and add random noise to it.
We then measured the distance of a rule's output to the ground truth.
Some of our rules are in fact the MLEs of these noise models, so as predicted they perform well in those cases. However some rules adapted better than others to different noise models. 
We observed that for all models, our rules tend to push the least approved candidates towards the extremes.

To compare approval-based and ranking-based rules, we introduce the \emph{noisy observation model}, inspired by random utility models such as the Thurstone--Mosteller model. Each candidate and voter $x \in C \cup V$ is associated with a position $p(x) \in \mathbb R$ on the line. Each voter $v$ estimates the position of each candidate $c$ under independent normal noise: $p_v(c) = p(c) + \mathcal N(0,\sigma)$ with $\sigma$ a parameter of the model. Voters approve (resp. rank) candidates based on their estimations. More precisely, the approval set of voter $v$ contains all candidates such that $|p(v)-p_v(c)| \le r$, where the approval radius $r$ is a parameter of the model. The ranking of $v$ is given by decreasing distances between $p(v)$ and $p_v(c)$. 

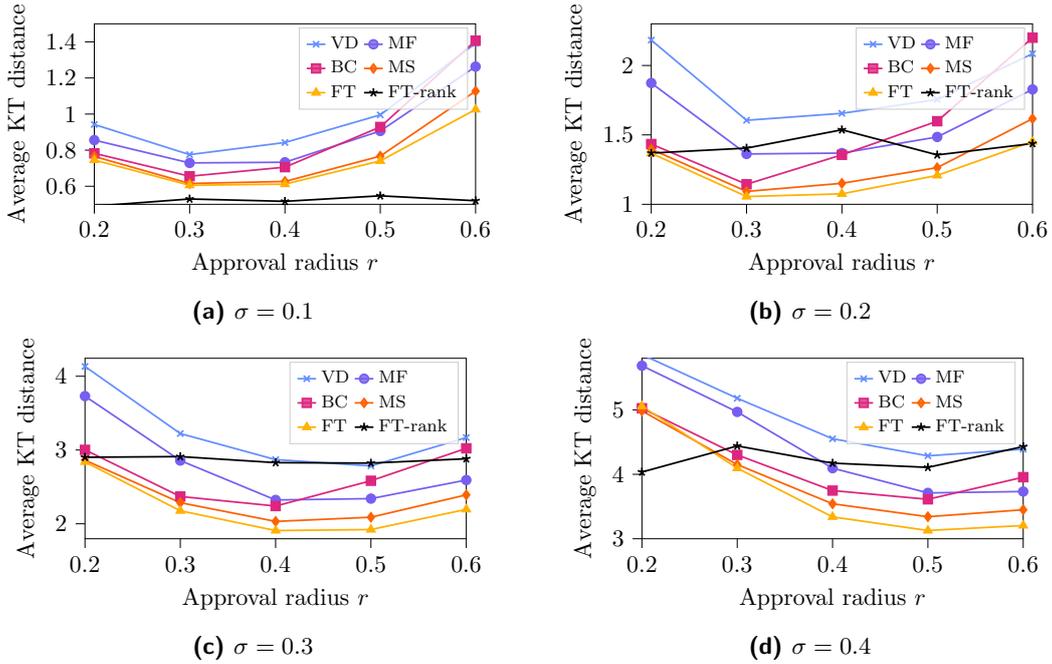
\begin{figure}[!t]
    \centering
    \begin{subfigure}{.45\textwidth}
        \centering
        \scalebox{0.82}{% This file was created with tikzplotlib v0.10.1.
\begin{tikzpicture}

\definecolor{crimson2143940}{RGB}{254, 97, 0}
\definecolor{darkgray176}{RGB}{176,176,176}
\definecolor{darkorange25512714}{RGB}{120, 94, 240}
\definecolor{forestgreen4416044}{RGB}{	220, 38, 127}
\definecolor{lightgray204}{RGB}{204,204,204}
\definecolor{mediumpurple148103189}{RGB}{255, 176, 0}
\definecolor{sienna1408675}{RGB}{0,0,0}
\definecolor{steelblue31119180}{RGB}{100, 143, 255}

\begin{axis}[
legend cell align={left},
legend style={fill opacity=0.8, draw opacity=1, text opacity=1, draw=lightgray204,
  legend columns=2,},
tick align=outside,
tick pos=left,
x grid style={darkgray176},
xlabel={Approval radius $r$},
xmin=0.2, xmax=0.6,
xtick style={color=black},
y grid style={darkgray176},
ylabel={Average KT distance},
xtick distance=0.1,
ymin=0.5, ymax=1.5,
ytick style={color=black},
height=4.5cm, width=7.7cm
]
\addplot [thick, steelblue31119180, mark=x, mark size=2, mark options={solid}]
table {%
0.2 0.942394047619048
0.3 0.775284523809524
0.4 0.842316666666667
0.5 0.995516666666667
0.6 1.39187591575092
};
\addlegendentry{\scriptsize VD}
\addplot [thick, darkorange25512714, mark=*, mark size=2, mark options={solid}]
table {%
0.2 0.855759523809524
0.3 0.729255555555555
0.4 0.733016666666667
0.5 0.906291666666667
0.6 1.26299186507937
};
\addlegendentry{\scriptsize MF}
\addplot [thick, forestgreen4416044, mark=square*, mark size=2, mark options={solid}]
table {%
0.2 0.7826
0.3 0.6555
0.4 0.705866666666667
0.5 0.928495238095238
0.6 1.40686686507937
};
\addlegendentry{\scriptsize BC}
\addplot [thick, crimson2143940, mark=diamond*, mark size=2, mark options={solid}]
table {%
0.2 0.76525
0.3 0.615333333333333
0.4 0.626533333333333
0.5 0.767
0.6 1.127225
};
\addlegendentry{\scriptsize MS}
\addplot [thick, mediumpurple148103189, mark=triangle*, mark size=2, mark options={solid}]
table {%
0.2 0.745416666666667
0.3 0.6055
0.4 0.611833333333333
0.5 0.7405
0.6 1.02558333333333
};
\addlegendentry{\scriptsize FT}
\addplot [thick, sienna1408675, mark=star, mark size=2, mark options={solid}]
table {%
0.2 0.49
0.3 0.5295
0.4 0.5165
0.5 0.547
0.6 0.519833333333333
};
\addlegendentry{\scriptsize FT-rank}
\end{axis}

\end{tikzpicture}}
        \caption{$\sigma=0.1$}
        \label{fig:noisy1}
    \end{subfigure}
    \begin{subfigure}{.45\textwidth}
        \centering
        \scalebox{0.82}{% This file was created with tikzplotlib v0.10.1.
\begin{tikzpicture}

\definecolor{crimson2143940}{RGB}{254, 97, 0}
\definecolor{darkgray176}{RGB}{176,176,176}
\definecolor{darkorange25512714}{RGB}{120, 94, 240}
\definecolor{forestgreen4416044}{RGB}{	220, 38, 127}
\definecolor{lightgray204}{RGB}{204,204,204}
\definecolor{mediumpurple148103189}{RGB}{255, 176, 0}
\definecolor{sienna1408675}{RGB}{0,0,0}
\definecolor{steelblue31119180}{RGB}{100, 143, 255}

\begin{axis}[
legend cell align={left},
legend style={fill opacity=0.8, draw opacity=1, text opacity=1, draw=lightgray204,
  legend columns=2,},
tick align=outside,
tick pos=left,
x grid style={darkgray176},
xlabel={Approval radius $r$},
xmin=0.2, xmax=0.6,
xtick style={color=black},
y grid style={darkgray176},
ylabel={Average KT distance},
xtick distance=0.1,
ymin=1, ymax=2.3,
ytick style={color=black},
height=4.5cm, width=7.7cm
]
\addplot [thick, steelblue31119180, mark=x, mark size=2, mark options={solid}]
table {%
0.2 2.18270833333333
0.3 1.60515634920635
0.4 1.65519047619048
0.5 1.75440714285714
0.6 2.08504041040659
};
\addlegendentry{\scriptsize VD}
\addplot [thick, darkorange25512714, mark=*, mark size=2, mark options={solid}]
table {%
0.2 1.87355238095238
0.3 1.36265
0.4 1.36896666666667
0.5 1.48513333333333
0.6 1.827592002442
};
\addlegendentry{\scriptsize MF}
\addplot [thick, forestgreen4416044, mark=square*, mark size=2, mark options={solid}]
table {%
0.2 1.43395833333333
0.3 1.14441666666667
0.4 1.35648333333333
0.5 1.59848571428571
0.6 2.20001666666667
};
\addlegendentry{\scriptsize BC}
\addplot [thick, crimson2143940, mark=diamond*, mark size=2, mark options={solid}]
table {%
0.2 1.3969
0.3 1.09275
0.4 1.15133333333333
0.5 1.26408333333333
0.6 1.61666666666667
};
\addlegendentry{\scriptsize MS}
\addplot [thick, mediumpurple148103189, mark=triangle*, mark size=2, mark options={solid}]
table {%
0.2 1.36766666666667
0.3 1.05583333333333
0.4 1.0755
0.5 1.20816666666667
0.6 1.45166666666667
};
\addlegendentry{\scriptsize FT}
\addplot [thick, sienna1408675, mark=star, mark size=2, mark options={solid}]
table {%
0.2 1.369
0.3 1.40483333333333
0.4 1.53666666666667
0.5 1.3559
0.6 1.43766666666667
};
\addlegendentry{\scriptsize FT-rank}
\end{axis}

\end{tikzpicture}}
        \caption{$\sigma=0.2$}
        \label{fig:noisy2}
    \end{subfigure}
    \begin{subfigure}{.45\textwidth}
        \centering
        \scalebox{0.82}{% This file was created with tikzplotlib v0.10.1.
\begin{tikzpicture}

\definecolor{crimson2143940}{RGB}{254, 97, 0}
\definecolor{darkgray176}{RGB}{176,176,176}
\definecolor{darkorange25512714}{RGB}{120, 94, 240}
\definecolor{forestgreen4416044}{RGB}{	220, 38, 127}
\definecolor{lightgray204}{RGB}{204,204,204}
\definecolor{mediumpurple148103189}{RGB}{255, 176, 0}
\definecolor{sienna1408675}{RGB}{0,0,0}
\definecolor{steelblue31119180}{RGB}{100, 143, 255}

\begin{axis}[
legend cell align={left},
legend style={fill opacity=0.8, draw opacity=1, text opacity=1, draw=lightgray204,
  legend columns=2,},
tick align=outside,
tick pos=left,
x grid style={darkgray176},
xlabel={Approval radius $r$},
xmin=0.2, xmax=0.6,
xtick style={color=black},
y grid style={darkgray176},
ylabel={Average KT distance},
xtick distance=0.1,
ymin=1.7971234539072, ymax=4.24374080128205,
ytick style={color=black},
height=4.5cm, width=7.7cm
]
\addplot [thick, steelblue31119180, mark=x, mark size=2, mark options={solid}]
table {%
0.2 4.13253092185592
0.3 3.22310266106443
0.4 2.86981507936508
0.5 2.78420952380952
0.6 3.16934484126984
};
\addlegendentry{\scriptsize VD}
\addplot [thick, darkorange25512714, mark=*, mark size=2, mark options={solid}]
table {%
0.2 3.73041666666667
0.3 2.85762619047619
0.4 2.32291666666667
0.5 2.34119761904762
0.6 2.59296984126984
};
\addlegendentry{\scriptsize MF}
\addplot [thick, forestgreen4416044, mark=square*, mark size=2, mark options={solid}]
table {%
0.2 3.00263333333333
0.3 2.37018333333333
0.4 2.24043333333333
0.5 2.58329285714286
0.6 3.02299761904762
};
\addlegendentry{\scriptsize BC}
\addplot [thick, crimson2143940, mark=diamond*, mark size=2, mark options={solid}]
table {%
0.2 2.86219444444444
0.3 2.28631666666667
0.4 2.03221666666667
0.5 2.0889
0.6 2.392
};
\addlegendentry{\scriptsize MS}
\addplot [thick, mediumpurple148103189, mark=triangle*, mark size=2, mark options={solid}]
table {%
0.2 2.83941666666667
0.3 2.1763
0.4 1.90833333333333
0.5 1.92191666666667
0.6 2.19566666666667
};
\addlegendentry{\scriptsize FT}
\addplot [thick, sienna1408675, mark=star, mark size=2, mark options={solid}]
table {%
0.2 2.90186666666667
0.3 2.91086666666667
0.4 2.82875
0.5 2.82266666666667
0.6 2.88133333333333
};
\addlegendentry{\scriptsize FT-rank}
\end{axis}

\end{tikzpicture}}
        \caption{$\sigma=0.3$}
        \label{fig:noisy3}
    \end{subfigure}
    \begin{subfigure}{.45\textwidth}
        \centering
        \scalebox{0.82}{% This file was created with tikzplotlib v0.10.1.
\begin{tikzpicture}

\definecolor{crimson2143940}{RGB}{254, 97, 0}
\definecolor{darkgray176}{RGB}{176,176,176}
\definecolor{darkorange25512714}{RGB}{120, 94, 240}
\definecolor{forestgreen4416044}{RGB}{	220, 38, 127}
\definecolor{lightgray204}{RGB}{204,204,204}
\definecolor{mediumpurple148103189}{RGB}{255, 176, 0}
\definecolor{sienna1408675}{RGB}{0,0,0}
\definecolor{steelblue31119180}{RGB}{100, 143, 255}

\begin{axis}[
legend cell align={left},
legend style={fill opacity=0.8, draw opacity=1, text opacity=1, draw=lightgray204,
  legend columns=2,},
tick align=outside,
tick pos=left,
x grid style={darkgray176},
xlabel={Approval radius $r$},
xmin=0.2, xmax=0.6,
xtick style={color=black},
y grid style={darkgray176},
ylabel={Average KT distance},
xtick distance=0.1,
ymin=3, ymax=5.8,
ytick style={color=black},
height=4.5cm, width=7.7cm
]
\addplot [thick, steelblue31119180, mark=x, mark size=2, mark options={solid}]
table {%
0.2 5.85683118548119
0.3 5.17958964646465
0.4 4.5515623015873
0.5 4.28635189255189
0.6 4.38900528083028
};
\addlegendentry{\scriptsize VD}
\addplot [thick, darkorange25512714, mark=*, mark size=2, mark options={solid}]
table {%
0.2 5.68447063492064
0.3 4.96767579365079
0.4 4.09504087301587
0.5 3.71126428571429
0.6 3.73289206349206
};
\addlegendentry{\scriptsize MF}
\addplot [thick, forestgreen4416044, mark=square*, mark size=2, mark options={solid}]
table {%
0.2 5.02454722222222
0.3 4.3005
0.4 3.74755119047619
0.5 3.61196666666667
0.6 3.95501666666667
};
\addlegendentry{\scriptsize BC}
\addplot [thick, crimson2143940, mark=diamond*, mark size=2, mark options={solid}]
table {%
0.2 4.98440714285714
0.3 4.15078571428571
0.4 3.54246666666667
0.5 3.34168333333333
0.6 3.44997619047619
};
\addlegendentry{\scriptsize MS}
\addplot [thick, mediumpurple148103189, mark=triangle*, mark size=2, mark options={solid}]
table {%
0.2 5.0554
0.3 4.09516666666667
0.4 3.339
0.5 3.128
0.6 3.20508333333333
};
\addlegendentry{\scriptsize FT}
\addplot [thick, sienna1408675, mark=star, mark size=2, mark options={solid}]
table {%
0.2 4.03533333333333
0.3 4.44
0.4 4.173
0.5 4.10746666666667
0.6 4.43016666666667
};
\addlegendentry{\scriptsize FT-rank}
\end{axis}

\end{tikzpicture}}
        \caption{$\sigma=0.4$}
        \label{fig:noisy4}
    \end{subfigure}
    \caption{Evolution of the average KT distance between the axes returned by the rules and the actual axes for $r \in [0.2,0.6]$, averaged over 1\,000 random samples.}
    \label{fig:expe-synth}
\end{figure}

The positions $p(c)$ of the candidates describe a ground truth axis ${\axis} = c_1c_2 \dots c_m$ such that $p(c_1) \le p(c_2) \le \dots \le p(c_m)$. \Cref{fig:expe-synth} shows the Kendall-tau (KT) swap distance between the axes output by different rule results and the ground truth for $\sigma \in \{0.1,0.2,0.3,0.4\}$ and $r \in \{0.2,0.3,0.4,0.5,0.6\}$. We conducted experiments with $m=7$ candidates, $n=100$ voters and $1\,000$ random profiles for each set of parameters. We compared our approval-based axis rules to two ranking-based axis rules: VD-rank and FT-rank, studied by \citet{escoffier2021nearlysp}, which are defined analogously to our VD and FT rules.%
\footnote{We did not study the rules based on local swaps or global swaps (which are somewhat similar to our MS and MF rules, respectively), since these are expensive to compute. Indeed, the known algorithm for computing the cost of a single given axis for these rules takes $O(m^3)$ time \citep[Theorem~6.21]{erdelyi2017nearlysp}.}

We find that VD-rank is always far from the true axes (at distance 7--8, too much to fit in the chart), and that for most values of $\sigma$ and $r$, approval rules actually perform better than FT-rank, returning axes with a lower average KT distance to the ground truth. This is surprising, 
as intuitively rankings provide more information than approvals.
We note however that FT-rank is better than approval methods when $r$ is very small or very large, so many approval sets are of size $0$ or $1$ (or $m$), and thus provide no information on candidates' proximity. FT-rank is also slightly better when $\sigma$ is small, but in this case all approval rules also have very good performance, with their average KT distances all below $1$.
We also observe that for all parameter values, the axes returned by the rules using more information (e.g., FT) are closer to the ground truth axes than those returned by the rules using less information (e.g., VD).

\subsection{The French Presidential Election}
\label{sec:experiments:france}

We now present the results of our rules on two political datasets: the 2017 and 2022 edition of the online experiment \emph{Voter Autrement} conducted during the French presidential elections~\citep{voterautrement2018}.
In parallel to the actual elections, the participants were invited to express their opinions on candidates using various voting methods, including approval and ranking-based ones. 
This allows us to compare our axis rules for both settings. 
After data cleaning, for the 2017 [2022] dataset, we obtained approval preferences of 20\,076 voters  [1\,379 voters] and preference rankings of 5\,796 voters [412 voters] over 11 candidates [12 candidates]. Details on how the data was gathered and the experiments conducted can be found in \Cref{sec:appFrench}, together with our detailed results. There, we also explain how we reweighted votes to counteract response bias and to match the distribution of official election results.

\begin{table*}[t]
	\centering
	\scalebox{0.9}{\begin{tabular}{cc c c c c c c c c c c} \toprule
			Institute & & & & & &$\axis$ & & & & &  \\ \midrule
			BVA & $\party{LO}$ & $\party{NPA}$ &  $\party{LFI}$ & $\party{PS}$ &  $\party{EM}$ & $\party{LR}$ & $\party{DLF}$ & $\party{FN}$ &$\party{UPR}$ & $\party{SP}$ &$\party{R}$\\
			Opinionway & $\party{LO}$ & $\party{NPA}$ &  $\party{LFI}$ & $\party{PS}$ &  $\party{EM}$ & $\party{LR}$ & $\party{DLF}$ & $\party{FN}$ &$\party{UPR}$& $\party{SP}$ &$\party{R}$ \\
			IFOP & $\party{LO}$ & $\party{NPA}$ &  $\party{LFI}$ & $\party{PS}$ &  $\party{EM}$&$\party{R}$ & $\party{LR}$ & $\party{DLF}$ & $\party{FN}$ &$\party{UPR}$ & $\party{SP}$ \\
			IPSOS & $\party{LO}$ & $\party{NPA}$ &  $\party{LFI}$ & $\party{PS}$ &  $\party{EM}$ &$\party{R}$& $\party{LR}$ & $\party{DLF}$ & $\party{FN}$ & $\party{SP}$&$\party{UPR}$  \\
			Harris Interactive & $\party{LO}$ & $\party{NPA}$ &  $\party{LFI}$ & $\party{PS}$ &  $\party{EM}$&$\party{R}$ & $\party{LR}$ & $\party{DLF}$ &$\party{SP}$ &$\party{UPR}$&  $\party{FN}$  \\
			Odoxa & $\party{LO}$ & $\party{NPA}$ &  $\party{LFI}$ & $\party{PS}$ &  $\party{EM}$&$\party{R}$ & $\party{LR}$ & $\party{DLF}$ &$\party{UPR}$& $\party{SP}$ & $\party{FN}$  \\
			Elabe & $\party{NPA}$ & $\party{LO}$ &  $\party{LFI}$ & $\party{PS}$ &  $\party{EM}$ &$\party{R}$ & $\party{LR}$&$\party{UPR}$ & $\party{DLF}$ & $\party{FN}$  & $\party{SP}$\\ \bottomrule
	\end{tabular}}
	\caption{Axes used by polling institutes for the 2017 French presidential election}
	\label{tab:appFrenchPolls2017}
\end{table*}

\begin{table*}[t]
    \centering
    \scalebox{0.9}{\begin{tabular}{cc c c c c c c c c c ccc} \toprule
        Rule & & & & & &$\axis$ & & & & &  & Min KT&Avg KT \\ \midrule
       VD  & $\party{R}$ & $\party{LO}$ & $\party{NPA}$ & $\party{LFI}$ & $\party{PS}$ & $\party{EM}$ & $\party{LR}$ & $\party{DLF}$ & $\party{FN}$& $\party{UPR}$ &$\party{SP}$ & 5 & 7.71  \\
       MF  & $\party{LO}$ & $\party{NPA}$ & $\party{LFI}$ & $\party{PS}$ & $\party{EM}$ & $\party{LR}$ & $\party{DLF}$ & $\party{FN}$& $\party{UPR}$ & $\party{R}$ &$\party{SP}$ & \textbf{1} & 4.43 \\
       BC  & $\party{LO}$ & $\party{NPA}$ & $\party{LFI}$ & $\party{PS}$ & $\party{EM}$ & $\party{LR}$ & $\party{DLF}$ & $\party{FN}$& $\party{R}$ &$\party{UPR}$ & $\party{SP}$ & 2 & 4.0\\
       MS  & $\party{LO}$ & $\party{NPA}$ & $\party{LFI}$ & $\party{PS}$ & $\party{EM}$ & $\party{LR}$ & $\party{DLF}$ & $\party{FN}$& $\party{R}$ &$\party{UPR}$ & $\party{SP}$ & 2 & 4.0 \\
       FT  & $\party{LO}$ & $\party{NPA}$ & $\party{PS}$ & $\party{LFI}$ &  $\party{EM}$& $\party{R}$ & $\party{LR}$ & $\party{DLF}$ & $\party{FN}$ &$\party{UPR}$ & $\party{SP}$ & \textbf{1} & \textbf{3.71}  \\ \midrule
       VD-rank & $\party{FN}$ & $\party{DLF}$ & $\party{R}$ &  $\party{LO}$ &  $\party{NPA}$ & $\party{LFI}$ & $\party{PS}$ & $\party{EM}$ &  $\party{SP}$ &   $\party{UPR}$ & $\party{LR}$ & 22 & 24.0  \\
       FT-rank & $\party{LO}$ & $\party{NPA}$ &$\party{R}$ &  $\party{LFI}$ & $\party{PS}$ &  $\party{EM}$& $\party{LR}$ & $\party{DLF}$ & $\party{FN}$ &$\party{UPR}$ & $\party{SP}$ & 3 & 5.71\\ \bottomrule
    \end{tabular}}
    \caption{Optimal axis of each rule for the 2017 French presidential election}
    \label{tab:appFrenchRes2017}
\end{table*}

In \Cref{tab:appFrenchRes2017}, we present the axes returned by all the tested rules, as well as the minimum and average Kendall-tau distance to the axes used by the main 7 polling institutes in France, which are displayed in \Cref{tab:appFrenchPolls2017}. (The corresponding results for the 2022 are shown in \Cref{sec:appFrench}.)
Regarding approval rules, we note that they all returned very similar axes.
They mostly differ on the position of less popular candidates (often placed at one of the extremes), and the relative order of candidates within their ideological subgroup (e.g., left-wing candidates). We computed the KT distance between the axes returned by our rules and the ones of the polling institutes. All rules return an axis that has a KT distance of less than $5$ to at least one polling institute axis (while the worst possible KT distance are 27 and 33 for $m = 11$ and $12$). For instance, the ordering obtained with FT is very similar to the one of the \emph{Ipsos institute}:
\begin{align*}
    \text{FT: }&  
\party{LO}, 
\party{NPA}, 
\party{PS}, 
\party{LFI}, 
\party{EM}, 
\party{R}, 
\party{LR}, 
\party{DLF}, 
\party{FN}, 
\party{UPR}, 
\party{SP} \\
    \text{Ipsos: }& 
\party{LO}, 
\party{NPA}, 
\party{LFI}, 
\party{PS}, 
\party{EM}, 
\party{R}, 
\party{LR}, 
\party{DLF}, 
\party{FN}, 
\party{SP}, 
\party{UPR}
\end{align*}

The KT distance between them is $2$. Small parties ($\party{LO}$, $\party{NPA}$, $\party{R}$, $\party{UPR}$, $\party{SP}$) are displayed using small font. Note that all but one of them are placed at the extremes.

Regarding ranking-based methods, the quality of the axes returned by FT-rank seems comparable to the axes returned by approval rules. Again, the VD-rank axes were much less convincing.
This corroborates other observations in the literature. For instance, \citet{sui2013voterdeletion} ran experiments on 2002 Irish General Election data and found that the VD-rank axis only fit 0.4\%--2.9\% of voters. \citet{escoffier2021nearlysp} ran experiments on a similar French presidential election dataset and also observed that the optimal axis found using VD-rank was very different from the orderings discussed in French media. In our experiments, the optimal VD-rank axes only cover less then 4\% of voters. For comparison, the approval version of VD returned axes covering more than 60\% of voters.

Finally, we observe that all rules violate the heredity property on our dataset. Removing even the least approved candidate could change the returned axis. However, these changes are marginal, like a less popular candidate being pushed towards an extreme or two left-wing candidates being inverted.

\subsection{Supreme Court of the United States}
\label{sec:SCOTUS}
Finally, we used our rules to obtain an ideological ordering of the 9 justices of the Supreme Court of the United States. The dataset is based on the opinions authored and joined by the justices, derived from the Supreme Court Database \citep{spaeth2023}. Each opinion, concurrence, or dissent becomes a ballot ``approving'' the justices that joined in it. The intuition is that justices joining the same opinion share an ideology so should be placed close together. See \Cref{sec:appScotus} for details and results. 

The problem of ordering the justices has been extensively studied; the standard method used by political analysts is the \emph{Martin-Quinn} (MQ) method, which uses a dynamic item response theory model \citep{martinquinn2002dynamic}.
A limitation of this model is that it can only use the vote data (whether a justice agreed with the majority or not), while our model can use more fine-grained data from which opinions were joined.
We compare the axes returned by our rules for 65 terms between 1946 and 2021, removing the years having more than 9 justices involved (e.g., if one is replaced mid-term). 

\begin{table*}[t]
    \centering
    \begin{tabular}{ccc} \toprule
      Rule   & Avg KT & Correct Median \\ \midrule
     VD    & 4.94 & 53.8 $\%$ \\
     MF    &  4.22&58.5 $\%$\\
     BC    & 3.68 & 56.9 $\%$\\
     MS    &  3.55&64.6 $\%$ \\
     FT    &  \textbf{3.43}& \textbf{66.2} $\%$\\\bottomrule
    \end{tabular}
    \caption{Average Kendall-tau distance to MQ axis, and $\%$ of time the axis has the same median candidate than the MQ axis, averaged over 65 terms.}
    \label{tab:appSCOTUS}
\end{table*}

\Cref{tab:appSCOTUS} shows the average KT distance of the axes returned by our rules to the Martin-Quinn axis. We see that these distances are on average quite low (noting that the worst possible KT distance for $m=9$ is 18). Moreover, we observe that the FT rule comes closest, while the VD rule is relatively far away. We also checked how often the axes computed by our rules agreed with the Martin-Quinn axis on which justice is placed in the median position. This is of particular interest since the median justice tends to be pivotal. All rules agreed with MQ on who was the median justice in more than half of the years. The FT rule agrees most frequently, choosing the same median justice in $66\%$ of terms. For future work, we see potential in adapting our rules to obtain methods perhaps more interesting than the Martin-Quinn method (as they will satisfy axiomatic properties).

\section{Future Work}  \label{sec:conclusion}

There are many promising directions for future work, such as considering methods that output other types of structures, like circular axes (in which the first and last candidates on the axis are next to each other) or embeddings into multiple dimensions, or introducing metric distances between candidates on the axis. An axiomatic approach could provide novel insights for all these problems. Moreover, the methods we present not only return a set of optimal axes, but also their ``cost'', which provides an indicator of how close a profiles is to be linear. One could try to analyze these methods as rules measuring the degrees of linearity of approval profiles. In addition, one can further investigate the interpretation of scoring rules as maximum likelihood estimators, mentioned at the end of \Cref{sec:preliminaries}. In particular, one could develop noise models that give rise to the rules that we study, or develop new natural noise models and (axiomatically) study the rules that they induce, similar to the work of \citet[Section 4.5]{tydrichova2023structural} for rankings.

Technically, several open questions remain. It would be interesting to obtain an axiomatic characterization of the class of scoring rules using the reinforcement axiom, though this is made challenging by the neutrality axiom being quite weak in our setting. 
It would also be useful to design polynomial-time computable rules that produce good outputs, to be able to deal with many candidates. Greedy versions of our rules are a natural starting point, but maybe better techniques exist.

 \section*{Acknowledgments}
This work was funded in part by the French government under management of Agence Nationale de la Recherche as part of the ``Investissements d’avenir'' program, reference ANR-19-P3IA-0001 (PRAIRIE 3IA Institute). The authors would like to thank Isabelle Lebon for her help in the early stages of this project, Jérôme Lang for his helpful comments, and the anonymous reviewers at IJCAI 2024 for helping improve the presentation.
 
\bibliographystyle{ACM-Reference-Format}

\begin{thebibliography}{41}

%
%

\ifx \showCODEN    \undefined \def \showCODEN     #1{\unskip}     \fi
\ifx \showDOI      \undefined \def \showDOI       #1{#1}\fi
\ifx \showISBNx    \undefined \def \showISBNx     #1{\unskip}     \fi
\ifx \showISBNxiii \undefined \def \showISBNxiii  #1{\unskip}     \fi
\ifx \showISSN     \undefined \def \showISSN      #1{\unskip}     \fi
\ifx \showLCCN     \undefined \def \showLCCN      #1{\unskip}     \fi
\ifx \shownote     \undefined \def \shownote      #1{#1}          \fi
\ifx \showarticletitle \undefined \def \showarticletitle #1{#1}   \fi
\ifx \showURL      \undefined \def \showURL       {\relax}        \fi
%
%
\providecommand\bibfield[2]{#2}
\providecommand\bibinfo[2]{#2}
\providecommand\natexlab[1]{#1}
\providecommand\showeprint[2][]{arXiv:#2}

\bibitem[Baujard and Lebon(2022)]%
        {france2017}
\bibfield{author}{\bibinfo{person}{Antoinette Baujard} {and}
  \bibinfo{person}{Isabelle Lebon}.} \bibinfo{year}{2022}\natexlab{}.
\newblock \showarticletitle{Retelling the story of the 2017 {French}
  presidential election: The contribution of approval voting}.
\newblock \bibinfo{journal}{\emph{Homo Oeconomicus}} (\bibinfo{year}{2022}),
  \bibinfo{pages}{1--22}.
\newblock
\urldef\tempurl%
\url{https://doi.org/10.1007/s41412-022-00134-7}
\showDOI{\tempurl}


\bibitem[Baxter(2003)]%
        {baxter2003statistics}
\bibfield{author}{\bibinfo{person}{Michael~J. Baxter}.}
  \bibinfo{year}{2003}\natexlab{}.
\newblock \bibinfo{booktitle}{\emph{Statistics in archaeology}}.
\newblock \bibinfo{publisher}{Arnold}.
\newblock


\bibitem[Black(1948)]%
        {black1948rationale}
\bibfield{author}{\bibinfo{person}{Duncan Black}.}
  \bibinfo{year}{1948}\natexlab{}.
\newblock \showarticletitle{On the rationale of group decision-making}.
\newblock \bibinfo{journal}{\emph{The Journal of Political Economy}}
  \bibinfo{volume}{56}, \bibinfo{number}{1} (\bibinfo{year}{1948}),
  \bibinfo{pages}{23--34}.
\newblock
\urldef\tempurl%
\url{https://doi.org/10.1086/256633}
\showDOI{\tempurl}


\bibitem[Booth(1975)]%
        {booth1975pq}
\bibfield{author}{\bibinfo{person}{Kellogg~Speed Booth}.}
  \bibinfo{year}{1975}\natexlab{}.
\newblock \emph{\bibinfo{title}{PQ-tree algorithms.}}
\newblock \bibinfo{thesistype}{Ph.\,D. Dissertation}.
  \bibinfo{school}{University of California, Berkeley, and Lawrence Livermore
  Laboratory}.
\newblock
\urldef\tempurl%
\url{https://dominik-peters.de/archive/booth1975.pdf}
\showURL{%
\tempurl}


\bibitem[Booth and Lueker(1976)]%
        {booth_lueker_C1P}
\bibfield{author}{\bibinfo{person}{Kellogg~S. Booth} {and}
  \bibinfo{person}{George~S. Lueker}.} \bibinfo{year}{1976}\natexlab{}.
\newblock \showarticletitle{Testing for the consecutive ones property, interval
  graphs, and graph planarity using {PQ}-tree algorithms}.
\newblock \bibinfo{journal}{\emph{J. Comput. System Sci.}}
  \bibinfo{volume}{13}, \bibinfo{number}{3} (\bibinfo{year}{1976}),
  \bibinfo{pages}{335--379}.
\newblock
\urldef\tempurl%
\url{https://doi.org/10.1016/S0022-0000(76)80045-1}
\showDOI{\tempurl}


\bibitem[Bouveret et~al\mbox{.}(2018)]%
        {voterautrement2018}
\bibfield{author}{\bibinfo{person}{Sylvain Bouveret}, \bibinfo{person}{Renaud
  Blanch}, \bibinfo{person}{Antoinette Baujard}, \bibinfo{person}{Fran\c{c}ois
  Durand}, \bibinfo{person}{Herrade Igersheim}, \bibinfo{person}{J\'er\^ome
  Lang}, \bibinfo{person}{Annick Laruelle}, \bibinfo{person}{Jean-Fran\c{c}ois
  Laslier}, \bibinfo{person}{Isabelle Lebon}, {and} \bibinfo{person}{Vincent
  Merlin}.} \bibinfo{year}{2018}\natexlab{}.
\newblock \bibinfo{title}{Voter Autrement 2017 - Online Experiment}.
\newblock \bibinfo{howpublished}{Dataset and companion article on Zenodo}.
\newblock
\urldef\tempurl%
\url{https://doi.org/10.5281/zenodo.1199545}
\showDOI{\tempurl}


\bibitem[Brandl et~al\mbox{.}(2016)]%
        {brandl2016consistent}
\bibfield{author}{\bibinfo{person}{Florian Brandl}, \bibinfo{person}{Felix
  Brandt}, {and} \bibinfo{person}{Hans~Georg Seedig}.}
  \bibinfo{year}{2016}\natexlab{}.
\newblock \showarticletitle{Consistent Probabilistic Social Choice}.
\newblock \bibinfo{journal}{\emph{Econometrica}} \bibinfo{volume}{84},
  \bibinfo{number}{5} (\bibinfo{year}{2016}), \bibinfo{pages}{1839--1880}.
\newblock
\urldef\tempurl%
\url{https://doi.org/10.3982/ECTA13337}
\showDOI{\tempurl}


\bibitem[Bredereck et~al\mbox{.}(2016)]%
        {bredereck2016there}
\bibfield{author}{\bibinfo{person}{Robert Bredereck}, \bibinfo{person}{Jiehua
  Chen}, {and} \bibinfo{person}{Gerhard~J. Woeginger}.}
  \bibinfo{year}{2016}\natexlab{}.
\newblock \showarticletitle{Are there any nicely structured preference profiles
  nearby?}
\newblock \bibinfo{journal}{\emph{Mathematical Social Sciences}}
  \bibinfo{volume}{79} (\bibinfo{year}{2016}), \bibinfo{pages}{61--73}.
\newblock
\urldef\tempurl%
\url{https://doi.org/10.1016/j.mathsocsci.2015.11.002}
\showDOI{\tempurl}


\bibitem[Ceron and Gonzalez(2021)]%
        {ceron2021approval}
\bibfield{author}{\bibinfo{person}{Federica Ceron} {and}
  \bibinfo{person}{St{\'e}phane Gonzalez}.} \bibinfo{year}{2021}\natexlab{}.
\newblock \showarticletitle{Approval voting without ballot restrictions}.
\newblock \bibinfo{journal}{\emph{Theoretical Economics}} \bibinfo{volume}{16},
  \bibinfo{number}{3} (\bibinfo{year}{2021}), \bibinfo{pages}{759--775}.
\newblock
\urldef\tempurl%
\url{https://doi.org/10.3982/te4087}
\showDOI{\tempurl}


\bibitem[Chauve et~al\mbox{.}(2009)]%
        {chauve2009gapped}
\bibfield{author}{\bibinfo{person}{Cedric Chauve}, \bibinfo{person}{J{\'a}n
  Ma{\v{n}}uch}, {and} \bibinfo{person}{Murray Patterson}.}
  \bibinfo{year}{2009}\natexlab{}.
\newblock \showarticletitle{On the gapped consecutive-ones property}.
\newblock \bibinfo{journal}{\emph{Electronic Notes in Discrete Mathematics}}
  \bibinfo{volume}{34} (\bibinfo{year}{2009}), \bibinfo{pages}{121--125}.
\newblock
\urldef\tempurl%
\url{https://doi.org/10.1016/j.endm.2009.07.020}
\showDOI{\tempurl}


\bibitem[Chen et~al\mbox{.}(2023)]%
        {chen2023efficient}
\bibfield{author}{\bibinfo{person}{Jiehua Chen}, \bibinfo{person}{Christian
  Hatschka}, {and} \bibinfo{person}{Sofia Simola}.}
  \bibinfo{year}{2023}\natexlab{}.
\newblock \showarticletitle{Efficient algorithms for {Monroe} and {CC} rules in
  multi-winner elections with (nearly) structured preferences}. In
  \bibinfo{booktitle}{\emph{Proceedings of the 26th European Conference on
  Artificial Intelligence (ECAI)}}. \bibinfo{pages}{397--404}.
\newblock
\urldef\tempurl%
\url{https://doi.org/10.3233/FAIA230296}
\showDOI{\tempurl}


\bibitem[Conitzer et~al\mbox{.}(2009)]%
        {conitzer2009preference}
\bibfield{author}{\bibinfo{person}{Vincent Conitzer}, \bibinfo{person}{Matthew
  Rognlie}, {and} \bibinfo{person}{Lirong Xia}.}
  \bibinfo{year}{2009}\natexlab{}.
\newblock \showarticletitle{Preference functions that score rankings and
  maximum likelihood estimation}. In \bibinfo{booktitle}{\emph{Proceedings of
  the 21st International Joint Conference on Artificial Intelligence (IJCAI)}}.
  \bibinfo{pages}{109--115}.
\newblock
\urldef\tempurl%
\url{https://www.cs.cmu.edu/~conitzer/preferenceIJCAI09.pdf}
\showURL{%
\tempurl}


\bibitem[Dietrich and List(2010)]%
        {dietrich2015}
\bibfield{author}{\bibinfo{person}{Franz Dietrich} {and}
  \bibinfo{person}{Christian List}.} \bibinfo{year}{2010}\natexlab{}.
\newblock \showarticletitle{Majority voting on restricted domains}.
\newblock \bibinfo{journal}{\emph{Journal of Economic Theory}}
  \bibinfo{volume}{145}, \bibinfo{number}{2} (\bibinfo{year}{2010}),
  \bibinfo{pages}{512--543}.
\newblock
\urldef\tempurl%
\url{https://doi.org/10.1016/j.jet.2010.01.003}
\showDOI{\tempurl}


\bibitem[Dom(2009)]%
        {consecutiveones}
\bibfield{author}{\bibinfo{person}{Michael Dom}.}
  \bibinfo{year}{2009}\natexlab{}.
\newblock \showarticletitle{Algorithmic aspects of the consecutive-ones
  property}.
\newblock \bibinfo{journal}{\emph{Bulletin of the European Association for
  Theoretical Computer Science}}  \bibinfo{volume}{98} (\bibinfo{year}{2009}),
  \bibinfo{pages}{27--59}.
\newblock
\urldef\tempurl%
\url{http://www.mdom.de/fsujena/publications/Dom09_s.pdf}
\showURL{%
\tempurl}


\bibitem[Dom et~al\mbox{.}(2010)]%
        {dom2010approximation}
\bibfield{author}{\bibinfo{person}{Michael Dom}, \bibinfo{person}{Jiong Guo},
  {and} \bibinfo{person}{Rolf Niedermeier}.} \bibinfo{year}{2010}\natexlab{}.
\newblock \showarticletitle{Approximation and fixed-parameter algorithms for
  consecutive ones submatrix problems}.
\newblock \bibinfo{journal}{\emph{J. Comput. System Sci.}}
  \bibinfo{volume}{76}, \bibinfo{number}{3-4} (\bibinfo{year}{2010}),
  \bibinfo{pages}{204--221}.
\newblock
\urldef\tempurl%
\url{https://doi.org/10.1016/j.jcss.2009.07.001}
\showDOI{\tempurl}


\bibitem[Elkind and Lackner(2014)]%
        {elkind2014detecting}
\bibfield{author}{\bibinfo{person}{Edith Elkind} {and} \bibinfo{person}{Martin
  Lackner}.} \bibinfo{year}{2014}\natexlab{}.
\newblock \showarticletitle{On detecting nearly structured preference
  profiles}. In \bibinfo{booktitle}{\emph{Proceedings of the 28th AAAI
  Conference on Artificial Intelligence (AAAI)}}. \bibinfo{pages}{661--667}.
\newblock
\urldef\tempurl%
\url{https://doi.org/10.1609/aaai.v28i1.8823}
\showDOI{\tempurl}


\bibitem[Elkind and Lackner(2015)]%
        {elkind2015structure}
\bibfield{author}{\bibinfo{person}{Edith Elkind} {and} \bibinfo{person}{Martin
  Lackner}.} \bibinfo{year}{2015}\natexlab{}.
\newblock \showarticletitle{Structure in dichotomous preferences}. In
  \bibinfo{booktitle}{\emph{Proceedings of the 24th International Joint
  Conference on Artificial Intelligence (IJCAI)}}. \bibinfo{pages}{2019--2025}.
\newblock
\urldef\tempurl%
\url{https://arxiv.org/pdf/1505.00341.pdf}
\showURL{%
\tempurl}


\bibitem[Elkind et~al\mbox{.}(2017)]%
        {elkind2017structured}
\bibfield{author}{\bibinfo{person}{Edith Elkind}, \bibinfo{person}{Martin
  Lackner}, {and} \bibinfo{person}{Dominik Peters}.}
  \bibinfo{year}{2017}\natexlab{}.
\newblock \showarticletitle{Structured Preferences}.
\newblock In \bibinfo{booktitle}{\emph{Trends in Computational Social Choice}},
  \bibfield{editor}{\bibinfo{person}{Ulle Endriss}} (Ed.).
  \bibinfo{publisher}{AI Access}, Chapter~10, \bibinfo{pages}{187--207}.
\newblock
\urldef\tempurl%
\url{https://archive.illc.uva.nl/COST-IC1205/BookDocs/Chapters/TrendsCOMSOC-10.pdf}
\showURL{%
\tempurl}


\bibitem[Elkind et~al\mbox{.}(2022)]%
        {elkind2022preference}
\bibfield{author}{\bibinfo{person}{Edith Elkind}, \bibinfo{person}{Martin
  Lackner}, {and} \bibinfo{person}{Dominik Peters}.}
  \bibinfo{year}{2022}\natexlab{}.
\newblock \bibinfo{title}{Preference restrictions in computational social
  choice: A survey}.
\newblock
\newblock
\showeprint[arxiv]{2205.09092}~[cs.GT]
\urldef\tempurl%
\url{https://arxiv.org/abs/2205.09092}
\showURL{%
\tempurl}


\bibitem[Erd{\'e}lyi et~al\mbox{.}(2017)]%
        {erdelyi2017nearlysp}
\bibfield{author}{\bibinfo{person}{G{\'a}bor Erd{\'e}lyi},
  \bibinfo{person}{Martin Lackner}, {and} \bibinfo{person}{Andreas Pfandler}.}
  \bibinfo{year}{2017}\natexlab{}.
\newblock \showarticletitle{Computational aspects of nearly single-peaked
  electorates}.
\newblock \bibinfo{journal}{\emph{Journal of Artificial Intelligence Research
  (JAIR)}}  \bibinfo{volume}{58} (\bibinfo{year}{2017}),
  \bibinfo{pages}{297--337}.
\newblock
\urldef\tempurl%
\url{https://doi.org/10.1613/jair.5210}
\showDOI{\tempurl}


\bibitem[Escoffier et~al\mbox{.}(2021)]%
        {escoffier2021nearlysp}
\bibfield{author}{\bibinfo{person}{Bruno Escoffier}, \bibinfo{person}{Olivier
  Spanjaard}, {and} \bibinfo{person}{Magdal{\'e}na Tydrichov{\'a}}.}
  \bibinfo{year}{2021}\natexlab{}.
\newblock \showarticletitle{Measuring nearly single-peakedness of an
  electorate: Some new insights}. In \bibinfo{booktitle}{\emph{Proceeding of
  the 7th International Conference on Algorithmic Decision Theory (ADT)}}.
  \bibinfo{pages}{19--34}.
\newblock
\urldef\tempurl%
\url{https://doi.org/10.1007/978-3-030-87756-9_2}
\showDOI{\tempurl}


\bibitem[Faliszewski et~al\mbox{.}(2011)]%
        {faliszewski2011shield}
\bibfield{author}{\bibinfo{person}{Piotr Faliszewski}, \bibinfo{person}{Edith
  Hemaspaandra}, \bibinfo{person}{Lane Hemaspaandra}, {and}
  \bibinfo{person}{J{\"o}rg Rothe}.} \bibinfo{year}{2011}\natexlab{}.
\newblock \showarticletitle{The shield that never was: {S}ocieties with
  single-peaked preferences are more open to manipulation and control}.
\newblock \bibinfo{journal}{\emph{Information and Computation}}
  \bibinfo{volume}{209}, \bibinfo{number}{2} (\bibinfo{year}{2011}),
  \bibinfo{pages}{89--107}.
\newblock
\urldef\tempurl%
\url{https://doi.org/10.1016/j.ic.2010.09.001}
\showDOI{\tempurl}


\bibitem[Faliszewski et~al\mbox{.}(2014)]%
        {faliszewski2014voterdeletion}
\bibfield{author}{\bibinfo{person}{Piotr Faliszewski}, \bibinfo{person}{Edith
  Hemaspaandra}, {and} \bibinfo{person}{Lane~A. Hemaspaandra}.}
  \bibinfo{year}{2014}\natexlab{}.
\newblock \showarticletitle{The complexity of manipulative attacks in nearly
  single-peaked electorates}.
\newblock \bibinfo{journal}{\emph{Artificial Intelligence}}
  \bibinfo{volume}{207} (\bibinfo{year}{2014}), \bibinfo{pages}{69--99}.
\newblock
\urldef\tempurl%
\url{https://doi.org/10.1016/j.artint.2013.11.004}
\showDOI{\tempurl}


\bibitem[Garey and Johnson(1979)]%
        {garey1979computers}
\bibfield{author}{\bibinfo{person}{Michael~R. Garey} {and}
  \bibinfo{person}{David~S. Johnson}.} \bibinfo{year}{1979}\natexlab{}.
\newblock \bibinfo{booktitle}{\emph{Computers and Intractability: A Guide to
  the Theory of {NP}-completeness}}.
\newblock \bibinfo{publisher}{W. H. Freeman and Company}.
\newblock
\showISBNx{9780716710448}


\bibitem[Garey et~al\mbox{.}(1976)]%
        {Garey1976}
\bibfield{author}{\bibinfo{person}{Michael~R. Garey}, \bibinfo{person}{David~S.
  Johnson}, {and} \bibinfo{person}{Larry Stockmeyer}.}
  \bibinfo{year}{1976}\natexlab{}.
\newblock \showarticletitle{Some simplified {NP}-complete graph problems}.
\newblock \bibinfo{journal}{\emph{Theoretical Computer Science}}
  \bibinfo{volume}{1}, \bibinfo{number}{3} (\bibinfo{year}{1976}),
  \bibinfo{pages}{237--267}.
\newblock
\urldef\tempurl%
\url{https://doi.org/10.1016/0304-3975(76)90059-1}
\showDOI{\tempurl}


\bibitem[Hajiaghayi and Ganjali(2002)]%
        {hajiaghayi2002note}
\bibfield{author}{\bibinfo{person}{Mohammad~Taghi Hajiaghayi} {and}
  \bibinfo{person}{Yashar Ganjali}.} \bibinfo{year}{2002}\natexlab{}.
\newblock \showarticletitle{A note on the consecutive ones submatrix problem}.
\newblock \bibinfo{journal}{\emph{Inform. Process. Lett.}}
  \bibinfo{volume}{83}, \bibinfo{number}{3} (\bibinfo{year}{2002}),
  \bibinfo{pages}{163--166}.
\newblock
\urldef\tempurl%
\url{https://doi.org/10.1016/s0020-0190(01)00325-8}
\showDOI{\tempurl}


\bibitem[Karp(1972)]%
        {Karp1972}
\bibfield{author}{\bibinfo{person}{Richard~M. Karp}.}
  \bibinfo{year}{1972}\natexlab{}.
\newblock \showarticletitle{Reducibility among Combinatorial Problems}. In
  \bibinfo{booktitle}{\emph{Complexity of Computer Computations}}.
  \bibinfo{pages}{85--103}.
\newblock
\urldef\tempurl%
\url{https://doi.org/10.1007/978-1-4684-2001-2_9}
\showDOI{\tempurl}


\bibitem[Lebon et~al\mbox{.}(2017)]%
        {france2012}
\bibfield{author}{\bibinfo{person}{Isabelle Lebon}, \bibinfo{person}{Antoinette
  Baujard}, \bibinfo{person}{Frédéric Gavrel}, \bibinfo{person}{Herrade
  Igersheim}, {and} \bibinfo{person}{Laslier Jean-François}.}
  \bibinfo{year}{2017}\natexlab{}.
\newblock \showarticletitle{What approval voting reveals about the preferences
  of {French} voters}.
\newblock \bibinfo{journal}{\emph{Revue économique}}  \bibinfo{volume}{68}
  (\bibinfo{year}{2017}), \bibinfo{pages}{1063--1076}.
\newblock
\urldef\tempurl%
\url{https://doi.org/10.3917/reco.pr2.0084}
\showDOI{\tempurl}


\bibitem[Martin and Quinn(2002)]%
        {martinquinn2002dynamic}
\bibfield{author}{\bibinfo{person}{Andrew~D. Martin} {and}
  \bibinfo{person}{Kevin~M. Quinn}.} \bibinfo{year}{2002}\natexlab{}.
\newblock \showarticletitle{Dynamic ideal point estimation via {Markov} chain
  {Monte} {Carlo} for the {US} {Supreme} {Court}, 1953--1999}.
\newblock \bibinfo{journal}{\emph{Political analysis}} \bibinfo{volume}{10},
  \bibinfo{number}{2} (\bibinfo{year}{2002}), \bibinfo{pages}{134--153}.
\newblock
\urldef\tempurl%
\url{https://www.jstor.org/stable/25791672}
\showURL{%
\tempurl}


\bibitem[Misra et~al\mbox{.}(2017)]%
        {misra2017complexity}
\bibfield{author}{\bibinfo{person}{Neeldhara Misra}, \bibinfo{person}{Chinmay
  Sonar}, {and} \bibinfo{person}{P.~R. Vaidyanathan}.}
  \bibinfo{year}{2017}\natexlab{}.
\newblock \showarticletitle{On the complexity of {Chamberlin--Courant} on
  almost structured profiles}. In \bibinfo{booktitle}{\emph{Proceedings of the
  5th International Conference on Algorithmic Decision Theory (ADT)}}.
  \bibinfo{pages}{124--138}.
\newblock
\urldef\tempurl%
\url{https://doi.org/10.1007/978-3-319-67504-6_9}
\showDOI{\tempurl}


\bibitem[Myerson(1995)]%
        {Myer95b}
\bibfield{author}{\bibinfo{person}{Roger~B. Myerson}.}
  \bibinfo{year}{1995}\natexlab{}.
\newblock \showarticletitle{Axiomatic derivation of scoring rules without the
  ordering assumption}.
\newblock \bibinfo{journal}{\emph{Social Choice and Welfare}}
  \bibinfo{volume}{12}, \bibinfo{number}{1} (\bibinfo{year}{1995}),
  \bibinfo{pages}{59--74}.
\newblock
\urldef\tempurl%
\url{https://doi.org/10.1007/BF00182193}
\showDOI{\tempurl}


\bibitem[Narayanaswamy and Subashini(2015)]%
        {narayanaswamy2015obtaining}
\bibfield{author}{\bibinfo{person}{N.~S. Narayanaswamy} {and}
  \bibinfo{person}{R. Subashini}.} \bibinfo{year}{2015}\natexlab{}.
\newblock \showarticletitle{Obtaining matrices with the consecutive ones
  property by row deletions}.
\newblock \bibinfo{journal}{\emph{Algorithmica}}  \bibinfo{volume}{71}
  (\bibinfo{year}{2015}), \bibinfo{pages}{758--773}.
\newblock
\urldef\tempurl%
\url{https://doi.org/10.1007/s00453-014-9925-1}
\showDOI{\tempurl}


\bibitem[Petrie(1899)]%
        {petrie1899sequences}
\bibfield{author}{\bibinfo{person}{W.~M.~Flinders Petrie}.}
  \bibinfo{year}{1899}\natexlab{}.
\newblock \showarticletitle{Sequences in prehistoric remains}.
\newblock \bibinfo{journal}{\emph{Journal of the Anthropological Institute of
  Great Britain and Ireland}} (\bibinfo{year}{1899}),
  \bibinfo{pages}{295--301}.
\newblock
\urldef\tempurl%
\url{https://doi.org/10.2307/2843012}
\showDOI{\tempurl}


\bibitem[Pivato(2013)]%
        {Piva13a}
\bibfield{author}{\bibinfo{person}{Marcus Pivato}.}
  \bibinfo{year}{2013}\natexlab{}.
\newblock \showarticletitle{Variable-population voting rules}.
\newblock \bibinfo{journal}{\emph{Journal of Mathematical Economics}}
  \bibinfo{volume}{49}, \bibinfo{number}{3} (\bibinfo{year}{2013}),
  \bibinfo{pages}{210--221}.
\newblock
\urldef\tempurl%
\url{https://doi.org/10.1016/j.jmateco.2013.02.001}
\showDOI{\tempurl}


\bibitem[Spaeth et~al\mbox{.}(2023)]%
        {spaeth2023}
\bibfield{author}{\bibinfo{person}{Harold~J. Spaeth}, \bibinfo{person}{Lee
  Epstein}, \bibinfo{person}{Andrew~D. Martin}, \bibinfo{person}{Jeffrey~A.
  Segal}, \bibinfo{person}{Theodore~J. Ruger}, {and} \bibinfo{person}{Sara~C.
  Benesh}.} \bibinfo{year}{2023}\natexlab{}.
\newblock \bibinfo{title}{2023 {Supreme} {Court} {Database}, {Version} 2023
  {Release} 01}.
\newblock \bibinfo{howpublished}{\url{http://supremecourtdatabase.org}}.
\newblock


\bibitem[Sui et~al\mbox{.}(2013)]%
        {sui2013voterdeletion}
\bibfield{author}{\bibinfo{person}{Xin Sui}, \bibinfo{person}{Alex
  Francois-Nienaber}, {and} \bibinfo{person}{Craig Boutilier}.}
  \bibinfo{year}{2013}\natexlab{}.
\newblock \showarticletitle{Multi-dimensional single-peaked consistency and its
  approximations}. In \bibinfo{booktitle}{\emph{Proceedings of the 23rd
  International Joint Conference on Artificial Intelligence (IJCAI)}}.
  \bibinfo{pages}{375--382}.
\newblock
\urldef\tempurl%
\url{https://www.cs.toronto.edu/~cebly/Papers/SuiEtAl_singlePeaked_ijcai13.pdf}
\showURL{%
\tempurl}


\bibitem[Tan and Zhang(2007)]%
        {tan2007consecutive}
\bibfield{author}{\bibinfo{person}{Jinsong Tan} {and} \bibinfo{person}{Louxin
  Zhang}.} \bibinfo{year}{2007}\natexlab{}.
\newblock \showarticletitle{The consecutive ones submatrix problem for sparse
  matrices}.
\newblock \bibinfo{journal}{\emph{Algorithmica}}  \bibinfo{volume}{48}
  (\bibinfo{year}{2007}), \bibinfo{pages}{287--299}.
\newblock
\urldef\tempurl%
\url{https://doi.org/10.1007/s00453-007-0118-z}
\showDOI{\tempurl}


\bibitem[Terzopoulou et~al\mbox{.}(2021)]%
        {terzopoulou2021restricted}
\bibfield{author}{\bibinfo{person}{Zoi Terzopoulou}, \bibinfo{person}{Alexander
  Karpov}, {and} \bibinfo{person}{Svetlana Obraztsova}.}
  \bibinfo{year}{2021}\natexlab{}.
\newblock \showarticletitle{Restricted domains of dichotomous preferences with
  possibly incomplete information}. In \bibinfo{booktitle}{\emph{Proceedings of
  the 35th AAAI Conference on Artificial Intelligence (AAAI)}}.
  \bibinfo{pages}{5726--5733}.
\newblock
\urldef\tempurl%
\url{https://doi.org/10.1609/aaai.v35i6.16718}
\showDOI{\tempurl}


\bibitem[Tideman(1987)]%
        {tideman1987independence}
\bibfield{author}{\bibinfo{person}{T.~Nicolaus Tideman}.}
  \bibinfo{year}{1987}\natexlab{}.
\newblock \showarticletitle{Independence of clones as a criterion for voting
  rules}.
\newblock \bibinfo{journal}{\emph{Social Choice and Welfare}}
  \bibinfo{volume}{4} (\bibinfo{year}{1987}), \bibinfo{pages}{185--206}.
\newblock
\urldef\tempurl%
\url{https://doi.org/10.1007/bf00433944}
\showDOI{\tempurl}


\bibitem[Tydrichov{\'a}(2023)]%
        {tydrichova2023structural}
\bibfield{author}{\bibinfo{person}{Magdal{\'e}na Tydrichov{\'a}}.}
  \bibinfo{year}{2023}\natexlab{}.
\newblock \emph{\bibinfo{title}{Structural and algorithmic aspects of
  preference domain restrictions in collective decision making: Contributions
  to the study of single-peaked and {Euclidean} preferences}}.
\newblock \bibinfo{thesistype}{Ph.\,D. Dissertation}. \bibinfo{school}{Sorbonne
  Universit{\'e}}.
\newblock
\urldef\tempurl%
\url{https://theses.hal.science/tel-04143294v1/document}
\showURL{%
\tempurl}


\bibitem[Young(1975)]%
        {young1975social}
\bibfield{author}{\bibinfo{person}{H.~Peyton Young}.}
  \bibinfo{year}{1975}\natexlab{}.
\newblock \showarticletitle{Social choice scoring functions}.
\newblock \bibinfo{journal}{\emph{SIAM J. Appl. Math.}} \bibinfo{volume}{28},
  \bibinfo{number}{4} (\bibinfo{year}{1975}), \bibinfo{pages}{824--838}.
\newblock
\urldef\tempurl%
\url{https://doi.org/10.1137/0128067}
\showDOI{\tempurl}


\end{thebibliography}

\clearpage 
\appendix 
\addtocontents{toc}{\protect\setcounter{tocdepth}{1}}

\section{Appendix of \Cref{sec:rules}}

\subsection{Non-equivalence of Axis Rules}\label{sec:app_rules_not_eq}

In \Cref{sec:rules} (\Cref{ex1}), we discussed an example with $m=4$ for which 3 rules returned different axes. In this section, we provide another example constructed so that every pair of rules select different axes. Consider the profile 
\vspace{-15pt}

\begin{minipage}{0.5\textwidth}
\begin{align*}
	\\[7pt]
    18 & \times \{a,b\}, \\
    N & \times \{b,c\}, \\
    N & \times \{c,d\}, \\
    15 & \times \{d,e\}, \\
    4 & \times \{e,f\}, \\
    1 & \times \{a,g\}, \\
    20 & \times \{b,c,f,g\}, \\
    15 & \times \{a,e,f,g\}, \\
    2 & \times \{a,d,g\}. \\
\end{align*}
\end{minipage}
\begin{minipage}{0.5\textwidth}
\begin{tikzpicture}
	\renewcommand{\rowheight}{0.595cm}
    \axisheader{0}{a,b,c,d,e,f,g}
    
    \multiplicity{1}{18}
    \interval{1}{1}{2}

    \multiplicity{2}{$N$}
    \interval{2}{2}{3}

    \multiplicity{3}{$N$}
    \interval{3}{3}{4}

    \multiplicity{4}{15}
    \interval{4}{4}{5}

    \multiplicity{5}{4}
    \interval{5}{5}{6}

    \multiplicity{6}{1}
    \interval{6}{1}{1}
    \interval{6}{7}{7}

    \multiplicity{7}{20}
    \interval{7}{2}{3} 
    \interval{7}{6}{7}

    \multiplicity{8}{15}
    \interval{8}{1}{1}
    \interval{8}{5}{7}

    \multiplicity{9}{2}
    \interval{9}{1}{1}
    \interval{9}{4}{4}
    \interval{9}{7}{7}
\end{tikzpicture}
\end{minipage}
on $C =\{a,b,c,d,e,f,g\}$, where $N$ is an integer large enough to ensure $bcd$ is an interval of every optimal axis for any rule (note that such $N$ can always be found). 

\begin{table}[!t]
	\centering
	\begin{tabular}{l c c c c c} 
		\toprule
		& VD & MF & BC & MS & FT \\ 
		\midrule
		$\mathit{aefgbcd}$ & \textbf{36} & 38 & 124 & 126 & 132 \\
		$\mathit{efgabcd}$ & 37 & \textbf{37} & 99  & 119 & 163 \\
		$\mathit{gfabcde}$ & 42 & 42 & \textbf{88}& 108 & 244 \\
		$\mathit{agfbcde}$ & 39 & 39 & 99 & \textbf{99} & 195 \\ 
		$\mathit{eagfbcd}$ & 40 & 40 & 122 & 122 & \textbf{128} \\
		\bottomrule
	\end{tabular}
	\caption{Five axes on the profile $P$ defined in \Cref{sec:app_rules_not_eq}, and their cost for the different scoring rules. For each axis, we give its cost for all axis rules. The optimal values for each rule are given in bold. 
	}
	\label{tab:rules_not_eq}
\end{table}

\Cref{tab:rules_not_eq} shows five axes and their respective cost for each scoring rule introduced in \Cref{sec:rules}. 
Note that each of these axes minimizes the cost for a distinct one of the five rules. This can be verified either using a computer or by hand (noting that the fact that $bcd$ must be an interval of axis reduces the search space significantly). Instead of giving all details of the computation, we will discuss some behavioural tendencies of rules in order to better understand their differences. 

First, we note that VD and MF often yield the same cost. This is because many ballots of $P$ only approve two candidates, in which case the VD cost and MF cost are equal. More generally, given a ballot $A$ and and axis $\axis$, $\cost_{\VD}(A, \axis) = \cost_{\MF}(A, \axis)$ if and only if (1) $A$ is an interval of $\axis$, (2) $\axis$ creates a unique contiguous hole of size 1 in $A$, or (3) $\axis$ creates a unique contiguous hole in $A$ and there is only one approved candidate on the left (or on the right) of this hole. To distinguish VD and MF, we have added a ballot $\{a,d,g\}$ to $P$ which creates two contiguous holes on $\underline{a}e\mathit{f}\underline{g}bc\underline{d}$. This ensures that this axis is only optimal for VD.

Let us now focus on differences between MF and BC. Roughly speaking, MF seems more sensitive to the number of contiguous holes, while BC seems more sensitive to total size of the holes. For instance, a ballot associated to the approval vector $(1,0,0,0,0,0,1)$ achieves higher BC-cost than a ballot associated to the approval vector $(1,0,0,0,0,1,0)$, while MF assigns both the same cost. On the other hand, BC assigns the same cost to $(0,1,1,0,0,0,1)$ and $(0,1,0,1,0,0,1)$ while MF assigns a lower cost to the first one. This observation makes it possible to distinguish MF and BC, by finding suitable weights of $\{d,e\}$, $\{e,f\}$, and $\{a,d,g\}$ in $P$. 

We then note that BC and MS seems to assign similar costs quite often. Actually, given a ballot $A$ and and axis $\axis$, we have $\cost_{\BC}(A, \axis) = \cost_{\MS}(A, \axis)$ if and only if (1) $A$ is an interval of $\axis$, (2) $\axis$ creates a unique contiguous hole in $A$ and there is only one approved candidate on the left (or on the right) of this hole, or (3) $\axis$ creates two contiguous holes in $A$ and there is a unique approved candidate on the left of the left-most hole and on the right of the right-most hole. As a consequence, the BC and MS cost function take the same values for every approval ballot of size $|A| \le 3$. Thus, only $\{b,c,f,g\}$ and $\{a,e,f,g\}$ are able to distinguish BC and MS. 

Finally, FT gives more importance to bigger ballots, as for each interfering candidate we multiply the number of approved candidates on its left by the number of approved candidates on its right, the FT cost becomes larger with an increasing number of approved candidates. Like for MS, we used the ballots $\{b,c,f,g\}$ and $\{a,e,f,g\}$ (with suitable weights) to help differentiate FT from other rules. We note that these ballots are intervals of $\mathit{eagfbcd}$.  

\subsection{Minimum Swaps} \label{sec:appMS}

We prove in this section that the formula 
\[
\textstyle \cost_{\MS}(A, \axis) = \sum_{x\notin A} \min(|\{y \in A: y\axis x\}|,|\{y \in A: x \axis y\}|)
\]
implements the description of the Minimum Swaps rule, for which the cost of an axis is the minimal number of swaps we need to perform on this axis to make $A$ an interval of it.

Fix some ballot $A$ and some axis $\axis$. Let $\axis'$ be an axis of minimum swap distance to $\axis$ such that $A$ is an interval of $\axis'$. Write $S = \{ \{x, y\} \subseteq C : x \axis y \text{ and } y \axis' x \}$ for this minimum swap distance. It is clear that $\axis$ and $\axis'$ agree on the ordering of the approved alternatives in $A$; if they ordered some pair of approved alternatives in different ways, we could swap them in one of the axes and thereby reduce the swap distance $S$.

We first show that $S \ge \sum_{x\notin A} \min(|\{y \in A: y\axis x\}|,|\{y \in A: x \axis y\}|)$. Note that because $A$ is an interval of $\axis'$, every non-approved candidates $x \notin A$ must appear either to the left or to the right of all approved candidates in $\axis'$. Thus, $x$ must have been swapped with at least all candidates $y \in A$ to its right or to its left.

To see that $S \le \sum_{x\notin A} \min(|\{y \in A: y\axis x\}|,|\{y \in A: x \axis y\}|)$, partition the set $C \setminus A$ of non-approved candidates into two parts, corresponding to those candidates for which it is cheaper to push them to the left or to the right, respectively:
\begin{align*}
	L &= \{x \notin A : |\{y \in A: y\axis x\}| < |\{y \in A: x \axis y\}|\}, \\
	R &= \{x \notin A : |\{y \in A: y\axis x\}| \ge |\{y \in A: x \axis y\}|\}.
\end{align*}
Consider the axis $\axis'' = \axis_L \axis_A \axis_R$ obtained by placing the $L$-candidates left, the $A$-candidates center, and the $R$-candidates right, but keeping the same ordering of candidates as $\axis$ within the sets $L$, $A$, and $R$. By construction, $A$ is an interval of $\axis''$, and it is easy to compute that the swap distance between $\axis$ and $\axis''$ is precisely $\sum_{x\notin A} \min(|\{y \in A: y\axis x\}|,|\{y \in A: x \axis y\}|)$, which by minimality of $S$ must be at least as large as $S$.

\subsection{Complexity} \label{sec:appComplexity}

In this section we provide the proof that the rules defined in \Cref{sec:rules} are all hard to compute. Some of these results are known. In particular, framed as \emph{near-C1P matrices problems}, the VD, MF, and BC rules are known to be NP-complete, as surveyed by \citet{consecutiveones}. Here, we give the reduction for VD and BC and use a simple argument to deduce hardness for the other rules MF, MS, and FT.

We first give the proof for \emph{Voter Deletion}. For this, we recall that a profile is linear if and only if its approval matrix satisfies the C1P. Thus, computing VD is equivalent to the \emph{consecutive ones submatrix} problem, which was already shown to be NP-complete \citep[Theorem 4.24]{booth1975pq}. For convenience, we provide the proof here.
\begin{restatable}{theorem}{complexityVD}
    \label{thm:vd-complexity}
    The VD problem is NP-complete even if each voter approves at most two candidates (i.e., $\max_i |A_i| = 2$).
\end{restatable} 
\begin{proof}
    We use a polynomial time reduction from the Hamiltonian path problem, known to be NP-complete \citep{Karp1972}. Let $G = (X, E)$ be an undirected graph with $|X| = n$ and $|E| = m$. A Hamiltonian path is a path that visits each vertex exactly once. The Hamiltonian path problem consists in deciding whether such a path exists. The Voter Deletion problem consists in deciding, given as input a profile $P$ and $k\in \mathbb N$, whether an axis of cost at most $k$ exists. We now show that we can reduce the Hamiltonian path problem to the VD problem.
    
    We create an election with $C = X$ as the set of candidates. Then, we define the profile $P$ as follows: for each edge $(u,v) \in E$, there exist a voter $v_e$ approving $\{u,v\}$. Thus, all voters are distinct.
    Since the size of all approval ballots is $2$, any axis can satisfy at most $n-1$ pairwise distinct voters, and if $n-1$ pairwise distinct voters are satisfied by the axis $c_1 \axis \dots \axis c_m$, then $(c_1,\dots, c_m)$ is a Hamiltonian path. Conversely, if $(c_1,\dots, c_m)$ is a Hamiltonian path of $G$, the axis $c_1 \axis \dots \axis c_m$ satisfies $n-1$ voters with approval ballots of the form $\{c_i, c_{i+1}\}$. Thus there exists a Hamiltonian path if and only if there exist an axis with a Voter Deletion cost of $n-1$ in the election $P$. This proves that VD is NP-hard. The completeness comes from the fact that it takes polynomial time to compute the VD cost.
\end{proof}

Similarly, one can show that the ballot completion rule is equivalent to the \emph{consecutive ones matrix augmentation} problem, which is also NP-complete \citep[Theorem 4.19]{booth1975pq}. For convenience, we provide the proof here.

\begin{restatable}{theorem}{complexityBC}
    \label{thm:bc-complexity}
    The BC problem is NP-complete even if each voter approves at most two candidates (i.e., $\max_i |A_i| = 2$). 
\end{restatable} 
\begin{proof}
    We use a polynomial time reduction from the Optimal Arrangement problem, known to be NP-complete \citep{Garey1976}. Let $G = (X, E)$ be an undirected graph with $|X| = n$ and $|E| = m$. The Optimal Arrangement problem decides, given an integer $k$, whether there is a one-to-one function $f:X \rightarrow [1,n]$ such that $\sum_{(u,v)\in E} |f(u) - f(v)| \le k$. 
    The Ballot Completion problem decides, given an integer $k$, whether there exists an axis with cost $\le k$.
    We create an election with $C = X$ the set of candidates. The set of voters is defined as follows: for each edge $(u,v) \in E$, we introduce a voter $v_e$ with ballot $\{u,v\}$. For an axis $\axis$, let $f_\axis(c)$ correspond to the position of the candidate $c$ on the axis (e.g., $1$ for the left-most candidate). Then given a ballot $A = \{u,v\}$ and an axis $\axis$, the Ballot Completion cost equals $\cost_{\BC}(A,\axis) = |f_\axis(u) - f_\axis(v)| - 1$. Thus, the ballot completion cost of an axis $\axis$ with this profile is equal to $\sum_{\{u,v\} \in E} |f_\axis(u) - f_\axis(v)| - 1 = \left (\sum_{\{u,v\} \in E} |f_\axis(u) - f_\axis(v)| \right )- |E|$. Therefore, there exists an arrangement of cost $\le k$ if and only if there exists an axis of BC cost $\le k - |E|$. Thus, BC is NP-hard. The completeness comes from the fact that the BC cost is computable in polynomial time.
\end{proof}

Now, observe that when $\max_i |A_i| = 2$, the VD and MF cost functions are identical, because if the ballot $A$ is not an interval of the axis, then it always costs one flip to make it an interval (by removing one of the two approved candidates from the ballot). Moreover, observe that the MS, FT, and BC cost functions are equivalent when $\max_i |A_i| = 2$, which becomes clear from their formulas, as for all interfering candidates $x \notin A$, we have $|\{y \in A, y \axis x\}| = |\{y \in A, x \axis y \}|) = 1$. Hence we obtain NP-completeness for all 5 rules.

\section{Omitted Proofs of \Cref{sec:axiom}}
\label{sec:appPROOFs}

\subsection{Neutrality and Consistency with Linearity}\label{sec:appNeutr}

\costfunction*
\begin{proof}
    Let $f$ be a scoring rule induced by the cost function $\cost$. It is straightforward to check that if $\cost$ satisfies conditions (1)--(3), then $f$ is neutral (due to (3)) and consistent with linearity (due to (1)).
    
    For the other direction, note that for each fixed ballot $A$, adding a constant to the cost function $\cost(A, \cdot)$ will not change the optimum axis. Thus, we can always select a cost function so that for each ballot $A$, we have $\min_{\axis} \cost(A, \axis) = 0$. We will show that the cost function, chosen in this way, satisfies (1) and (2), and afterwards we will derive another cost function that also induces $f$ and that in addition satisfies (3).
    
    We first show (1). That we always have $\cost(A, \axis) \ge 0$ is clear from our choice of cost function. Assume for a contradiction that there is an axis $\axis$ and a ballot $A$ such that $\cost(A,\axis) = 0$ but $A$ is not an interval of $\axis$. Then, on the linear profile $P = \{A\}$, we have ${\axis} \in f(P)$, which is a contradiction with consistency with linearity. Similarly, if $A$ is an interval of $\axis$ but $\cost(A,\axis) > 0$, then on the linear profile $P = \{A\}$, we have ${\axis} \not\in f(P)$ while $\axis$ is consistent with $P$, a contradiction with consistency with linearity.

    We now show (2). If $A$ is an interval of $\axis$, it is also an interval of $\cev \axis$, so from (1) we clearly have $\cost(A,\axis) = \cost(A,\cev \axis)$. Assume now that $A$ is not an interval of $\axis$. Thus, $y = \cost(A,\axis) > 0$ and $y' = \cost(A,\cev\axis) > 0$. Assume for a contradiction that $y \ne y'$, and without loss of generality that $y < y'$. Let us denote the candidates $c_1, \dots,c_m$ such that ${\axis} = c_1c_2\dots c_m$. Moreover, let $z > 0$ be the minimum value of $\cost(A',\axis)$ over all ballots $A'$ that are not an interval of $\axis$. Take $q \in \mathbb N$ such that $q > y/z$ and consider the profile $P$ which contains $A$ and for each $i \in [1,m-1]$, $q$ ballots $\{c_i, c_{i+1}\}$. Clearly, any axis ${\axis'} \notin \{\axis, \cev\axis\}$ is breaking at least one pair, inducing a cost greater that $q \cdot z > y$. The cost of $\axis$ is $y$ and the cost of $\cev\axis$ is $y'> y$. Thus, $f(P) = \{\axis\}$ which contradicts the definition of axis rules (which requires that whenever an axis is selected, then so is its reverse axis). Therefore, $y = y'$ and $\cost(A,\axis) = \cost(A,\cev \axis)$.

    For (3), we show that $f$ is induced by a cost function $\cost^*$ such that $\cost^*(A, \axis)$ only depends on $x_{A,\axis}$.
    Let $\Pi$ be the set of all permutations of the candidates. For a permutation $\pi \in \Pi$, for each ballot $A$ we write $\pi(A) =\{\pi(a):a\in A\}$, and for each axis ${\axis} = c_1\dots c_m$ we write $\pi(\axis) = \pi(c_1)\dots\pi(c_m)$.  Then we define 
    \[
    	\cost^*(A,\axis) = \sum_{\pi\in \Pi} \cost(\pi(A),\pi(\axis)) \text{ for all $A$ and ${\axis} \in \mathcal A$.}
   	\]
   	We will show that this cost function still induces $f$, and that it satisfies conditions (1)--(3).
   	
    To show that $f$ is still induced by this cost function, %
    let ${\axis} \in f(P)$ be an optimal axis for profile $P$.
    Then, by neutrality, $\pi(\axis) \in f(\pi (P))$ for all $\pi\in\Pi$. This implies that $\cost(\pi(\axis), \pi(P))\leq \cost(\pi(\axis'),\pi(P))$ for all axes ${\axis'} \in \mathcal A$. Since this inequality carries over to the sum over all $\pi\in\Pi$, this implies $\cost^*(\axis,P)\leq \cost^*(\axis',P)$ for all $\axis'$. For the other direction, let ${\axis'} \notin f(P)$ and fix some ${\axis} \in f(P)$. With the same argument, we obtain  $\cost^*(\axis,P)< \cost^*(\axis',P)$, which shows that, for all profiles, an axis $\axis$ has minimal cost w.r.t. $\cost^*$ if and only if it is chosen by $f$. 
    
    Finally, we check that $\cost^*$ satisfies the conditions of the lemma. Because $\cost$ satisfies conditions (1) and (2), it is clear that $\cost^*$ also satisfies conditions (1) and (2). For condition (3), take any $A,\axis$ and $A',\axis'$ with the same approval vector, i.e., $x_{A,\axis} = x_{A',\axis'}$. Then, there exists a permutation $\tau \in \Pi$ with $\tau(A)=A'$ and $\tau(\axis)=\axis'$.
    Thus, we obtain that  $\cost^*(A',\axis') =  \cost^*(\tau (A),\tau(\axis)) = \sum_{\pi\in \Pi} \cost(\pi(\tau(A)),\pi(\tau(\axis))) = \sum_{\pi'\in \Pi} \cost(\pi'(A),\pi'(\axis)) = \cost^*(A,\axis)$.
\end{proof}

\subsection{Characterization of Voter Deletion} \label{sec:appCharacterization}
\vdcharac*

\begin{proof}

    We already showed that VD satisfies all the axioms. For the other direction, let $f$ be a scoring rule satisfying neutrality, consistency with linearity, resistance to cloning and ballot monotonicity. As shown in \Cref{sec:appNeutr}, $f$ is induced by a symmetric cost function $\cost$ with $\cost(A,\axis) = 0$ if and only if $A$ forms an interval in $\axis$. Further, $\cost$ only depends on the approval vector $x_{A,\axis}$, i.e., there exists a function $g:\{0,1\}^m \rightarrow \mathbb R_{\ge 0}$ such that $\cost(A,\axis) = g(x_{A,\axis})$ for all ballots $A$ and axis $\axis$.

    The steps of the proof are as follows:
    \begin{itemize}
        \item In \Cref{lem:characterization:num_interfering}, using ballot monotonicity, we show that there is a function $h$ such that for all $A$ and $\axis$ such that $A$ is not an interval of $\axis$, $\cost(A,\axis) = h(m, k_{\text{app}} ,k_{\text{int}})$, where $m$ is the number of candidates, $k_{\text{app}} =|A|$ is the number of approved candidates and $k_{\text{int}}$ is the number of interfering candidates.
        \item In \Cref{lem:characterization:sum}, using resistance to cloning, we show that for $A$ not an interval of $\axis$, $\cost(A,\axis)$ only depends on the sum $k_{\text{app}}+k_{\text{int}}$, i.e., there is $h$ such that $\cost(A,\axis) = h(m,  k_{\text{app}}+k_{\text{int}})$.
        \item In \Cref{lem:characterization:extremes}, we show that for $A$ not interval of $\axis$, $\cost(A,\axis)$ can only take two values: $\cost(A,\axis) = h_m^*$ if $k_{\text{app}}+k_{\text{int}}=m$ and $\cost(A,\axis) = h_m$ otherwise.
        \item Finally, in \Cref{lem:characterization:at_least,lem:characterization:at_most}, we show that $h^*_m = h_m$ and thus that the rule is VD.
    \end{itemize}

    In all of the following, $A$ and $\axis$ are chosen such that $A$ is not an interval of $\axis$, and thus we already know that $\cost(A,\axis) > 0$. In the first step of the proof, we apply ballot monotonicity to a very symmetric profile.
    \begin{lemma}
    	\label{lem:characterization:num_interfering}
        There is a function $h$ such that $\cost(A,\axis) = h(m,k_{\textup{app}} ,k_{\textup{int}})$ for all $A$ and $\axis$ such that $A$ is not an interval of $\axis$, where $m$ is the number of candidates, $k_{\textup{app}} = |A|$ and $k_{\textup{int}}$ is the number of interfering candidates.
    \end{lemma}
    \begin{proof}
        Assume for a contradiction that there are ballots $A^{\dagger}$, $A^{\ddagger}$ and axes $\axis^{\dagger}$ and $\axis^{\ddagger}$ such that $\cost(A^{\dagger},\axis^{\dagger}) \ne \cost(A^{\ddagger},\axis^{\ddagger})$ but $|A^{\dagger}| = |A^{\ddagger}| = k_{\text{app}}$ and the number of interfering candidates $k_{\text{int}}$ is the same in both cases. Denote by $x^{\dagger} = x_{A^{\dagger},\axis^{\dagger}}$ and $x^{\ddagger} = x_{A^{\ddagger},\axis^{\ddagger}}$ the respective approval vectors. Obviously, $x^{\dagger} \ne x^{\ddagger}$.
        
        Consider the set of candidates $\mathcal C = \{c_1,\dots, c_{k_{\text{app}}}\} \cup \{d_1,\dots, d_{k_{\text{int}}}\} \cup \{b_1,\dots, b_{m-(k_{\text{app}}+k_{\text{int}})}\} = C \cup D \cup B$. %
        Let $P^*$ be the profile in which every ballot $A$ of size $|A| = k_{\text{app}}$ is approved by one voter. By neutrality, all axes are chosen.

		Call an axis $\axis$ \emph{good} if the set of interfering candidates of $C$ on $\axis$ is $D$, i.e., if $\{d \notin C: \exists c_i,c_j \in C, c_i \axis d \axis c_j\} = D$. 
        Now, let $P$ be the profile obtained from $P^*$ where the voter with ballot $C$ changes it to $C\cup D$. Since every good axis is chosen at $P^*$, by ballot monotonicity, every good axis must be chosen at $P$.
        
        Let $\axis^1$ and $\axis^2$ be two good axes such that the approval vectors of $C$ on these axes are $x_{C,\axis^1} = x^{\dagger}$ and $x_{C,\axis^2} = x^{\ddagger}$. This is possible since there are $k_{\text{app}}$ approved candidates in $C$ and $k_{\text{int}}$ interfering candidates with respect to both axes. Because both axes are chosen at both $P^*$ and $P$, we have that $\cost(P^*,\axis^1) = \cost(P^*,\axis^2)$  and $\cost(P,\axis^1) = \cost(P,\axis^2)$.
        Note that, for any good axis $\axis$, we have $\cost(C \cup D, \axis) = 0$, and therefore 
        \[
        \cost(P^*,\axis) = \cost(P,\axis) - \cost(C \cup D, \axis) + \cost(C,\axis) = \cost(P,\axis) + \cost(C,\axis).
        \]
        Thus, we deduce $\cost(C,\axis^1) = \cost(C,\axis^2)$ which means $g(x^{\dagger}) = g(x^{\ddagger})$. This contradicts $\cost(A^{\dagger},\axis^{\dagger}) \ne \cost(A^{\ddagger},\axis^{\ddagger})$, and concludes the proof of the lemma.
    \end{proof}

    Next, we use resistance to cloning to show that the cost actually only depends on $k_{\text{app}} + k_{\text{int}}$.
    \begin{lemma}
    	\label{lem:characterization:sum}
        There is a function $h$ such that $\cost(A,\axis) = h(m, k_{\textup{app}} + k_{\textup{int}})$ for all $A$ and $\axis$ such that $A$ is not an interval of $\axis$.
    \end{lemma}
    \begin{proof}
        For $m=3$, all scoring rules are equivalent. 
        So let $m\geq 4$.
        Let $3\leq k_{\text{app}}+k_{\text{int}} \leq m $. By \Cref{lem:characterization:num_interfering}, there exists $h$ such that $\cost(A,\axis) = h(m, k_{\text{app}}, k_{\text{int}})$ for all $A$ and $\axis$. In this proof, we show that for all $k_{\text{app}}$ and $k_{\text{int}}$, we have $h(m, k_{\text{app}}, k_{\text{int}}) = h(m, k_{\text{app}}+ k_{\text{int}}-1,1)$, implying that the cost function only depends on $k_{\text{app}}+ k_{\text{int}}$.

        This is clearly true for $k_{\text{app}}+ k_{\text{int}}=3$ as in this case the only possibility is $k_{\text{app}} = 2$ and $k_{\text{int}} = 1$, otherwise the ballot is an interval.
        For $k_{\text{app}}+ k_{\text{int}}>3$, let the set of candidates be $\mathcal C = \{c_1,\dots, c_{k_{\text{app}}}\} \cup \{d_1,\dots, d_{k_{\text{int}}}\} \cup \{b_1,\dots, b_{m-(k_{\text{app}}+k_{\text{int}})}\} = C \cup D \cup B$. Note that if $k_{\text{app}}+k_{\text{int}}=m$, $B = \emptyset$.

        Assume first that $k_{\text{app}}+k_{\text{int}}<m$.
        Let $z$ be the minimal value of $\cost(A,\axis)$, taken over all ballots $A$ and axes $\axis$ defined on $m$ or fewer candidates, such that $A$ is not an interval of $\axis$. Further, let $y = \max(h(4,2,1), h(m,k_{\text{app}},k_{\text{int}}), h(m,k_{\text{app}}+k_{\text{int}}-1,1))$.
        Take $q \in \mathbb N$ such that $q > y/z$ and consider the following profile $P$ on 4 candidates $\{c_1,c_2,d_1,b_1\}$.
		\vspace{-10pt}

        \begin{minipage}{0.5\textwidth}
            \begin{align*}
                \\[7pt]
                q & \times \{c_1, d_1\}, \\
                q & \times \{c_1, d_1, c_2\}, \\
                q & \times \{d_1\}, \\
                q & \times \{c_2\}, \\
                q & \times \{b_1\}, \\
                1 & \times \{c_1, c_2\},\\
                1 & \times \{d_1,c_2\}.\\
            \end{align*}
            \end{minipage}
            \begin{minipage}{0.5\textwidth}
            \begin{tikzpicture}
                \renewcommand{\rowheight}{0.595cm}
                \axisheader{0}{c_1,d_1,c_2,b_1}

                \multiplicity{1}{$q$}
                \interval{1}{1}{2}

                \multiplicity{2}{$q$}
                \interval{2}{1}{3}

                \multiplicity{3}{$q$}
                \interval{3}{2}{2}

                \multiplicity{4}{$q$}
                \interval{4}{3}{3}

                \multiplicity{5}{$q$}
                \interval{5}{4}{4}

                \multiplicity{6}{1}
                \interval{6}{1}{1}
                \interval{6}{3}{3}

                \multiplicity{7}{1}
                \interval{7}{2}{3}
            \end{tikzpicture}
            \end{minipage}

        Clearly, any axis such that $\{c_1,d_1\}$ or $\{c_1,d_1,c_2\}$ do not form an interval has cost greater than $q \cdot z > y$. The other axes (up to reversal) are:
        \begin{align*}
            \underline{c_1}d_1\underline{c_2}b_1 \\
            b_1\underline{c_1}d_1\underline{c_2} \\
            \underline{d_1}c_1\underline{c_2}b_1 \\
            b_1\underline{d_1}c_1\underline{c_2}
        \end{align*}
        They each break one of $\{c_1,c_2\}$ or $\{d_1,c_2\}$ with cost $h(4,2,1) \le y$, thus they all are optimal.

        Now, clone $c_2$ into $\{c_2, \dots, c_{k_{\text{app}}}\}$, $d_1$ into $\{d_1,\dots,d_{k_{\text{int}}}\}$ and $b_1$ into $\{b_1,\dots,b_{m-(k_{\text{app}}+k_{\text{int}})}\}$. Clearly, all clones of each category need to be next to each other on the axis, otherwise the $q$ ballots containing these clones (obtained from $\{d_1\}$, $\{c_2\}$ or $\{b_1\}$) would induce a cost greater than $q \cdot z > y$. Combining this with resistance to cloning gives that the axes should be of one of the following forms (up to change of the positions of the clones):
        \begin{align*}
            \underline{c_1}d_1\dots d_{k_{\text{int}}}\underline{ c_2\dots c_{k_{\text{app}}} }b_1\dots b_{m-(k_{\text{app}}+k_{\text{int}})} \\
            b_1\dots b_{m-(k_{\text{app}}+k_{\text{int}})}\underline{c_1}d_1\dots d_{k_{\text{int}}}\underline{ c_2\dots c_{k_{\text{app}}} } \\
            \underline{d_1\dots d_{k_{\text{int}}} } c_1 \underline{ c_2\dots c_{k_{\text{app}}} } b_1\dots b_{m-(k_{\text{app}}+k_{\text{int}})} \\
            b_1\dots b_{m-(k_{\text{app}}+k_{\text{int}})} \underline{d_1\dots d_{k_{\text{int}}} } c_1 \underline{ c_2\dots c_{k_{\text{app}}} }
        \end{align*}
        Indeed, for all these axes, there is only one ballot that is not an interval. The first two axes break the one obtained from $\{c_1,c_2\}$ (by adding clones) with $k_{\text{app}}$ approved candidates and $k_{\text{int}}$ interfering ones, while the last two break the ballot obtained from $\{d_1,c_2\}$ with $k_{\text{app}}+k_{\text{int}}-1$ approved candidates and $1$ interfering one ($c_1$). Thus, the first two have cost $h(m,k_{\text{app}},k_{\text{int}}) \le y$ and the last two have cost $h(m,k_{\text{app}}+k_{\text{int}}-1,1) \le y$. Thus, they are respectively the axes with lowest cost that can be reduced to the axes obtained with $4$ candidates $\{c_1,c_2,d_1,b_1\}$. By resistance to cloning, this means that each of these axes should be among the optimal ones. This directly implies that  $h(m,k_{\text{app}},k_{\text{int}}) = h(m,k_{\text{app}}+k_{\text{int}}-1,1)$. 
        
        The proof if $k_{\text{app}}+k_{\text{int}} = m$ is exactly the same, but without the candidates $b_i$.
    \end{proof}

    We now proceed to show that for a given $m$, the cost only depends on whether both extremes of the axis are approved.
    \begin{lemma}
    	\label{lem:characterization:extremes}
        For each number $m$ of alternatives, there are positive constants $h^*_m$ and $h_m$ such that for all $A$ and $\axis$ such that $A$ is not an interval of $\axis$, 
        \[
        	\cost(A,\axis) = \begin{cases}
        		h^*_m & \text{if $A$ contains both extremes of $\axis$,} \\
        		h_m & \text{otherwise.}
        	\end{cases}
        \]
    \end{lemma}
    \begin{proof}
       From \Cref{lem:characterization:sum}, we know that there is $h$ such that $\cost(A,\axis) = h(m, k_{\text{app}}+k_{\text{int}})$. In the rest of the proof, we write $k = k_{\text{app}} + k_{\text{int}}$.
       
       To obtain the values required by the lemma, we can simply take $h^*_m = h(m, m)$, as $A$ contains both extremes of $\axis$ if and only if $k = m$. To define $h_m$, we will show that for $4 \le k < m$, we have $h(m,k) = h(m,k-1)$. Since we know that $k \ge 3$ (otherwise the ballot is an interval), this implies that there is $h_m$ such that for all $k \in [3,m-1]$, $h(m,k) = h_m$, proving the lemma. 
       
       Let $k \in [4,m-2]$ (the proof for $k = m-1$ is identical, but without the candidates $b_i$).
        As in the proof of \Cref{lem:characterization:sum}, let $z$ be the minimum cost for a non-interval ballot of an axis (for at most $m$ candidates) and let $y = \max(h(5,3),h(m,k),h(m,k-1))$. Take $q \in \mathbb N$ such that $q > y/z$ and consider the following profile $P$ over the set of 5 candidates $C = \{a_1,b_1,c_1,c_2,d_1\}$:
        \vspace{-10pt}
        
        \begin{minipage}{0.5\textwidth}
            \begin{align*}
                \\[7pt]
                q & \times \{d_1,c_2\}, \\
                q & \times \{c_2, a_1\}, \\
                q & \times \{c_1,d_1,c_2,a_1\}, \\
                q & \times \{d_1\}, \\
                q & \times \{c_2\},\\
                q & \times \{b_1 \}, \\
                1 & \times \{c_1,c_2\}.\\
            \end{align*}
            \end{minipage}
            \begin{minipage}{0.5\textwidth}
            \begin{tikzpicture}
                \renewcommand{\rowheight}{0.595cm}
                \axisheader{0}{c_1,d_1,c_2,a_1,b_1}

                \multiplicity{1}{$q$}
                \interval{1}{2}{3}

                \multiplicity{2}{$q$}
                \interval{2}{3}{4}

                \multiplicity{3}{$q$}
                \interval{3}{1}{4}

                \multiplicity{4}{$q$}
                \interval{4}{2}{2}

                \multiplicity{5}{$q$}
                \interval{5}{3}{3}

                \multiplicity{6}{$q$}
                \interval{6}{5}{5}

                \multiplicity{7}{1}
                \interval{7}{1}{1}
                \interval{7}{3}{3}

            \end{tikzpicture}
            \end{minipage}

        On this profile, any axis breaking one of the ballots of the first three categories induces a cost greater than $q \cdot z > y$. The only axes that do not break these ballots are the following:
        \begin{align*}
            \underline{c_1}d_1\underline{c_2}a_1b_1 \\
            b_1\underline{c_1}d_1\underline{c_2}a_1 \\
            d_1\underline{c_2}a_1\underline{c_1}b_1 \\
            b_1d_1\underline{c_2}a_1\underline{c_1}
        \end{align*}
        Note that the only ballot in $P$ that is not an interval of these axes is $\{c_1,c_2\}$ with a cost of $h(5,3) \le y$. Thus, all these axes are optimal for $f$.

        Now, clone $c_2$ into $\{c_2, \dots, c_{k-2}\}$, $d_1$ into $\{d_1,d_2\}$ and $b_1$ into $\{b_1,\dots,b_{m-k-1}\}$. Clearly, all clones of each category need to be next to each other on the axis, otherwise the $q$ ballots containing these clones (obtained from $\{d_1\}$, $\{c_2\}$ or $\{b_1\}$) would induce a cost greater than $q \cdot z > y$. Combining this with resistance to cloning gives that the axes should be of one of the following forms (up to change of the positions of the clones):
        \begin{align*}
            \underline{c_1}d_1d_2\underline{c_2\dots c_{k-2}}a_1 b_1 \dots b_{m-k-1} \\
          b_1 \dots b_{m-k-1}   \underline{c_1}d_1d_2\underline{c_2\dots c_{k-2}}a_1 \\
            d_1d_2\underline{c_2\dots c_{k-2}}a_1  \underline{c_1}b_1 \dots b_{m-k-1} \\
          b_1 \dots b_{m-k-1} d_1d_2\underline{c_2\dots c_{k-2}}a_1  \underline{c_1}
        \end{align*}
        Indeed, for all these axes, there is only one ballot that is not an interval: $\{c_1,c_2,\dots,c_{k-2}\}$ with cost $h(m,k)$ for the first two axes and $h(m,k-1)$ for the last two axes. In both cases, this cost is at most $y$, so they are respectively the axes with lowest cost that can be reduced to the axes obtained with $5$ candidates $\{a_1,b_1,c_1,c_2,d_1\}$. By resistance to cloning, this means that each of these axes should be among the optimal ones. This  implies that  $h(m,k) = h(m,k-1)$. The proof if $k = m-1$ is exactly the same, but without candidates $b_i$.

        This proves that for a given $m$, there exist some value $h_m$ such that for all non-interval ballots $A$ on $\axis$ with $k_{\text{app}} + k_{\text{int}} < m$, we have $\cost(A,\axis) = h_m$.
    \end{proof}

    Finally, we prove that for all $m \ge 4$, we have $h_m = h_m^*$. The argument proceeds in three parts: (i) for all $m \ge 4$, we have $h_m^* \le h_m$ (\Cref{lem:characterization:at_least}), (ii) for all $m \ge 6$, we have $h_m^* \ge h_m$ (\Cref{lem:characterization:at_most}), and finally (iii) for all $m \ge 4$, if $h^*_{m+1} = h_{m+1}$, then $h_m^* = h_m$. 
    
    \begin{lemma} \label{lem:lem5vd}
    	\label{lem:characterization:at_most}
        $h^*_m\leq h_m$ for all $m\geq 4$.
    \end{lemma}
\begin{proof}
    First, we show $h^*_m\leq h_m$ for all $m \ge 4$. Let $m \ge 4$, and assume for a contradiction that $h^*_m > h_m$. Again, let $z > 0$ be the minimal cost of a non-interval ballot on any axis defined on at most $m$ candidates. Take $q \in \mathbb N$ such that $q > \max(h_m,h_{m+1})/z$, and consider the profile $P$ on $m$ candidates $\mathcal C = \{c_1,c_2\} \cup \{b_1\} \cup \{d_1,\dots,d_{m-3}\}$ with $D = \{d_1,\dots,d_{m-3}\}$:
    \vspace{-10pt}
    
    \begin{minipage}{0.5\textwidth}
        \begin{align*}
            \\[7pt]
            q &\times D \cup \{c_2\}, \\
            q &\times D \cup \{c_1,b_1\}, \\
            1 &\times \{c_1,c_2\}.\\
        \end{align*}
        \end{minipage}
        \begin{minipage}{0.5\textwidth}
        \begin{tikzpicture}
            \renewcommand{\rowheight}{0.595cm}
            \renewcommand{\colwidth}{0.8cm}
            \axisheader{0}{c_1,b_1,d_1,\cdots,d_{m-3},c_2}

            \multiplicity{1}{$q$}
            \interval{1}{3}{6}

            \multiplicity{2}{$q$}
            \interval{2}{1}{5}

            \multiplicity{3}{1}
            \interval{3}{1}{1}
            \interval{3}{6}{6}

        \end{tikzpicture}
        \end{minipage}

    Note that all axes that breaks one of the first two ballots have cost at least $q \cdot z > h_m$. The other axes are of the following form (up to change of positions among candidates of $D$):
    \begin{align*}
        \underline{c_1}b_1d_1\dots d_{m-3}\underline{c_2} \\
        b_1\underline{c_1}d_1\dots d_{m-3}\underline{c_2}
    \end{align*}
    The only ballot that is not an interval of these axis is $\{c_1,c_2\}$ with a cost of $h^*_m$ on the first axis and $h_m$ on the second one. Since $h_m < h^*_m$, this means only the second axis is optimal. 
    
    Now, let's clone $b_1$ into $\{b_1,b_2\}$. Again, all axes that do not comply with the ballots of the first two categories have cost higher than $q \cdot z > h_{m+1}$. The only other axes that generalizes $b_1c_1d_1\dots d_{m-3}c_2$ and do not break the ballots of the first two categories are:
    \begin{align*}
        \underline{b_2}b_1c_1d_1\dots d_{m-3}c_2 \\
        b_1\underline{b_2}c_1d_1\dots d_{m-3}c_2 \\
        b_1c_1\underline{b_2}d_1\dots d_{m-3}c_2 
    \end{align*}
    The cost of $\{c_1,c_2\}$ on any of these axes is $h_{m+1}$. By resistance to cloning, at least one of them should be among the optimal axes. Since they all have the same cost, they all are optimal for this profile. 

    Now, let us remove the clone $b_1$ of $b_2$. By resistance to cloning, the following two axes should both be among the optimal axes:
    \begin{align*}
        \underline{b_2}c_1d_1\dots d_{m-3}c_2 \\
        c_1\underline{b_2}d_1\dots d_{m-3}c_2 
    \end{align*}
    Since the only ballot that is not an interval of these axes is $\{c_1,c_2\}$, their respective costs are $h_m$ and $h^*_m$. However, we assumed that $h^*_m > h_m$, so the second axis cannot be optimal, a contradiction.
\end{proof}

    \begin{lemma}
    	\label{lem:characterization:at_least}
        $h^*_m  \geq h_m$ for all $m\geq 6$.
    \end{lemma}
    \begin{proof}
        Let $m \ge 5$ and assume for a contradiction that $h^*_{m+1} < h_{m+1}$. Consider the set of candidates $\mathcal C = \{c_1, c_2\}\cup \{d_1,d_2\}\cup\{b_1,\dots b_{m-4}\}$ with $B =\{b_1,\dots b_{m-4}\} \ne \emptyset$. Again, let $z$ be the lowest cost of any non-interval approval ballot on any axis defined on at most $m$ candidates. Let $y = \max (h^*_m, h^*_{m+1})$ and take $q \in \mathbb N$ such that $q > y/z$, and consider the following profile:
        \vspace{-10pt}

        \begin{minipage}{0.5\textwidth}
            \begin{align*}
                \\[7pt]
                q & \times \{d_1,c_2\}, \\
                q & \times \{d_2,c_2\}, \\
                q & \times \{b_1,\dots,b_{m-4}\},\\
                1 & \times \{c_1,c_2\}.\\
            \end{align*}
            \end{minipage}
            \begin{minipage}{0.5\textwidth}
            \begin{tikzpicture}
                \renewcommand{\rowheight}{0.595cm}
                            \renewcommand{\colwidth}{0.8cm}
                \axisheader{0}{d_1,c_2,d_2,c_1,b_1,\cdots,b_{m-4}}
    
                \multiplicity{1}{$q$}
                \interval{1}{1}{2}
    
                \multiplicity{2}{$q$}
                \interval{2}{2}{3}

                \multiplicity{3}{$q$}
                \interval{3}{5}{7}

                \multiplicity{4}{1}
                \interval{4}{2}{2}
                \interval{4}{4}{4}
    
            \end{tikzpicture}
            \end{minipage}
        Again, any axis breaking one of the ballots of the first three categories induces a cost of at least $q \cdot z > y$. These ballots are intervals of an axis $\axis$ if $\axis$ contains the interval $d_1 \axis c_2 \axis d_2$ and the set $B$ forms an interval. It can have $B$ before or after $d_1 \axis c_2 \axis d_2$ on the axis, and $c_1$ between the two intervals or on one extremity of the axis. In any axis of this kind, the only ballot that is not an interval is $\{c_1,c_2\}$ with cost $h_m \le y$, since $c_2$ is not an extremity of the axis. Thus, all these axes are selected by the rule. In particular, the axis $\axis^* = d_1\underline{c_2}d_2\underline{c_1}b_1\dots b_{m-4}$ is selected.

        Let us now clone $c_1$ into $\{c_1,c_3\}$, thereby obtaining $m+1 \ge 6$ candidates. Any axis generalizing $\axis^*$ breaks at least the ballot $\{c_1,c_2,c_3\}$, and at least one extreme of the axis is not part of that ballot (since $c_1$ and $c_2$ are not on the extremes). Thus, the cost of any axis generalizing $\axis^*$ is at least $h_{m+1}$. However, consider the axis ${\axis'} = \underline{c_1}d_1\underline{c_2}d_2b_1\dots b_{m-4} \underline{c_3}$. The only ballot that is not an interval of this axis is $\{c_1,c_2,c_3\}$, and both extremes of the axis are approved by it, so the cost is $h^*_{m+1} < h_{m+1}$. This implies that no axes generalizing $\axis^*$ can be selected as they do not have lowest cost. This contradicts resistance to cloning. Therefore, $h^*_{m+1} \ge h_{m+1}$.
    \end{proof}

    We now know that $h_m = h^*_m$ for all $m \ge 6$ by combining the last two lemmas. Finally, we show that for $m > 4$, $h_{m} = h^*_{m}$ implies $h_{m-1} = h^*_{m-1}$. For this, take $m > 4$ and consider the profile $P$ defined in \Cref{lem:lem5vd}. Recall that the two axes that do not break the ballots appearing $q$ times are of the following form (up to change of positions among candidates in $D$):
    \begin{align*}
        \underline{c_1}b_1d_1\dots d_{m-3}\underline{c_2} \\
        b_1\underline{c_1}d_1\dots d_{m-3}\underline{c_2}
    \end{align*}

    The cost of these axes are respectively $h^*_m$ and $h_m$ so they both are optimal since $h_m = h^*_m$. If we remove the clone $d_{m-3}$, by resistance to cloning the following two axes should be selected by the rule:
    \begin{align*}
        \underline{c_1}b_1d_1\dots d_{m-4}\underline{c_2} \\
        b_1\underline{c_1}d_1\dots d_{m-4}\underline{c_2}
    \end{align*}
    The cost of these two axes are respectively $h^*_{m-1}$ and $h_{m-1}$. Thus, $h^*_{m-1} = h_{m-1}$.

    This implies that for all $m \ge 4$, the cost is $0$ if the ballot is an interval of the axis, and $h_m$ otherwise. Without loss of generality, we can take $h_m = 1$. For $m=3$, the only approval vector that is induced by non interval ballots is $(1,0,1)$, and we can assume without loss of generality that its cost is $1$. Thus, $f$ is equal to VD.
    \let\oldqed\qedsymbol
    \renewcommand{\qedsymbol}{}
  	 \qedhere (Proof of \Cref{thm:vdcharac}) \oldqed
\end{proof}

Regarding the independence of the conditions of the characterization among neutral scoring rules, note that the trivial rule TRIV returning all axes satisfies every axiom but  consistency with linearity, the \emph{genus rule}\footnote{This is the scoring rule with $\cost_{G}(A,\axis) = |\{(x,y) \in A: \exists z, x \axis z \axis y \text{ and } \forall z \text{ s.t. } x \axis z \axis y, z \notin A\}|$.} minimizing the total number of contiguous holes only fails ballot monotonicity, and the BC rule only violates resistance to cloning. We do not have an example showing that neutrality is necessary, but this axiom can be dropped if we allow an infinitely large ground set of candidates, because then resistance to cloning and consistency with linearity imply neutrality for scoring rules using standard arguments \citep[Lemma~1]{brandl2016consistent}.
     
\subsection{Supplementary Result: Resistance to cloning implies almost topological.}
\label{sec:appTopological}
In this section, we investigate the class of rules satisfying resistance to cloning.
A scoring rule $f$ belongs to the class of \emph{topological rules} if there is a monotone function $h$ such that $\cost_f(A,\axis) = h(k)$ for all $A$ and $\axis$, where $k$ is the number of contiguous holes that $A$ creates in $\axis$.

The following axiom of inclusion clearance is a (very mild) counterpart to clearance: While the latter demands that only axes can be chosen in which unapproved candidates are not interfering, the following axiom demands that many such axes must be included in the choice set. In contrast to clearance, inclusion clearance is satisfied by all five introduced rules in this paper.

\property{Inclusion Clearance}{
We say that a rule $f$ satisfies \emph{inclusion clearance} if the following holds:
let $X$ be the set of candidates that are never approved in $P$. Then there is ${\axis} \in f(P)$ such that there is no $A \in P$ with $y,z \in A$, $x\in X$ and $y \axis x \axis z$.
Further, all other $\axis'$ that have $X$ on the extremes and coincide with $\axis$ on $C\setminus X$ are chosen too.
}

\begin{theorem}\label{thm:almostTopological}
    Let $f$ be a neutral scoring rule that is consistent with linearity. If $f$ satisfies cloning consistency and inclusion clearance, then there are functions $h^*$ and $h$ such that for all $A$ and $\axis$,
    \[
    \cost(A,\axis) = \begin{cases}
    	h^*(n) & \text{if $A$ contains both extremes of $\axis$,} \\
    	h(n) & \text{otherwise,}
    \end{cases}
    \]
    where $n$ is the number of contiguous holes $A$ creates in $\axis$. Further, $h^*(n)\leq h(n)\leq h^*(n+1)$ for all $n$.
\end{theorem}
Since many arguments remain similar to the ones in \Cref{sec:appCharacterization}, we only provide an outline of the proof.
\begin{proof}
    The steps are as follows:
    \begin{enumerate}
        \item There is a function $h$ such that $\cost(A,\axis) = h(n,x_1,\dots,x_{n+1} ,y_1,\dots y_n, i, m)$ (or $\cost(A,\axis)= 0$).
        \begin{itemize}
            \item where $n$ is the number of holes $m$ is the number of candidates present in the axis $\axis$, $x_i$ is the cardinality of the $i-th$ approved interval, and $y_i$ is the cardinality of the $i-th$ hole.
        \end{itemize}
        \item There is a function $h$ such that $\cost(A,\axis) = h(n,x,y_1,\dots y_n,m)$
        \begin{itemize}
            \item where $x =|A|= x_1+\dots+x_{n+1}$
        \end{itemize}
        \item There is $h$ such that $\cost(A,\axis) = h(n,x+y, m)$
        \begin{itemize}
            \item where $y$ = $y_1+ \dots y_n$ is the number of interfering candidates.    
        \end{itemize}
        \item There is a function $h$ such that $\cost(A,\axis) = h(n,i,m)$
        \begin{itemize}
            \item where $i=1$ if $A$ contains both extremes of $\axis$ and $i=0$ else.
        \end{itemize}
        \item There is a function $h$ such that $\cost(A,\axis) = h(n,i)$.
        \item For all $n$, we have $h(n,1) \leq h(n,0) \leq h(n+1,1)$.
    \end{enumerate}
    
    Step 1 follows from inclusion clearance and neutrality.
    
    Step 2 follows from two lemmas that work similarly to the characterization of VD. The first shows that $h(n,x_1,\dots, x_m,y_1,\dots y_n,m) = h(n,x_1\pm 1,\dots, x_m\mp 1,y_1,\dots y_n,m)$, while the second shows that we can invert the first $r$ approved intervals, i.e., $h(n,x_1,\dots, x_m,y_1,\dots y_n,m)= h(n,x_r,\dots,x_1,\dots, x_m,y_{r-1},\dots, y_1,y_r ,\dots y_n,m)$.

    Step 3 follows from a lemma which works similarly to the characterization of VD. There, we can flip $x_1$ and $y_1$ and still obtain the same cost. Thus, further combined with the previous two lemmas, we obtain $h(n,x,y_1,\dots y_n,m) = h(n,x+y-n,n,m)$.

    Step 4 works again exactly as in the characterization of VD, taking one hole of size $2$ and one of size $1$.

    Step 5 This uses a new construction. First, normalize $h(1,0,m)=1$ for all $m$.
    Then show (by induction) that $h(n+1,0,m) - h(1,0,m)= h(n+1,0,m+1) - h(1,0,m+1)$ for all $n,m$.
    For this, consider two axes differing in a single swap, $12 \dots m$, $21\dots m$. Take the ballot $\{1,3,5,\dots\}$. Then, assume that the differences are not equal, use this to create a profile where one of these two axes is chosen but not the other and thus cloning consistency is violated.

    For Step 6, the first inequality is obtained exactly as in the VD characterization and the second inequality follows from weak clearance.
\end{proof}

We can further restrict the class of scoring rules to the class of \emph{local} scoring rules fo which the cost does not depend on non-interfering candidates. Formally, a neutral scoring rule is local if $g(x_{A,\axis})$ = $g(x')$ for all $A,\axis$, where $x'$ is the subvector of $x_{A,\axis}$ where the non-interfering $0$'s are cut off.
Note that as long as we use the same rule for all feasible sets, all five introduced rules satisfy locality.

Clearly, among local scoring rules, \Cref{thm:almostTopological} turns into a characterization. We leave it open whether locality is required.

\newpage
\section{Details of the Experiments}

\subsection{Implementation}

\label{sec:appExpe}
In this section, we explain the methods we used for implementing the rules. We focus here on explaining our approach to reduce the runtime.
First, we present how we improved the brute-force method to be usable in all our experiments. Then, we explain the implementation of the Integer Linear Programming (ILP) encodings which we used for two rules: Voter Deletion and Ballot Completion.

\subsubsection{Brute-force Method} \label{sec:appbruteforce}

The brute-force method is straightforward: compute the cost of all the axes for the given profile, and return the ones with minimal cost. However, this approach takes time exponential in $m$, and hence is not usable in practice even for relatively small values of $m$. Thus, we used pruning methods and heuristics.

We start by pre-processing the approval profile. We assume that each voter has a weight $w_i$ (usually the weights are all initially equal to $1$). Then, we aggregate the weights of all the  voters whose approval ballots are identical. For instance, if two voters $i$ and $j$ have the same ballot $A_i = A_j$, we replace them by a unique voter having the ballot $A_i = A_j$ and the weight $w_i+w_j$. 
Moreover, we remove all ballots that are intervals of any axis, and hence do not help to identify the axes with minimal cost: namely, these are empty ballots, singletons, and full ballots ($A = \mathcal C$). 

We keep a variable containing the lowest cost found \emph{so far}, as well as a variable containing all axes with this cost. Every time we compute the cost of an axis for a given profile by adding up the costs of the ballots, the sum might surpass this value before we read the whole profile. We can then move to the next axis. To save as much running time as possible, we order the ballots by decreasing weights so that we start by the ballots of highest weight. 

Another similar method that reduces the running time is the following: for any axis on $m$ candidates, we can compute a lower bound of its cost by removing two candidates from the profile and computing the cost of the reduced axis on the $m-2$ remaining candidates. This works because one can check that for each of our rules, the cost weakly decreases as candidates are removed.
In our implementation, we group axes into sets of $(m-1)(m-2)/2$ axes (one for each position of the missing pair of candidates) and if we observe that the cost of their common reduced axis is higher than the current lowest cost, this means that no axis of this set will be optimal and that we can completely skip all of them.

Finally, we can initialize the current lowest cost axis with an axis expected to be good, such as one obtained by a greedy algorithm. For political datasets, we can use the axes adopted by the media. 

Combining all these strategies, we have never needed more than one hour to find optimal axes for profiles on up to 12 candidates. For less than 7 candidates, the result was always returned in less than one second. Our code is available at \url{https://github.com/TheoDlmz/AxisRules}.

\subsubsection{ILP Encoding} \label{sec:applinearsolver}

We also implemented ILP formulations of the Voter Deletion and Ballot Completion rules. In this section we briefly describe them.

First, we conduct the same pre-processing on the approval profile by merging the weights of identical ballots. Then, we create a binary variable $x_{a,b}$ such that $x_{a,b} = 1$ if and only if $a \axis b$ on the axis. Then, the formulation is different for VD and BC:
\begin{itemize}
    \item \textbf{VD:} For each voter $i$, we add a binary variable $y_i$ such that $y_i = 1$ if and only if the ballot $A_i$ is an interval of the axis. Then, for each voter $i$, each pair $(a,c) \in (A_i)^2$ of approved candidates, and each disapproved candidate $b$, we add the constraint
    \[ x_{a,b} + x_{b, c} \le 2 - y_i. \]
    Finally, the cost is equal to the sum of the $w_i \cdot (1-y_i)$ where $w_i$ is the weight of ballot $i$.
    \item \textbf{BC:} For each candidate $a \in C$, we introduce an integer variable $p_a \in [0,m-1]$ that encodes the position of the candidate on the axis ($p_a = \sum_b x_{b,a}$). Then, for each ballot $A_i$, we define two variables $M_i$ and $m_i$ respectively for the right-most and left-most position of candidates approved in $A_i$, which can be derived as the maximum and minimum values of the position variables. (This can be encoded using standard techniques or via Gurobi's ``general constraints''.) The BC cost of the ballot $A_i$ is then given by $M_i-m_i-|A_i|+1$. Finally, we sum this cost over all ballots (multiplied by the weights $w_i$) to obtain the overall cost.
\end{itemize}

For the precise implementation, refer to our code at \url{https://github.com/TheoDlmz/AxisRules}.

\subsection{Synthetic Data}
\label{sec:appSynthetic}

In this section, we present our experiments and results on synthetic data models.
As mentioned in \Cref{sec:preliminaries}, any scoring rule can be interpreted as the Maximum Likelihood Estimator of some appropriate noise model.
In each of our models, conditioned on a given \emph{ground truth axis} (which we will draw uniformly at random), the approval ballots will be sampled i.i.d.
The performance of the rules are likely to reflect simply how similar they are to the MLE of the models used. However, these experiments can give an idea of how well the rules can generalize to different models. In this section, we study four models, each inspired by one of our rules. Only the Maverick Voters model actually corresponds to the model of which VD is the MLE. For the other rules (MF, BC, and MS), the precise model for which they are MLEs are less natural than the intuitively similar models that we use here.

\begin{itemize}
    \item \textbf{Maverick Voters:} In this model, for each voter, we randomly decide if they are a ``maverick voter''. We throw a coin, and with probability $p \in [0,\frac12)$, they are a maverick and we sample an approval ballot at random (whether or not it is an interval of the axis). Otherwise we sample an approval ballot that is an interval of $\axis$ uniformly at random.
    \item \textbf{Random Flips:} In this model, we first sample for each voter an approval ballot that is an interval of the axis $\axis$ uniformly at random (among all interval ballots). Then for each candidate, we switch its status (from approved to non approved, or conversely) independently with probability $p \in  [0,\frac12)$ .
    \item \textbf{Random Omissions:} In this model,  we first sample for each voter an approval ballot that is an interval of the axis $\axis$ uniformly at random (among all interval ballots). Then for each \emph{approved} candidate, we switch its status (from approved to non approved) independently with probability $p \in  [0,\frac12)$.
    \item \textbf{Random Swaps:} In this model, for each voter, we sample an axis $\axis'$ using the Mallows model with center $\axis$ and dispersion parameter $\phi \in [0,1]$. As a reminder, the probability of $\axis'$ in this model is proportional $\phi^{KT(\axis,\axis')}$ where $KT$ is the Kendall-tau distance. Once $\axis'$ is sampled, we sample uniformly at random an approval ballot which is an interval of $\axis'$.
\end{itemize}

We do not have a noise model that corresponds intuitively to the FT rule.

For a given model and a given rule, we sample a profile according to the model and we compute the Kendall-tau distance between the axis returned by the rule and the ground truth axis. In case of a tie, we take the average KT over all returned axes. For all our experiments, we set $m=7$ candidates (for bigger $m$ the computation takes too long), $n=100$ voters, and we average over $1\,000$ random profiles.
\begin{figure}[!t]
    \centering
    \begin{subfigure}{.49\textwidth}
        \centering
        \includegraphics[width=1\textwidth]{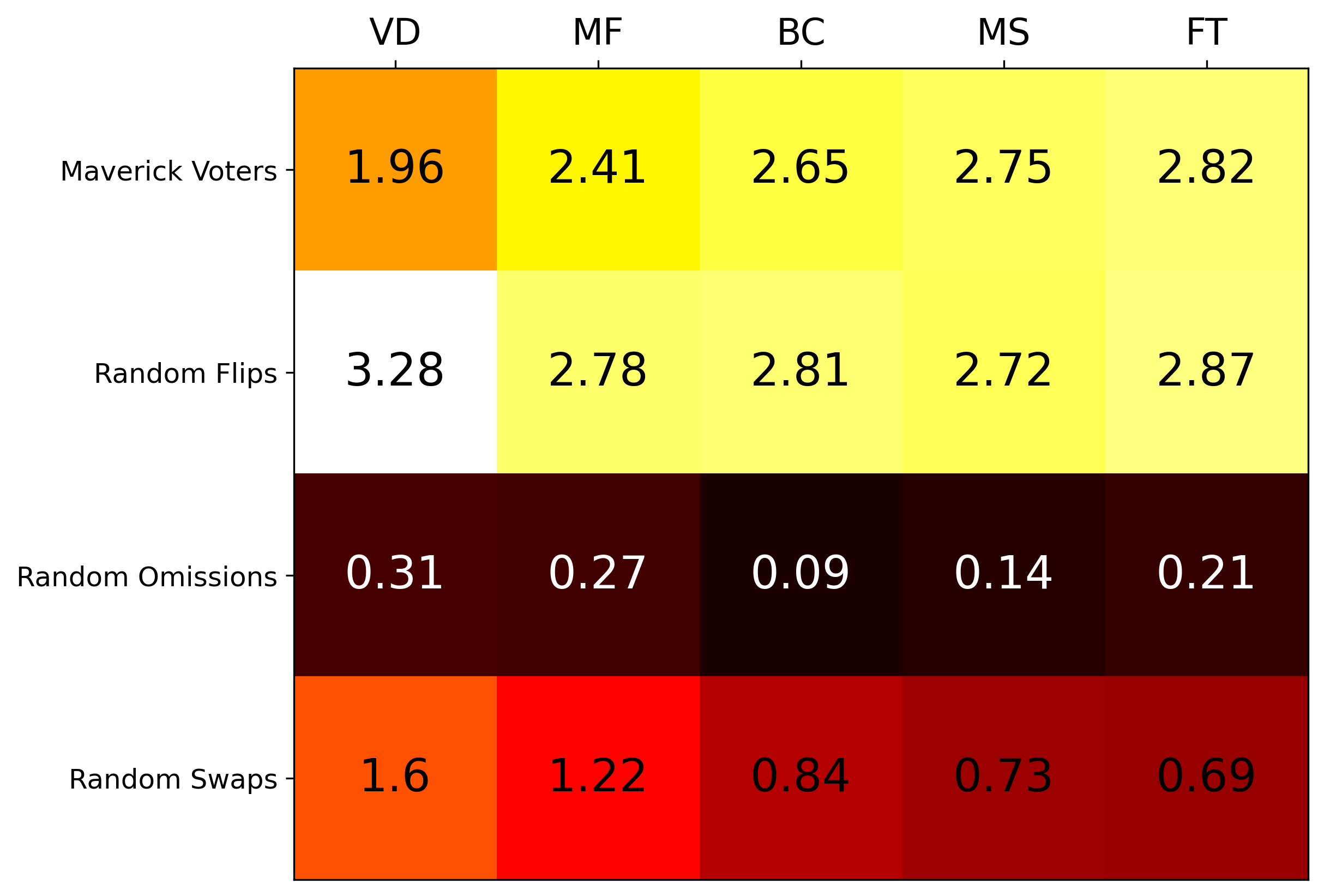}
        \caption{Not normalized.}
        \label{fig:appModels1}
    \end{subfigure}
    \begin{subfigure}{.49\textwidth}
        \centering
        \includegraphics[width=1\textwidth]{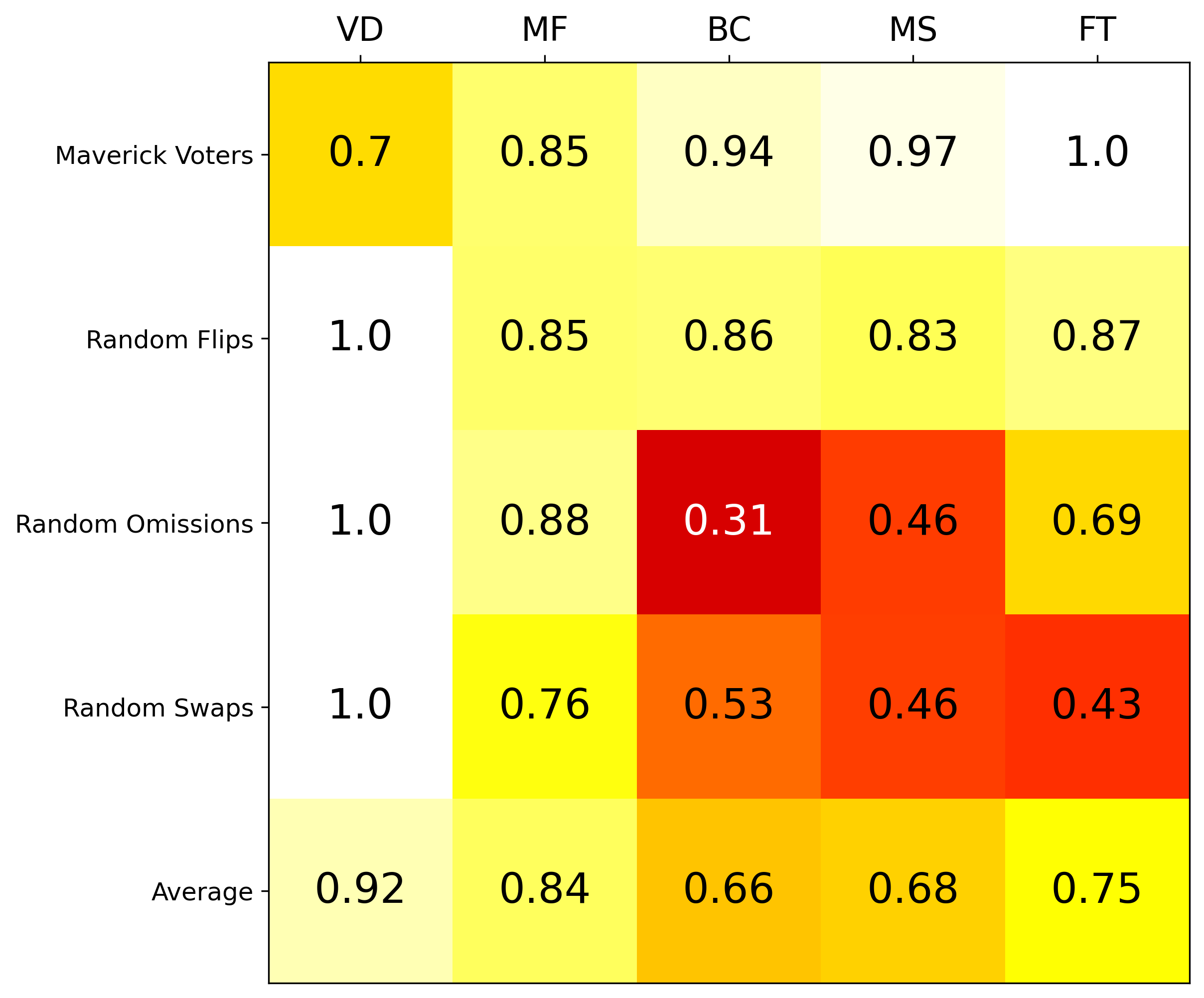}
        \caption{Normalized. }
        \label{fig:appModels2}
    \end{subfigure}
    \caption{Average Kendall-tau distance to the ground truth axes for different rules and models, averaged over 1\,000 profiles. The row labeled ``average'' in figure (b) shows the average KT distance induced by each rule across all 4 models.}
\end{figure}

We used the following parameters for the models: $p=0.2$ for Maverick Voters, $p=0.3$ for Random Flips, $p=0.45$ for Random Omissions, and $\phi=0.5$ for Random Swaps. These parameters were chosen to give roughly similar KT distances across models. \Cref{fig:appModels1} presents those KT distances. We normalized the distances such that the maximum is $1$ for each model. We also looked at the average performance of each rule across all models. The results are displayed in \Cref{fig:appModels2}. The main conclusion seems to be that no rule really generalizes to all models, but VD is particularly bad at generalizing beyond the Maverick Voters model.

\subsection{The French Presidential Election}
\label{sec:appFrench}

In this section, we present the results for the French presidential elections datasets. These datasets were gathered in parallel to the actual presidential elections of 2017 and 2022 and were part of the \emph{``Voter Autrement''} project.%
\footnote{See \url{https://www.gate.cnrs.fr/vote/}.}
During one month, anyone could answer an online survey, which was promoted on social networks and mailing lists. Participants were asked what would have been their vote for various voting methods, such as approval voting, score voting, Borda, instant runoff voting and the majority judgement rule. 

In our experiments, we only need approvals and rankings. 
For the approval preferences, no preprocessing is required, as we can simply use the approval votes of participants. 
For the ranking preferences, it is more complicated. Indeed, participants were allowed to rank only a subset of the candidates, for instance their four most favorite ones. However, in our experiments we need full rankings, so we removed all incomplete rankings from the datasets. Voters submitted these rankings in the context of instant runoff voting.

We added weights to the voters so that the sample is more representative. Indeed, the set of participants is heavily biased towards the left. Luckily, we know for each participant for which candidate they actually voted at the election (if they agreed to answer this question). Thus, we can adapt the weights of the participants based on this information, so that the distribution of opinions reflects the actual election result. For instance, participants who voted for the main candidate from the left are over-represented, so they get a weight smaller than 1, while participants who voted for the far-right candidates are under-represented, so they get a larger weight. Note that, obviously, this does not completely eliminate the bias.

As a benchmark, we used axes developed by the main polling institutes operating in France. They use these axes (1) when asking the participants which candidate they support and (2) when they present the results. We collected these axes from documents published online by the institutes.%
\footnote{See \url{https://en.wikipedia.org/w/index.php?title=Opinion_polling_for_the_2022_French_legislative_election&oldid=1207035923} and \url{https://en.wikipedia.org/w/index.php?title=Opinion_polling_for_the_2017_French_legislative_election&oldid=1059418478} for references.}
\Cref{tab:appFrenchPolls2017} in the main body (for 2017) and \Cref{tab:appFrenchPolls2022} (for 2022) show these axes. Note that the axes differ by polling institute. The main differences are (i) the positions of the ``small'' candidates, as these are hard to place since they often have no obvious classification as left-wing nor right-wing, and (ii) the positions of candidates inside an ideological subgroup (e.g., the far-left candidates or the far-right candidates).

\begin{table*}[t]
    \centering
    \makebox[\textwidth][c]{ %
    \scalebox{0.9}{\begin{tabular}{cc c c c c c c c c c c ccc} \toprule
        Rule & & & & & &$\axis$ & & & & & &  & \clap{Min KT}\qquad &Avg KT\\ \midrule
       VD  & $\party{PCF}$ &  $\party{LO}$ & $\party{NPA}$ & $\party{LFI}$ & $\party{EELV}$ & $\party{PS}$ & $\party{EM}$ & $\party{LR}$ & $\party{DLF}$ & $\party{REC}$ & $\party{RN}$ & $\party{R}$ & 4&5.62  \\
       MF  &   $\party{LO}$ & $\party{NPA}$ & $\party{LFI}$ &$\party{PCF}$ & $\party{PS}$ & $\party{EELV}$ & $\party{EM}$ & $\party{LR}$ & $\party{R}$ &  $\party{RN}$ &$\party{REC}$ &$\party{DLF}$ &   4&5.38  \\
       BC  &   $\party{LO}$ & $\party{NPA}$ &$\party{PCF}$ &  $\party{LFI}$ & $\party{EELV}$ &$\party{PS}$ & $\party{EM}$ & $\party{LR}$ & $\party{R}$ &  $\party{RN}$ &$\party{REC}$ &$\party{DLF}$ &   3&5.12  \\
       MS  &   $\party{LO}$ & $\party{NPA}$ &$\party{PCF}$ &  $\party{LFI}$ & $\party{PS}$ & $\party{EELV}$ &$\party{EM}$ & $\party{LR}$ & $\party{R}$ &  $\party{RN}$ &$\party{REC}$ &$\party{DLF}$ &   3&4.88 \\
       FT  &   $\party{LO}$ & $\party{NPA}$ & $\party{LFI}$ &$\party{PCF}$ &  $\party{PS}$ & $\party{EELV}$ &$\party{EM}$ & $\party{LR}$ & $\party{R}$ &  $\party{RN}$ &$\party{REC}$ &$\party{DLF}$ &  4&5.38  \\\midrule 
       
       VD-rank  &  $\party{DLF}$ & $\party{R}$ &$\party{PCF}$&   $\party{LO}$ & $\party{NPA}$ & $\party{LFI}$  & $\party{EELV}$ & $\party{PS}$ & $\party{EM}$ & $\party{LR}$ &  $\party{RN}$ &$\party{REC}$ & 18 & 20.62  \\
       FT-rank  &   $\party{LO}$ & $\party{NPA}$&$\party{PCF}$ & $\party{LFI}$ &  $\party{PS}$ & $\party{EELV}$ &$\party{EM}$ & $\party{LR}$ & $\party{R}$ &  $\party{RN}$  &$\party{DLF}$ &$\party{REC}$&   \textbf{2} & \textbf{3.88}      
       \\\bottomrule
    \end{tabular}}}
    \caption{Optimal axis of each rule for the 2022 French presidential election}
    \label{tab:appFrenchRes2022}
\end{table*}

\begin{table*}[t]
    \centering
    \scalebox{0.9}{\begin{tabular}{cc c c c c c c c c c c c} \toprule
        Institute & & & & & &$\axis$ & & & & & &  \\ \midrule
          BVA  &   $\party{LO}$ & $\party{NPA}$&$\party{LFI}$&$\party{PCF}$ &   $\party{PS}$& $\party{EELV}$  & $\party{EM}$& $\party{LR}$ &$\party{DLF}$ &  $\party{REC}$&$\party{RN}$  & $\party{R}$  \\ 
          Opinionway  &   $\party{LO}$ & $\party{NPA}$&$\party{PCF}$ & $\party{LFI}$&  $\party{PS}$& $\party{EELV}$  & $\party{EM}$& $\party{LR}$& $\party{R}$  &$\party{DLF}$ & $\party{REC}$ &$\party{RN}$  \\ 
          IFOP  &   $\party{LO}$ & $\party{NPA}$&$\party{PCF}$ & $\party{LFI}$&  $\party{PS}$& $\party{EELV}$  & $\party{EM}$& $\party{LR}$ &$\party{DLF}$ &  $\party{RN}$  &$\party{REC}$& $\party{R}$  \\ 
          IPSOS  &   $\party{NPA}$ & $\party{LO}$ & $\party{LFI}$ &$\party{PCF}$&$\party{EELV}$ &  $\party{PS}$ & $\party{EM}$ & $\party{LR}$ & $\party{R}$ &  $\party{RN}$ &$\party{DLF}$ &$\party{REC}$ \\ 
          Harris Interactive  &   $\party{LO}$ & $\party{NPA}$&$\party{PCF}$ & $\party{LFI}$&  $\party{PS}$& $\party{EELV}$  & $\party{EM}$& $\party{LR}$ &$\party{DLF}$ &  $\party{RN}$  &$\party{REC}$& $\party{R}$  \\ 
          Cluster17  &   $\party{LO}$ & $\party{NPA}$&$\party{PCF}$ & $\party{LFI}$& $\party{EELV}$ &  $\party{PS}$ & $\party{EM}$& $\party{R}$ & $\party{LR}$ &$\party{DLF}$ &  $\party{RN}$  &$\party{REC}$ \\ 
          Odoxa  &   $\party{LO}$ & $\party{NPA}$&$\party{PCF}$ & $\party{LFI}$& $\party{EELV}$ &  $\party{PS}$ & $\party{EM}$& $\party{R}$ & $\party{LR}$ &$\party{DLF}$ &  $\party{RN}$  &$\party{REC}$ \\ 
          Elabe  &   $\party{NPA}$&$\party{LO}$ & $\party{PCF}$ & $\party{LFI}$&  $\party{PS}$& $\party{EELV}$  & $\party{EM}$& $\party{LR}$ &$\party{DLF}$ &  $\party{RN}$  &$\party{REC}$& $\party{R}$  \\ \bottomrule
    \end{tabular}}
    \caption{Axes used by polling institutes for the 2022 French presidential election}
    \label{tab:appFrenchPolls2022}
\end{table*}

\Cref{tab:appFrenchRes2017} in the main body (for 2017) and \Cref{tab:appFrenchRes2022} (for 2022) show the axes returned by each of our rules (including ranking rules), their minimal KT distance to the axes of polling institutes (i.e., the distance to the closest of those axes), and their average KT distance to polling institutes. Note that the axes show the parties of the candidates (not the candidate names), and for the colors we followed the choices made by editors of \emph{Wikipedia}.%
\footnote{\href{https://fr.wikipedia.org/wiki/Modèle:Infobox_Parti_politique_français/couleurs}{https://fr.wikipedia.org/wiki/Mod\`ele:Infobox\_Parti\_politique\_\allowbreak fran\c{c}ais/couleurs}}

The axes returned by the different rules are very similar, and they are also close to the axes used by the institutes (except for the VD-rank rule). The differences mainly concern the positions of the less popular candidates (e.g., $\party{R}$) and the positions of the candidates inside each ideological subgroup (e.g., between $\party{PS}$ and $\party{EELV}$ for the 2022 election).

\subsection{Supreme Court of the United States}
\label{sec:appScotus}

We derived this dataset from the Supreme Court Database (\url{http://scdb.wustl.edu/}), which contains data for Supreme Court decisions starting in 1946. The Court consists of 9 justices who \emph{vote} on each case about which of the two parties to the case wins. The Court then publishes a \emph{majority opinion} explaining the Court's reasoning. Justices can also submit \emph{concurring opinions} and \emph{dissenting opinions}, and \emph{join} any of the opinions submitted by others. Concurring opinions explain additional or alternative reasons, written by justices who voted with the majority. Dissenting opinions explain why a justice did not vote with the majority.

The Martin-Quinn method for deriving an axis of justices uses only the binary vote data (i.e., whether a justice voted for or against the winning party), and its underlying model assumes that a decision divides the axis of justices in the middle, with all justices to one side of the cutoff voting the same way. One issue with this approach is that justices may vote for the same party but have different reasons for it. It could be for example that the most progressive and most conservative justices vote the same way, while the centrist justices vote the other way, for example due to procedural reasons. This is not well-captured by the model. In addition, the model does not use some relevant information. For example, if two justices very frequently join each other in their concurring or dissenting opinions, this suggests that these justices should be placed near each other on the axis.

In our experiments, we discarded all terms with more than 9 justices (e.g., if one is replaced mid-term), giving us 65 terms and thus 65 profiles of approval ballots. We compared our rules to the axes obtained by the established Martin-Quinn method, by computing the KT distance between the axes.

Figures \ref{fig:SCOTUS_MQ} to \ref{fig:SCOTUS_FT} show the evolution of the positions of the justices on the axes for the last 20 terms, according to the axes produced by the Martin-Quinn method and by our rules. It is very clear that the Martin-Quinn method is smoother over time, which is by the rule's design, since it takes the justice positions of the last term as a prior for their positions in the next term. Our rules are less stable. 

\begin{figure*}
    \centering
    \includegraphics[width=0.8\textwidth]{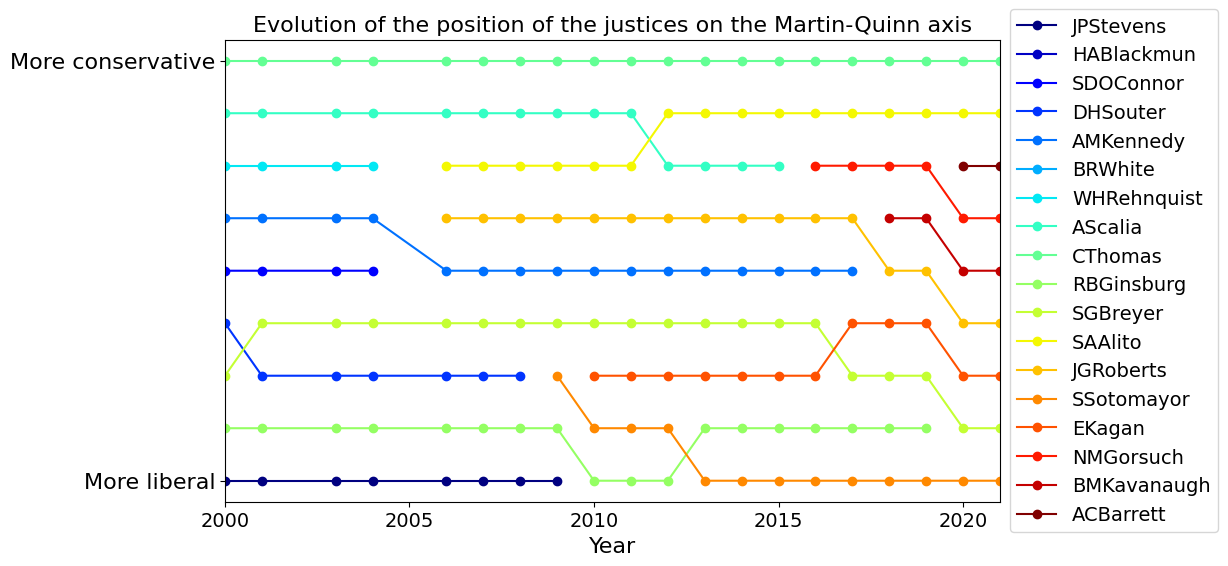}
    \caption{Positions of the justices for terms between 2000 and 2021 for the MQ method.}
    \label{fig:SCOTUS_MQ}
\end{figure*}

\begin{figure*}
    \centering
    \includegraphics[width=0.8\textwidth]{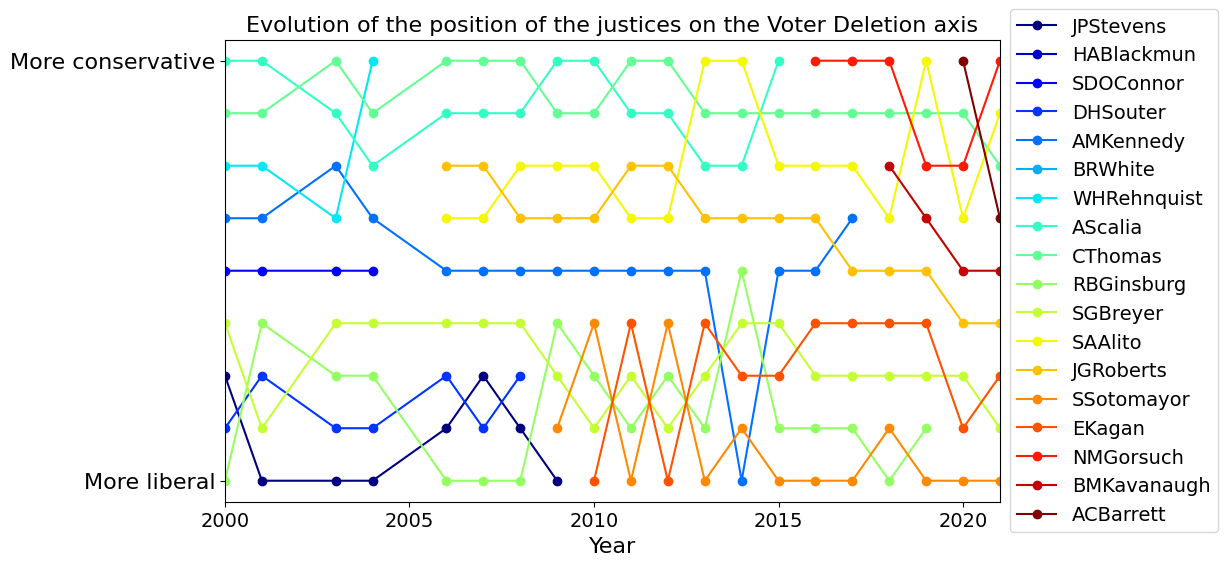}
    \caption{Positions of the justices for terms between 2000 and 2021 for the VD rule.}
    \label{fig:SCOTUS_VD}
\end{figure*}

\begin{figure*}
    \centering
    \includegraphics[width=0.8\textwidth]{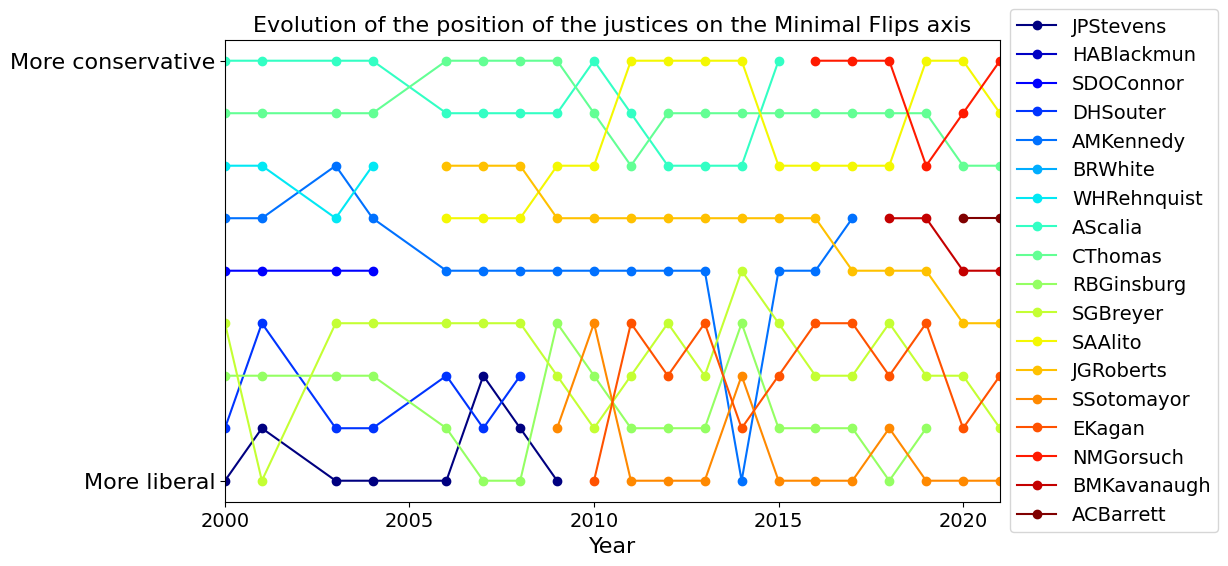}
    \caption{Positions of the justices for terms between 2000 and 2021 for the MF rule.}
    \label{fig:SCOTUS_MF}
\end{figure*}

\begin{figure*}
    \centering
    \includegraphics[width=0.8\textwidth]{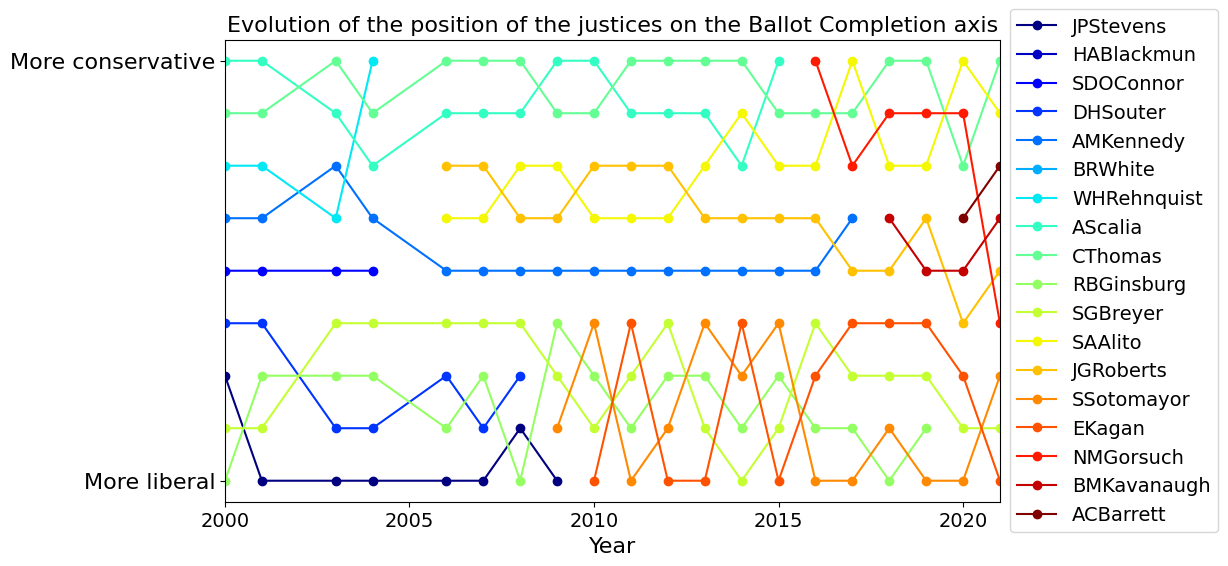}
    \caption{Positions of the justices for terms between 2000 and 2021 for the BC rule.}
    \label{fig:SCOTUS_BC}
\end{figure*}

\begin{figure*}
    \centering
    \includegraphics[width=0.8\textwidth]{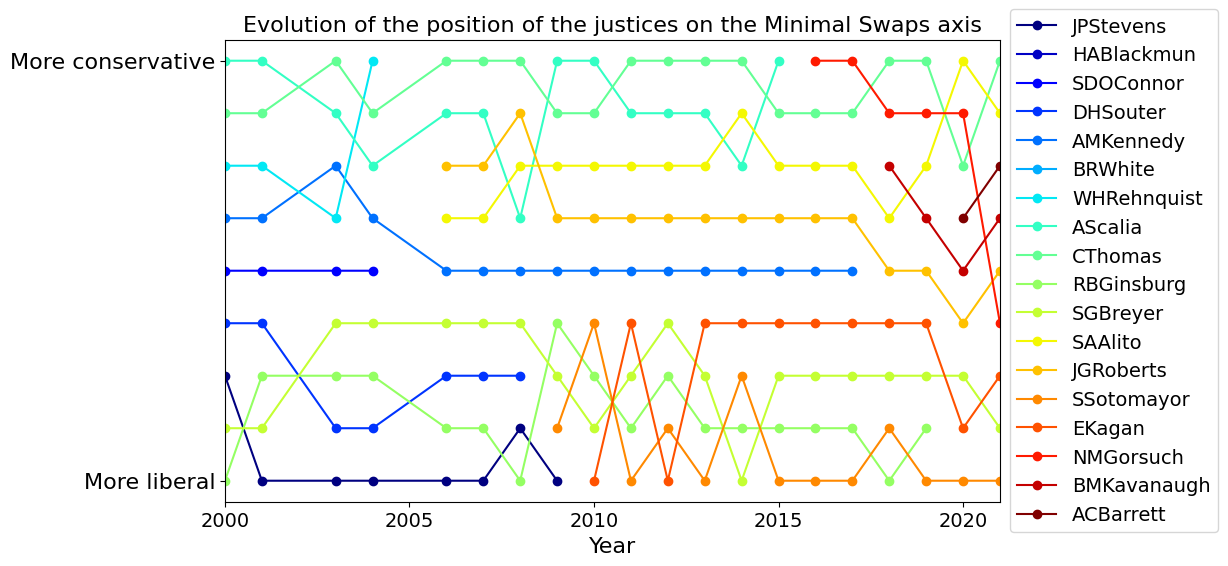}
    \caption{Positions of the justices for terms between 2000 and 2021 for the MS rule.}
    \label{fig:SCOTUS_MS}
\end{figure*}

\begin{figure*}
    \centering
    \includegraphics[width=0.8\textwidth]{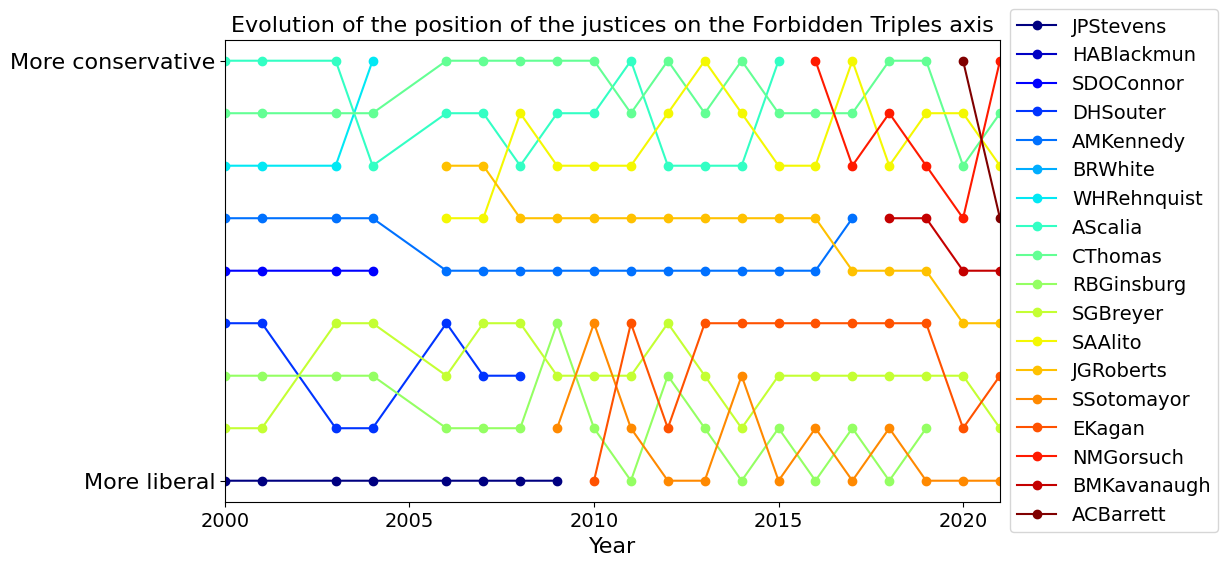}
    \caption{Positions of the justices for terms between 2000 and 2021 for the FT rule.}
    \label{fig:SCOTUS_FT}
\end{figure*}

\end{document}